\documentclass[a4paper,onecolumn,11pt,unpublished,amsmath,amssymb]{quantumarticle}
\pdfoutput=1

\usepackage[utf8]{inputenc}
\usepackage[english]{babel}
\usepackage[T1]{fontenc}
\usepackage{microtype}
\usepackage{graphicx}
\usepackage{amsthm}
\usepackage{thmtools}   %
\usepackage{mathtools}
\usepackage{xspace}
\usepackage{array}
\usepackage{braket}
\usepackage{booktabs}
\usepackage{placeins}
\usepackage{enumitem}
\usepackage{bm}
\usepackage{hyperref}

\usepackage{algorithm}
\usepackage{algpseudocode}
\usepackage{quantikz}
\usepackage{subcaption}
\usepackage[bold,basic]{complexity}
\usepackage[numbers]{natbib}
\usetikzlibrary{calc}

\setlist{nosep}

\hypersetup{colorlinks=true,linkcolor=blue!50!black,citecolor=blue!50!black,urlcolor=blue!50!black}
\AtBeginDocument{%
}

\newcommand{\appref}[1]{Appendix~\ref{#1}}

\newcommand{\F}{\mathbb F}
\newcommand{\Z}{\mathbb Z}
\renewcommand{\C}{\mathbb C}
\newcommand{\rank}{\operatorname{rank}}

\renewcommand{\cc}{\operatorname{cc}}

\newcommand{\bx}{\bm{x}}
\newcommand{\by}{\bm{y}}
\newcommand{\bz}{\bm{z}}
\newcommand{\bu}{\bm{u}}
\newcommand{\bw}{\bm{w}}
\newcommand{\bs}{\bm{s}}

\newcommand{\blambda}{\bm{\lambda}}
\newcommand{\ba}{\bm{a}}
\newcommand{\bb}{\bm{b}}
\newcommand{\bA}{\bm{A}}
\newcommand{\bzero}{\bm{0}}
\newcommand{\bsigma}{\bm{\sigma}}
\newcommand{\balpha}{\bm{\alpha}}
\newcommand{\bbeta}{\bm{\beta}}
\newcommand{\bgamma}{\bm{\gamma}}

\theoremstyle{plain}   %
\newtheorem{theorem}{Theorem}
\newtheorem{lemma}{Lemma}
\newtheorem{corollary}{Corollary}
\newtheorem{proposition}{Proposition}
\newtheorem{definition}{Definition}

\theoremstyle{remark}
\newtheorem{remark}{Remark}

\theoremstyle{definition}
\newtheorem{example}{Example}

\newcommand{\tw}{\operatorname{tw}}
\newcommand{\rw}{\operatorname{rw}}
\newcommand{\lrw}{\operatorname{lrw}}

\newcommand{\Dcycl}{\mathbb D}
\newcommand{\mmomega}{\omega_{\mathrm{MM}}}
\newcolumntype{L}[1]{>{\raggedright\arraybackslash}p{#1}}
\newcommand{\SopCount}{\textbf{SopCount}\xspace}
\newcommand{\CZ}{\mathrm{CZ}}
\newcommand{\CCZ}{\mathrm{CCZ}}
\newcommand{\D}{\mathrm{D}}
\renewcommand\H{\mathrm{H}}
\newcommand\CX{\mathrm{CX}}
\newcommand{\T}{\mathrm{T}}
\newcommand{\Gate}{\mathcal G}

\newcommand\doublecheck{{\checkmark\kern-1.02em\checkmark}}

\newcommand{\leanroot}{https://github.com/gaperez64/dlx4sop/blob/main/docs/lean}
\newcommand{\leanbadge}{\textsf{\color{cyan}\footnotesize[Lean\,\doublecheck]}}
\newcommand{\lean}[2]{%
  \ifcsname leanline@#1:#2\endcsname
    \,\href{\leanroot/#1\#L\csname leanline@#1:#2\endcsname}{\leanbadge}%
  \else
    \,\textsf{\color{red}\footnotesize[Lean?]}%
    \GenericWarning{}{Lean anchor not found: #1:#2 (run scripts/lean_anchors.py)}%
  \fi
}
\IfFileExists{lean-anchors.tex}{
\expandafter\def\csname leanline@Formal/Core/Cost.lean:SopInstance.costFull\endcsname{108}
\expandafter\def\csname leanline@Formal/Core/Cost.lean:SopInstance.costMode\endcsname{114}
\expandafter\def\csname leanline@Formal/Core/Cost.lean:Subtree\endcsname{26}
\expandafter\def\csname leanline@Formal/Core/Cost.lean:WidthBounded\endcsname{86}
\expandafter\def\csname leanline@Formal/Core/Cost.lean:costFull_le\endcsname{188}
\expandafter\def\csname leanline@Formal/Core/Cost.lean:costFull_le_sharp\endcsname{136}
\expandafter\def\csname leanline@Formal/Core/Cost.lean:costFull_leaf\endcsname{122}
\expandafter\def\csname leanline@Formal/Core/Cost.lean:costFull_node\endcsname{124}
\expandafter\def\csname leanline@Formal/Core/Cost.lean:costMode_le\endcsname{203}
\expandafter\def\csname leanline@Formal/Core/Cost.lean:costMode_le'\endcsname{217}
\expandafter\def\csname leanline@Formal/Core/Cost.lean:costMode_le_sharp\endcsname{162}
\expandafter\def\csname leanline@Formal/Core/Cost.lean:costMode_leaf\endcsname{128}
\expandafter\def\csname leanline@Formal/Core/Cost.lean:costMode_node\endcsname{130}
\expandafter\def\csname leanline@Formal/Core/Cost.lean:leafCount\endcsname{34}
\expandafter\def\csname leanline@Formal/Core/Cost.lean:leafCount_eq\endcsname{54}
\expandafter\def\csname leanline@Formal/Core/Cost.lean:leafCount_le_card\endcsname{90}
\expandafter\def\csname leanline@Formal/Core/Cost.lean:leafCount_le_card_verts\endcsname{62}
\expandafter\def\csname leanline@Formal/Core/Cost.lean:leafCount_leaf\endcsname{43}
\expandafter\def\csname leanline@Formal/Core/Cost.lean:leafCount_node\endcsname{45}
\expandafter\def\csname leanline@Formal/Core/Cost.lean:nodeCount\endcsname{39}
\expandafter\def\csname leanline@Formal/Core/Cost.lean:nodeCount_leaf\endcsname{48}
\expandafter\def\csname leanline@Formal/Core/Cost.lean:nodeCount_node\endcsname{50}
\expandafter\def\csname leanline@Formal/Core/Cost.lean:statesFull\endcsname{78}
\expandafter\def\csname leanline@Formal/Core/Cost.lean:statesFull_le\endcsname{95}
\expandafter\def\csname leanline@Formal/Core/Cost.lean:statesSig\endcsname{82}
\expandafter\def\csname leanline@Formal/Core/Cost.lean:statesSig_le\endcsname{100}
\expandafter\def\csname leanline@Formal/Core/DP.lean:Dspec\endcsname{36}
\expandafter\def\csname leanline@Formal/Core/DP.lean:SopInstance.Dalg\endcsname{43}
\expandafter\def\csname leanline@Formal/Core/DP.lean:combineSig\endcsname{31}
\expandafter\def\csname leanline@Formal/Core/DP.lean:crossTerm\endcsname{27}
\expandafter\def\csname leanline@Formal/Core/DPCorrect.lean:Dalg_eq_Dspec\endcsname{440}
\expandafter\def\csname leanline@Formal/Core/DPCorrect.lean:Dalg_leaf_eq_Dspec\endcsname{50}
\expandafter\def\csname leanline@Formal/Core/DPCorrect.lean:Dalg_node_eq_Dspec\endcsname{385}
\expandafter\def\csname leanline@Formal/Core/DPCorrect.lean:Dspec_as_sum\endcsname{333}
\expandafter\def\csname leanline@Formal/Core/DPCorrect.lean:Dspec_node_as_pair\endcsname{359}
\expandafter\def\csname leanline@Formal/Core/DPCorrect.lean:add_supported_node\endcsname{92}
\expandafter\def\csname leanline@Formal/Core/DPCorrect.lean:crossEdge\endcsname{185}
\expandafter\def\csname leanline@Formal/Core/DPCorrect.lean:dartSum_eq_orderedCross\endcsname{217}
\expandafter\def\csname leanline@Formal/Core/DPCorrect.lean:dp_returns_counts\endcsname{449}
\expandafter\def\csname leanline@Formal/Core/DPCorrect.lean:edgeCross_eq_rawCross\endcsname{256}
\expandafter\def\csname leanline@Formal/Core/DPCorrect.lean:edgeProd_add\endcsname{188}
\expandafter\def\csname leanline@Formal/Core/DPCorrect.lean:edgeSum_eq_dartSum\endcsname{196}
\expandafter\def\csname leanline@Formal/Core/DPCorrect.lean:eta\endcsname{138}
\expandafter\def\csname leanline@Formal/Core/DPCorrect.lean:eta_add\endcsname{140}
\expandafter\def\csname leanline@Formal/Core/DPCorrect.lean:fib\endcsname{340}
\expandafter\def\csname leanline@Formal/Core/DPCorrect.lean:half_selCount\endcsname{177}
\expandafter\def\csname leanline@Formal/Core/DPCorrect.lean:mem_suppFin\endcsname{328}
\expandafter\def\csname leanline@Formal/Core/DPCorrect.lean:orderedCross_eq_rawCross\endcsname{241}
\expandafter\def\csname leanline@Formal/Core/DPCorrect.lean:phi_node_add\endcsname{271}
\expandafter\def\csname leanline@Formal/Core/DPCorrect.lean:restrict\endcsname{79}
\expandafter\def\csname leanline@Formal/Core/DPCorrect.lean:restrict_add_self\endcsname{119}
\expandafter\def\csname leanline@Formal/Core/DPCorrect.lean:restrict_apply_mem\endcsname{88}
\expandafter\def\csname leanline@Formal/Core/DPCorrect.lean:restrict_left_add\endcsname{100}
\expandafter\def\csname leanline@Formal/Core/DPCorrect.lean:restrict_right_add\endcsname{109}
\expandafter\def\csname leanline@Formal/Core/DPCorrect.lean:restrict_supported\endcsname{82}
\expandafter\def\csname leanline@Formal/Core/DPCorrect.lean:selParity\endcsname{163}
\expandafter\def\csname leanline@Formal/Core/DPCorrect.lean:selParity_add\endcsname{262}
\expandafter\def\csname leanline@Formal/Core/DPCorrect.lean:selParity_eq\endcsname{166}
\expandafter\def\csname leanline@Formal/Core/DPCorrect.lean:sig_node_add\endcsname{301}
\expandafter\def\csname leanline@Formal/Core/DPCorrect.lean:suppFin\endcsname{325}
\expandafter\def\csname leanline@Formal/Core/DPCorrect.lean:supported_leaf_iff\endcsname{24}
\expandafter\def\csname leanline@Formal/Core/DPCorrect.lean:zmod2_cases\endcsname{135}
\expandafter\def\csname leanline@Formal/Core/Fourier.lean:Aalg\endcsname{178}
\expandafter\def\csname leanline@Formal/Core/Fourier.lean:Aalg_root\endcsname{183}
\expandafter\def\csname leanline@Formal/Core/Fourier.lean:N_inversion\endcsname{146}
\expandafter\def\csname leanline@Formal/Core/Fourier.lean:Nhat\endcsname{73}
\expandafter\def\csname leanline@Formal/Core/Fourier.lean:Nhat_eq_sum_N\endcsname{76}
\expandafter\def\csname leanline@Formal/Core/Fourier.lean:Nhat_one\endcsname{103}
\expandafter\def\csname leanline@Formal/Core/Fourier.lean:chi\endcsname{31}
\expandafter\def\csname leanline@Formal/Core/Fourier.lean:chi_add\endcsname{59}
\expandafter\def\csname leanline@Formal/Core/Fourier.lean:chi_natCast\endcsname{66}
\expandafter\def\csname leanline@Formal/Core/Fourier.lean:chi_zero\endcsname{63}
\expandafter\def\csname leanline@Formal/Core/Fourier.lean:omega_pow_half\endcsname{47}
\expandafter\def\csname leanline@Formal/Core/Fourier.lean:omega_pow_mod\endcsname{41}
\expandafter\def\csname leanline@Formal/Core/Fourier.lean:omega_pow_r\endcsname{34}
\expandafter\def\csname leanline@Formal/Core/Fourier.lean:single_amplitude\endcsname{198}
\expandafter\def\csname leanline@Formal/Core/Fourier.lean:sum_chi\endcsname{111}
\expandafter\def\csname leanline@Formal/Core/Histogram.lean:Delta\endcsname{67}
\expandafter\def\csname leanline@Formal/Core/Histogram.lean:Mlin\endcsname{77}
\expandafter\def\csname leanline@Formal/Core/Histogram.lean:Mlin_eq_N_add\endcsname{113}
\expandafter\def\csname leanline@Formal/Core/Histogram.lean:N_recovered\endcsname{233}
\expandafter\def\csname leanline@Formal/Core/Histogram.lean:Nlin\endcsname{73}
\expandafter\def\csname leanline@Formal/Core/Histogram.lean:S_eq_sum_Delta\endcsname{139}
\expandafter\def\csname leanline@Formal/Core/Histogram.lean:card_pair\endcsname{97}
\expandafter\def\csname leanline@Formal/Core/Histogram.lean:coords_unique\endcsname{177}
\expandafter\def\csname leanline@Formal/Core/Histogram.lean:f_eq_flin_add\endcsname{80}
\expandafter\def\csname leanline@Formal/Core/Histogram.lean:f_eq_flin_or\endcsname{86}
\expandafter\def\csname leanline@Formal/Core/Histogram.lean:flin\endcsname{70}
\expandafter\def\csname leanline@Formal/Core/Histogram.lean:half\endcsname{39}
\expandafter\def\csname leanline@Formal/Core/Histogram.lean:half_add_half\endcsname{41}
\expandafter\def\csname leanline@Formal/Core/Histogram.lean:half_ne_zero\endcsname{53}
\expandafter\def\csname leanline@Formal/Core/Histogram.lean:neg_half\endcsname{49}
\expandafter\def\csname leanline@Formal/Core/LinearLayout.lean:LayoutWidthLe\endcsname{124}
\expandafter\def\csname leanline@Formal/Core/LinearLayout.lean:caterpillar\endcsname{34}
\expandafter\def\csname leanline@Formal/Core/LinearLayout.lean:caterpillar_WF\endcsname{68}
\expandafter\def\csname leanline@Formal/Core/LinearLayout.lean:caterpillar_append\endcsname{40}
\expandafter\def\csname leanline@Formal/Core/LinearLayout.lean:caterpillar_nil\endcsname{37}
\expandafter\def\csname leanline@Formal/Core/LinearLayout.lean:caterpillar_subtree\endcsname{97}
\expandafter\def\csname leanline@Formal/Core/LinearLayout.lean:caterpillar_verts\endcsname{45}
\expandafter\def\csname leanline@Formal/Core/LinearLayout.lean:costMode_layout_le\endcsname{212}
\expandafter\def\csname leanline@Formal/Core/LinearLayout.lean:costMode_layout_le_all_k\endcsname{179}
\expandafter\def\csname leanline@Formal/Core/LinearLayout.lean:cutRankOf_singleton_le\endcsname{137}
\expandafter\def\csname leanline@Formal/Core/LinearLayout.lean:layoutDecomp\endcsname{91}
\expandafter\def\csname leanline@Formal/Core/LinearLayout.lean:layoutWidthLe_take\endcsname{128}
\expandafter\def\csname leanline@Formal/Core/LinearLayout.lean:layout_widthBounded\endcsname{143}
\expandafter\def\csname leanline@Formal/Core/LinearLayout.lean:statesSig_leaf_le\endcsname{157}
\expandafter\def\csname leanline@Formal/Core/Signature.lean:adj\endcsname{32}
\expandafter\def\csname leanline@Formal/Core/Signature.lean:adj_symm\endcsname{34}
\expandafter\def\csname leanline@Formal/Core/Signature.lean:chi_well_defined\endcsname{124}
\expandafter\def\csname leanline@Formal/Core/Signature.lean:crossPar\endcsname{118}
\expandafter\def\csname leanline@Formal/Core/Signature.lean:f_eq_c_add_phi\endcsname{57}
\expandafter\def\csname leanline@Formal/Core/Signature.lean:phi\endcsname{53}
\expandafter\def\csname leanline@Formal/Core/Signature.lean:rawCross\endcsname{61}
\expandafter\def\csname leanline@Formal/Core/Signature.lean:rawCross_eq_sum_v\endcsname{71}
\expandafter\def\csname leanline@Formal/Core/Signature.lean:rawCross_eq_sum_w\endcsname{64}
\expandafter\def\csname leanline@Formal/Core/Signature.lean:rawCross_left\endcsname{79}
\expandafter\def\csname leanline@Formal/Core/Signature.lean:rawCross_right\endcsname{90}
\expandafter\def\csname leanline@Formal/Core/Signature.lean:rep\endcsname{109}
\expandafter\def\csname leanline@Formal/Core/Signature.lean:rep_spec\endcsname{112}
\expandafter\def\csname leanline@Formal/Core/Signature.lean:sig\endcsname{45}
\expandafter\def\csname leanline@Formal/Core/Signature.lean:sig_apply_not_mem\endcsname{48}
\expandafter\def\csname leanline@Formal/Core/Signature.lean:supported\endcsname{42}
\expandafter\def\csname leanline@Formal/Core/SopInstance.lean:N\endcsname{76}
\expandafter\def\csname leanline@Formal/Core/SopInstance.lean:S\endcsname{73}
\expandafter\def\csname leanline@Formal/Core/SopInstance.lean:SopInstance\endcsname{29}
\expandafter\def\csname leanline@Formal/Core/SopInstance.lean:edgeProd\endcsname{56}
\expandafter\def\csname leanline@Formal/Core/SopInstance.lean:f\endcsname{65}
\expandafter\def\csname leanline@Formal/Core/SopInstance.lean:half_mul_reduce\endcsname{83}
\expandafter\def\csname leanline@Formal/Core/SopInstance.lean:omega\endcsname{70}
\expandafter\def\csname leanline@Formal/Core/SopInstance.lean:selCount\endcsname{60}
\expandafter\def\csname leanline@Formal/Core/Width.lean:card_sig_image_le\endcsname{61}
\expandafter\def\csname leanline@Formal/Core/Width.lean:card_state_image_le\endcsname{80}
\expandafter\def\csname leanline@Formal/Core/Width.lean:cutMatrix\endcsname{28}
\expandafter\def\csname leanline@Formal/Core/Width.lean:cutRankOf\endcsname{32}
\expandafter\def\csname leanline@Formal/Core/Width.lean:recover\endcsname{39}
\expandafter\def\csname leanline@Formal/Core/Width.lean:sig_eq_recover_vecMul\endcsname{45}
\expandafter\def\csname leanline@Formal/Core/Width.lean:toRow\endcsname{35}
\expandafter\def\csname leanline@Formal/Foundations/CutRank.lean:card_image_vecMul\endcsname{26}
\expandafter\def\csname leanline@Formal/Foundations/CutRank.lean:card_image_vecMul_le\endcsname{50}
\expandafter\def\csname leanline@Formal/Foundations/RankDecomp.lean:RTree\endcsname{22}
\expandafter\def\csname leanline@Formal/Foundations/RankDecomp.lean:RankDecomp\endcsname{49}
\expandafter\def\csname leanline@Formal/Foundations/RankDecomp.lean:WF\endcsname{36}
\expandafter\def\csname leanline@Formal/Foundations/RankDecomp.lean:WF_leaf\endcsname{42}
\expandafter\def\csname leanline@Formal/Foundations/RankDecomp.lean:WF_node\endcsname{43}
\expandafter\def\csname leanline@Formal/Foundations/RankDecomp.lean:verts\endcsname{31}
\expandafter\def\csname leanline@Formal/Foundations/RankDecomp.lean:verts_leaf\endcsname{40}
\expandafter\def\csname leanline@Formal/Foundations/RankDecomp.lean:verts_node\endcsname{41}
\expandafter\def\csname leanline@Formal/Paper.lean:amplitude_by_rank_dp_normalized\endcsname{98}
\expandafter\def\csname leanline@Formal/Paper.lean:amplitude_eq_sop_normalized\endcsname{89}
\expandafter\def\csname leanline@Formal/Paper.lean:fourier_mode_spec\endcsname{41}
\expandafter\def\csname leanline@Formal/Paper.lean:histogram_one_mode_spec\endcsname{54}
\expandafter\def\csname leanline@Formal/Paper.lean:linear_fourier_mode_spec\endcsname{73}
\expandafter\def\csname leanline@Formal/Paper.lean:rank_dp_spec\endcsname{27}
\expandafter\def\csname leanline@Formal/Quantum/Capstone.lean:amplitude_by_rank_dp\endcsname{55}
\expandafter\def\csname leanline@Formal/Quantum/Capstone.lean:amplitude_eq_sop\endcsname{43}
\expandafter\def\csname leanline@Formal/Quantum/Capstone.lean:circuitInstance\endcsname{36}
\expandafter\def\csname leanline@Formal/Quantum/Circuit.lean:Circuit\endcsname{22}
\expandafter\def\csname leanline@Formal/Quantum/Circuit.lean:amplitude\endcsname{31}
\expandafter\def\csname leanline@Formal/Quantum/Circuit.lean:circuitMatrix\endcsname{26}
\expandafter\def\csname leanline@Formal/Quantum/Circuit.lean:hadamardCount\endcsname{35}
\expandafter\def\csname leanline@Formal/Quantum/Circuit.lean:normalization\endcsname{39}
\expandafter\def\csname leanline@Formal/Quantum/Coupling.lean:H_rec_ext_last\endcsname{352}
\expandafter\def\csname leanline@Formal/Quantum/Coupling.lean:H_rec_ext_lt\endcsname{338}
\expandafter\def\csname leanline@Formal/Quantum/Coupling.lean:Sym.DataBounded\endcsname{225}
\expandafter\def\csname leanline@Formal/Quantum/Coupling.lean:circuitMatrix_snoc\endcsname{29}
\expandafter\def\csname leanline@Formal/Quantum/Coupling.lean:compileInvariant_holds\endcsname{698}
\expandafter\def\csname leanline@Formal/Quantum/Coupling.lean:compile_invariant\endcsname{665}
\expandafter\def\csname leanline@Formal/Quantum/Coupling.lean:compile_mH_eq_hadamardCount\endcsname{44}
\expandafter\def\csname leanline@Formal/Quantum/Coupling.lean:compile_snoc\endcsname{35}
\expandafter\def\csname leanline@Formal/Quantum/Coupling.lean:dataBounded_compile\endcsname{320}
\expandafter\def\csname leanline@Formal/Quantum/Coupling.lean:dataBounded_init\endcsname{228}
\expandafter\def\csname leanline@Formal/Quantum/Coupling.lean:dataBounded_step\endcsname{231}
\expandafter\def\csname leanline@Formal/Quantum/Coupling.lean:lift28_val\endcsname{74}
\expandafter\def\csname leanline@Formal/Quantum/Coupling.lean:lift28_val4\endcsname{77}
\expandafter\def\csname leanline@Formal/Quantum/Coupling.lean:lift28_val4_diag\endcsname{85}
\expandafter\def\csname leanline@Formal/Quantum/Coupling.lean:lift28_val4_swap\endcsname{81}
\expandafter\def\csname leanline@Formal/Quantum/Coupling.lean:omega8_pow_8\endcsname{54}
\expandafter\def\csname leanline@Formal/Quantum/Coupling.lean:omega8_pow_mod\endcsname{61}
\expandafter\def\csname leanline@Formal/Quantum/Coupling.lean:omega8_pow_val_add\endcsname{65}
\expandafter\def\csname leanline@Formal/Quantum/Coupling.lean:pathSum_H_assemble\endcsname{511}
\expandafter\def\csname leanline@Formal/Quantum/Coupling.lean:pathSum_step_CZ\endcsname{133}
\expandafter\def\csname leanline@Formal/Quantum/Coupling.lean:pathSum_step_H\endcsname{621}
\expandafter\def\csname leanline@Formal/Quantum/Coupling.lean:pathSum_step_T\endcsname{101}
\expandafter\def\csname leanline@Formal/Quantum/Coupling.lean:phase_H_snoc_none\endcsname{367}
\expandafter\def\csname leanline@Formal/Quantum/Coupling.lean:phase_H_snoc_some\endcsname{431}
\expandafter\def\csname leanline@Formal/Quantum/Coupling.lean:sqrt2C_ne_zero\endcsname{69}
\expandafter\def\csname leanline@Formal/Quantum/Coupling.lean:sum_agree_off_reindex\endcsname{642}
\expandafter\def\csname leanline@Formal/Quantum/Coupling.lean:wireVal_none\endcsname{89}
\expandafter\def\csname leanline@Formal/Quantum/Coupling.lean:wireVal_some\endcsname{94}
\expandafter\def\csname leanline@Formal/Quantum/Gates.lean:Gate\endcsname{33}
\expandafter\def\csname leanline@Formal/Quantum/Gates.lean:bit\endcsname{39}
\expandafter\def\csname leanline@Formal/Quantum/Gates.lean:gateMatrix\endcsname{43}
\expandafter\def\csname leanline@Formal/Quantum/Gates.lean:omega8\endcsname{30}
\expandafter\def\csname leanline@Formal/Quantum/Pinning.lean:CurInj\endcsname{34}
\expandafter\def\csname leanline@Formal/Quantum/Pinning.lean:EparIrrefl\endcsname{41}
\expandafter\def\csname leanline@Formal/Quantum/Pinning.lean:EparSym\endcsname{38}
\expandafter\def\csname leanline@Formal/Quantum/Pinning.lean:PinningHolds\endcsname{90}
\expandafter\def\csname leanline@Formal/Quantum/Pinning.lean:constDelta\endcsname{56}
\expandafter\def\csname leanline@Formal/Quantum/Pinning.lean:freeSet\endcsname{48}
\expandafter\def\csname leanline@Formal/Quantum/Pinning.lean:liveSet\endcsname{44}
\expandafter\def\csname leanline@Formal/Quantum/Pinning.lean:pinVal\endcsname{52}
\expandafter\def\csname leanline@Formal/Quantum/Pinning.lean:pinnedInstance\endcsname{63}
\expandafter\def\csname leanline@Formal/Quantum/PinningLemmas.lean:curInj_compile\endcsname{158}
\expandafter\def\csname leanline@Formal/Quantum/PinningLemmas.lean:curInj_init\endcsname{27}
\expandafter\def\csname leanline@Formal/Quantum/PinningLemmas.lean:curInj_step\endcsname{35}
\expandafter\def\csname leanline@Formal/Quantum/PinningLemmas.lean:disjoint_free_live\endcsname{301}
\expandafter\def\csname leanline@Formal/Quantum/PinningLemmas.lean:eparIrrefl_compile\endcsname{164}
\expandafter\def\csname leanline@Formal/Quantum/PinningLemmas.lean:eparIrrefl_init\endcsname{33}
\expandafter\def\csname leanline@Formal/Quantum/PinningLemmas.lean:eparIrrefl_step\endcsname{111}
\expandafter\def\csname leanline@Formal/Quantum/PinningLemmas.lean:eparSym_compile\endcsname{161}
\expandafter\def\csname leanline@Formal/Quantum/PinningLemmas.lean:eparSym_init\endcsname{31}
\expandafter\def\csname leanline@Formal/Quantum/PinningLemmas.lean:eparSym_step\endcsname{73}
\expandafter\def\csname leanline@Formal/Quantum/PinningLemmas.lean:ext_fin\endcsname{204}
\expandafter\def\csname leanline@Formal/Quantum/PinningLemmas.lean:ext_mergeAssign\endcsname{192}
\expandafter\def\csname leanline@Formal/Quantum/PinningLemmas.lean:flipEdge_eparIrrefl\endcsname{102}
\expandafter\def\csname leanline@Formal/Quantum/PinningLemmas.lean:flipEdge_eparSym\endcsname{63}
\expandafter\def\csname leanline@Formal/Quantum/PinningLemmas.lean:free_union_live\endcsname{297}
\expandafter\def\csname leanline@Formal/Quantum/PinningLemmas.lean:indicator_iff\endcsname{251}
\expandafter\def\csname leanline@Formal/Quantum/PinningLemmas.lean:invariants_compile\endcsname{140}
\expandafter\def\csname leanline@Formal/Quantum/PinningLemmas.lean:lift28_mul_ite\endcsname{338}
\expandafter\def\csname leanline@Formal/Quantum/PinningLemmas.lean:liveSet_subset_range\endcsname{173}
\expandafter\def\csname leanline@Formal/Quantum/PinningLemmas.lean:mem_freeSet_iff\endcsname{200}
\expandafter\def\csname leanline@Formal/Quantum/PinningLemmas.lean:mem_liveSet_iff\endcsname{170}
\expandafter\def\csname leanline@Formal/Quantum/PinningLemmas.lean:mergeAssign\endcsname{188}
\expandafter\def\csname leanline@Formal/Quantum/PinningLemmas.lean:merge_free\endcsname{285}
\expandafter\def\csname leanline@Formal/Quantum/PinningLemmas.lean:merge_live\endcsname{290}
\expandafter\def\csname leanline@Formal/Quantum/PinningLemmas.lean:omega_pinned\endcsname{278}
\expandafter\def\csname leanline@Formal/Quantum/PinningLemmas.lean:pathSum_pinned\endcsname{525}
\expandafter\def\csname leanline@Formal/Quantum/PinningLemmas.lean:phase_merge\endcsname{509}
\expandafter\def\csname leanline@Formal/Quantum/PinningLemmas.lean:pinVal_spec\endcsname{178}
\expandafter\def\csname leanline@Formal/Quantum/PinningLemmas.lean:pin_bsum\endcsname{328}
\expandafter\def\csname leanline@Formal/Quantum/PinningLemmas.lean:pin_freefree\endcsname{344}
\expandafter\def\csname leanline@Formal/Quantum/PinningLemmas.lean:pin_mixed\endcsname{413}
\expandafter\def\csname leanline@Formal/Quantum/PinningLemmas.lean:pin_quad\endcsname{448}
\expandafter\def\csname leanline@Formal/Quantum/PinningLemmas.lean:pinned_f_eval\endcsname{471}
\expandafter\def\csname leanline@Formal/Quantum/PinningLemmas.lean:pinningHolds\endcsname{555}
\expandafter\def\csname leanline@Formal/Quantum/PinningLemmas.lean:sum_free_attach\endcsname{307}
\expandafter\def\csname leanline@Formal/Quantum/PinningLemmas.lean:sum_pinned_reindex\endcsname{210}
\expandafter\def\csname leanline@Formal/Quantum/PinningLemmas.lean:sum_union_double\endcsname{316}
\expandafter\def\csname leanline@Formal/Quantum/SymLemmas.lean:WFS\endcsname{180}
\expandafter\def\csname leanline@Formal/Quantum/SymLemmas.lean:addB_b\endcsname{56}
\expandafter\def\csname leanline@Formal/Quantum/SymLemmas.lean:addB_c\endcsname{54}
\expandafter\def\csname leanline@Formal/Quantum/SymLemmas.lean:addB_cur\endcsname{52}
\expandafter\def\csname leanline@Formal/Quantum/SymLemmas.lean:addB_epar\endcsname{59}
\expandafter\def\csname leanline@Formal/Quantum/SymLemmas.lean:addB_ext\endcsname{63}
\expandafter\def\csname leanline@Formal/Quantum/SymLemmas.lean:addB_mH\endcsname{61}
\expandafter\def\csname leanline@Formal/Quantum/SymLemmas.lean:addB_next\endcsname{50}
\expandafter\def\csname leanline@Formal/Quantum/SymLemmas.lean:addB_wireVal\endcsname{65}
\expandafter\def\csname leanline@Formal/Quantum/SymLemmas.lean:addC_b\endcsname{40}
\expandafter\def\csname leanline@Formal/Quantum/SymLemmas.lean:addC_c\endcsname{38}
\expandafter\def\csname leanline@Formal/Quantum/SymLemmas.lean:addC_cur\endcsname{36}
\expandafter\def\csname leanline@Formal/Quantum/SymLemmas.lean:addC_epar\endcsname{42}
\expandafter\def\csname leanline@Formal/Quantum/SymLemmas.lean:addC_ext\endcsname{46}
\expandafter\def\csname leanline@Formal/Quantum/SymLemmas.lean:addC_mH\endcsname{44}
\expandafter\def\csname leanline@Formal/Quantum/SymLemmas.lean:addC_next\endcsname{34}
\expandafter\def\csname leanline@Formal/Quantum/SymLemmas.lean:addC_wireVal\endcsname{48}
\expandafter\def\csname leanline@Formal/Quantum/SymLemmas.lean:flipEdge_b\endcsname{74}
\expandafter\def\csname leanline@Formal/Quantum/SymLemmas.lean:flipEdge_c\endcsname{72}
\expandafter\def\csname leanline@Formal/Quantum/SymLemmas.lean:flipEdge_cur\endcsname{70}
\expandafter\def\csname leanline@Formal/Quantum/SymLemmas.lean:flipEdge_epar\endcsname{76}
\expandafter\def\csname leanline@Formal/Quantum/SymLemmas.lean:flipEdge_ext\endcsname{82}
\expandafter\def\csname leanline@Formal/Quantum/SymLemmas.lean:flipEdge_mH\endcsname{80}
\expandafter\def\csname leanline@Formal/Quantum/SymLemmas.lean:flipEdge_next\endcsname{68}
\expandafter\def\csname leanline@Formal/Quantum/SymLemmas.lean:flipEdge_wireVal\endcsname{84}
\expandafter\def\csname leanline@Formal/Quantum/SymLemmas.lean:lift28_succ\endcsname{223}
\expandafter\def\csname leanline@Formal/Quantum/SymLemmas.lean:pathSum_init\endcsname{354}
\expandafter\def\csname leanline@Formal/Quantum/SymLemmas.lean:phase_addB\endcsname{231}
\expandafter\def\csname leanline@Formal/Quantum/SymLemmas.lean:phase_addC\endcsname{226}
\expandafter\def\csname leanline@Formal/Quantum/SymLemmas.lean:phase_flipEdge\endcsname{283}
\expandafter\def\csname leanline@Formal/Quantum/SymLemmas.lean:step_CZ_none_none\endcsname{135}
\expandafter\def\csname leanline@Formal/Quantum/SymLemmas.lean:step_CZ_none_some\endcsname{140}
\expandafter\def\csname leanline@Formal/Quantum/SymLemmas.lean:step_CZ_some_none\endcsname{145}
\expandafter\def\csname leanline@Formal/Quantum/SymLemmas.lean:step_CZ_some_some_eq\endcsname{150}
\expandafter\def\csname leanline@Formal/Quantum/SymLemmas.lean:step_CZ_some_some_ne\endcsname{155}
\expandafter\def\csname leanline@Formal/Quantum/SymLemmas.lean:step_H_ext_snoc_last\endcsname{341}
\expandafter\def\csname leanline@Formal/Quantum/SymLemmas.lean:step_H_ext_snoc_lt\endcsname{331}
\expandafter\def\csname leanline@Formal/Quantum/SymLemmas.lean:step_H_none\endcsname{160}
\expandafter\def\csname leanline@Formal/Quantum/SymLemmas.lean:step_H_some\endcsname{168}
\expandafter\def\csname leanline@Formal/Quantum/SymLemmas.lean:step_T_none\endcsname{127}
\expandafter\def\csname leanline@Formal/Quantum/SymLemmas.lean:step_T_some\endcsname{131}
\expandafter\def\csname leanline@Formal/Quantum/SymLemmas.lean:step_cur_CZ\endcsname{120}
\expandafter\def\csname leanline@Formal/Quantum/SymLemmas.lean:step_cur_H\endcsname{95}
\expandafter\def\csname leanline@Formal/Quantum/SymLemmas.lean:step_cur_T\endcsname{106}
\expandafter\def\csname leanline@Formal/Quantum/SymLemmas.lean:step_mH_CZ\endcsname{115}
\expandafter\def\csname leanline@Formal/Quantum/SymLemmas.lean:step_mH_H\endcsname{92}
\expandafter\def\csname leanline@Formal/Quantum/SymLemmas.lean:step_mH_T\endcsname{102}
\expandafter\def\csname leanline@Formal/Quantum/SymLemmas.lean:step_next_CZ\endcsname{110}
\expandafter\def\csname leanline@Formal/Quantum/SymLemmas.lean:step_next_H\endcsname{89}
\expandafter\def\csname leanline@Formal/Quantum/SymLemmas.lean:step_next_T\endcsname{98}
\expandafter\def\csname leanline@Formal/Quantum/SymLemmas.lean:sum_extend\endcsname{316}
\expandafter\def\csname leanline@Formal/Quantum/SymLemmas.lean:sum_sum_ite_pair\endcsname{266}
\expandafter\def\csname leanline@Formal/Quantum/SymLemmas.lean:sum_sum_ite_single\endcsname{250}
\expandafter\def\csname leanline@Formal/Quantum/SymLemmas.lean:wfs_compile\endcsname{218}
\expandafter\def\csname leanline@Formal/Quantum/SymLemmas.lean:wfs_foldl\endcsname{210}
\expandafter\def\csname leanline@Formal/Quantum/SymLemmas.lean:wfs_init\endcsname{182}
\expandafter\def\csname leanline@Formal/Quantum/SymLemmas.lean:wfs_step\endcsname{186}
\expandafter\def\csname leanline@Formal/Quantum/Symbolic.lean:CompileInvariant\endcsname{123}
\expandafter\def\csname leanline@Formal/Quantum/Symbolic.lean:Sym\endcsname{31}
\expandafter\def\csname leanline@Formal/Quantum/Symbolic.lean:addB\endcsname{59}
\expandafter\def\csname leanline@Formal/Quantum/Symbolic.lean:addC\endcsname{63}
\expandafter\def\csname leanline@Formal/Quantum/Symbolic.lean:compile\endcsname{93}
\expandafter\def\csname leanline@Formal/Quantum/Symbolic.lean:ext\endcsname{99}
\expandafter\def\csname leanline@Formal/Quantum/Symbolic.lean:flipEdge\endcsname{54}
\expandafter\def\csname leanline@Formal/Quantum/Symbolic.lean:init\endcsname{50}
\expandafter\def\csname leanline@Formal/Quantum/Symbolic.lean:lift28\endcsname{67}
\expandafter\def\csname leanline@Formal/Quantum/Symbolic.lean:pathSum\endcsname{115}
\expandafter\def\csname leanline@Formal/Quantum/Symbolic.lean:phase\endcsname{109}
\expandafter\def\csname leanline@Formal/Quantum/Symbolic.lean:step\endcsname{70}
\expandafter\def\csname leanline@Formal/Quantum/Symbolic.lean:wireVal\endcsname{103}
\expandafter\def\csname leanline@Formal/Stab/Affine.lean:affine\endcsname{85}
\expandafter\def\csname leanline@Formal/Stab/Affine.lean:affineIndicator\endcsname{104}
\expandafter\def\csname leanline@Formal/Stab/Affine.lean:affineSystem\endcsname{141}
\expandafter\def\csname leanline@Formal/Stab/Affine.lean:affine_mul_expand\endcsname{165}
\expandafter\def\csname leanline@Formal/Stab/Affine.lean:affine_restrict\endcsname{89}
\expandafter\def\csname leanline@Formal/Stab/Affine.lean:affine_single\endcsname{97}
\expandafter\def\csname leanline@Formal/Stab/Affine.lean:affine_succ\endcsname{263}
\expandafter\def\csname leanline@Formal/Stab/Affine.lean:cast_val_mul_lift4\endcsname{77}
\expandafter\def\csname leanline@Formal/Stab/Affine.lean:eval_affineIndicator\endcsname{127}
\expandafter\def\csname leanline@Formal/Stab/Affine.lean:isQPF_affineIndicator\endcsname{136}
\expandafter\def\csname leanline@Formal/Stab/Affine.lean:isQPF_affinePhase\endcsname{256}
\expandafter\def\csname leanline@Formal/Stab/Affine.lean:isQPF_affinePhase_aux\endcsname{223}
\expandafter\def\csname leanline@Formal/Stab/Affine.lean:isQPF_affineProdSign\endcsname{192}
\expandafter\def\csname leanline@Formal/Stab/Affine.lean:isQPF_affineSystem\endcsname{151}
\expandafter\def\csname leanline@Formal/Stab/Affine.lean:isQPF_coordPhase\endcsname{206}
\expandafter\def\csname leanline@Formal/Stab/Affine.lean:isQPF_onePlusIPow\endcsname{274}
\expandafter\def\csname leanline@Formal/Stab/Affine.lean:lift4_mul\endcsname{71}
\expandafter\def\csname leanline@Formal/Stab/Affine.lean:sat_affineIndicator_iff\endcsname{113}
\expandafter\def\csname leanline@Formal/Stab/Affine.lean:sgn_eq_ite\endcsname{64}
\expandafter\def\csname leanline@Formal/Stab/Affine.lean:sgn_one\endcsname{41}
\expandafter\def\csname leanline@Formal/Stab/Affine.lean:zeta4_one\endcsname{46}
\expandafter\def\csname leanline@Formal/Stab/Affine.lean:zeta4_three\endcsname{56}
\expandafter\def\csname leanline@Formal/Stab/Affine.lean:zeta4_two\endcsname{51}
\expandafter\def\csname leanline@Formal/Stab/Closure.lean:IsQPF.comp_elim_zero\endcsname{468}
\expandafter\def\csname leanline@Formal/Stab/Closure.lean:IsQPF.comp_inl\endcsname{462}
\expandafter\def\csname leanline@Formal/Stab/Closure.lean:IsQPF.pushforward\endcsname{523}
\expandafter\def\csname leanline@Formal/Stab/Closure.lean:IsQPF.sumLeft\endcsname{490}
\expandafter\def\csname leanline@Formal/Stab/Closure.lean:IsQPF.sumOne\endcsname{314}
\expandafter\def\csname leanline@Formal/Stab/Closure.lean:IsQPF.sumOver\endcsname{334}
\expandafter\def\csname leanline@Formal/Stab/Closure.lean:dropVar\endcsname{57}
\expandafter\def\csname leanline@Formal/Stab/Closure.lean:dropVar_M\endcsname{66}
\expandafter\def\csname leanline@Formal/Stab/Closure.lean:dropVar_d\endcsname{69}
\expandafter\def\csname leanline@Formal/Stab/Closure.lean:eq_pin_of_sat\endcsname{224}
\expandafter\def\csname leanline@Formal/Stab/Closure.lean:eval_padRight\endcsname{415}
\expandafter\def\csname leanline@Formal/Stab/Closure.lean:eval_sum_update_free\endcsname{172}
\expandafter\def\csname leanline@Formal/Stab/Closure.lean:eval_sum_update_pinned\endcsname{301}
\expandafter\def\csname leanline@Formal/Stab/Closure.lean:eval_update_pin\endcsname{286}
\expandafter\def\csname leanline@Formal/Stab/Closure.lean:eval_zeroLeft\endcsname{451}
\expandafter\def\csname leanline@Formal/Stab/Closure.lean:linVal_dropVar\endcsname{75}
\expandafter\def\csname leanline@Formal/Stab/Closure.lean:linVal_padRight\endcsname{396}
\expandafter\def\csname leanline@Formal/Stab/Closure.lean:linVal_update\endcsname{105}
\expandafter\def\csname leanline@Formal/Stab/Closure.lean:linVal_zeroLeft\endcsname{432}
\expandafter\def\csname leanline@Formal/Stab/Closure.lean:mulVec_update\endcsname{195}
\expandafter\def\csname leanline@Formal/Stab/Closure.lean:padRight\endcsname{384}
\expandafter\def\csname leanline@Formal/Stab/Closure.lean:pin\endcsname{210}
\expandafter\def\csname leanline@Formal/Stab/Closure.lean:pinCoef\endcsname{206}
\expandafter\def\csname leanline@Formal/Stab/Closure.lean:pin_eq\endcsname{213}
\expandafter\def\csname leanline@Formal/Stab/Closure.lean:quadCoef\endcsname{101}
\expandafter\def\csname leanline@Formal/Stab/Closure.lean:quadVal_dropVar\endcsname{84}
\expandafter\def\csname leanline@Formal/Stab/Closure.lean:quadVal_padRight\endcsname{400}
\expandafter\def\csname leanline@Formal/Stab/Closure.lean:quadVal_update\endcsname{113}
\expandafter\def\csname leanline@Formal/Stab/Closure.lean:quadVal_zeroLeft\endcsname{436}
\expandafter\def\csname leanline@Formal/Stab/Closure.lean:sat_dropVar_iff\endcsname{71}
\expandafter\def\csname leanline@Formal/Stab/Closure.lean:sat_padRight_iff\endcsname{404}
\expandafter\def\csname leanline@Formal/Stab/Closure.lean:sat_substVar_iff\endcsname{265}
\expandafter\def\csname leanline@Formal/Stab/Closure.lean:sat_update_of_col_zero\endcsname{154}
\expandafter\def\csname leanline@Formal/Stab/Closure.lean:sat_zeroLeft_iff\endcsname{440}
\expandafter\def\csname leanline@Formal/Stab/Closure.lean:substVar\endcsname{236}
\expandafter\def\csname leanline@Formal/Stab/Closure.lean:substVar_eps\endcsname{245}
\expandafter\def\csname leanline@Formal/Stab/Closure.lean:substVar_linVal\endcsname{248}
\expandafter\def\csname leanline@Formal/Stab/Closure.lean:substVar_quadVal\endcsname{251}
\expandafter\def\csname leanline@Formal/Stab/Closure.lean:substVar_rowVal\endcsname{254}
\expandafter\def\csname leanline@Formal/Stab/Closure.lean:sum_split_erase\endcsname{38}
\expandafter\def\csname leanline@Formal/Stab/Closure.lean:sum_sumElim_eq\endcsname{476}
\expandafter\def\csname leanline@Formal/Stab/Closure.lean:sum_zmod_two\endcsname{43}
\expandafter\def\csname leanline@Formal/Stab/Closure.lean:zeroLeft\endcsname{424}
\expandafter\def\csname leanline@Formal/Stab/Hybrid.lean:bestS\endcsname{172}
\expandafter\def\csname leanline@Formal/Stab/Hybrid.lean:bestS_leaf\endcsname{178}
\expandafter\def\csname leanline@Formal/Stab/Hybrid.lean:bestS_node\endcsname{181}
\expandafter\def\csname leanline@Formal/Stab/Hybrid.lean:costOpt\endcsname{143}
\expandafter\def\csname leanline@Formal/Stab/Hybrid.lean:costOpt_le\endcsname{157}
\expandafter\def\csname leanline@Formal/Stab/Hybrid.lean:costOpt_leaf\endcsname{148}
\expandafter\def\csname leanline@Formal/Stab/Hybrid.lean:costOpt_node\endcsname{151}
\expandafter\def\csname leanline@Formal/Stab/Hybrid.lean:costSwitch\endcsname{65}
\expandafter\def\csname leanline@Formal/Stab/Hybrid.lean:costSwitch_bestS\endcsname{189}
\expandafter\def\csname leanline@Formal/Stab/Hybrid.lean:costSwitch_clifford_le\endcsname{125}
\expandafter\def\csname leanline@Formal/Stab/Hybrid.lean:costSwitch_false_eq_costMode\endcsname{85}
\expandafter\def\csname leanline@Formal/Stab/Hybrid.lean:costSwitch_false_le\endcsname{96}
\expandafter\def\csname leanline@Formal/Stab/Hybrid.lean:costSwitch_leaf\endcsname{71}
\expandafter\def\csname leanline@Formal/Stab/Hybrid.lean:costSwitch_node\endcsname{75}
\expandafter\def\csname leanline@Formal/Stab/Hybrid.lean:cutRankOf_verts_root\endcsname{105}
\expandafter\def\csname leanline@Formal/Stab/Hybrid.lean:statesSig_root\endcsname{117}
\expandafter\def\csname leanline@Formal/Stab/Hybrid.lean:switchCharge\endcsname{59}
\expandafter\def\csname leanline@Formal/Stab/Join.lean:IsQPF.comp_inr\endcsname{40}
\expandafter\def\csname leanline@Formal/Stab/Join.lean:crossPar\endcsname{47}
\expandafter\def\csname leanline@Formal/Stab/Join.lean:crossQuad\endcsname{66}
\expandafter\def\csname leanline@Formal/Stab/Join.lean:crossQuad_apply\endcsname{72}
\expandafter\def\csname leanline@Formal/Stab/Join.lean:isQPF_crossSign\endcsname{91}
\expandafter\def\csname leanline@Formal/Stab/Join.lean:isQPF_join\endcsname{101}
\expandafter\def\csname leanline@Formal/Stab/Join.lean:isQPF_join_point\endcsname{125}
\expandafter\def\csname leanline@Formal/Stab/Join.lean:joinMap\endcsname{52}
\expandafter\def\csname leanline@Formal/Stab/Join.lean:joinMap_mulVec\endcsname{56}
\expandafter\def\csname leanline@Formal/Stab/Magic.lean:ACliff\endcsname{233}
\expandafter\def\csname leanline@Formal/Stab/Magic.lean:ACliff_of_dvd_b\endcsname{259}
\expandafter\def\csname leanline@Formal/Stab/Magic.lean:ACliff_of_dvd_mode\endcsname{264}
\expandafter\def\csname leanline@Formal/Stab/Magic.lean:ACliff_of_dvd_mul\endcsname{250}
\expandafter\def\csname leanline@Formal/Stab/Magic.lean:AMagic\endcsname{237}
\expandafter\def\csname leanline@Formal/Stab/Magic.lean:Nhat_even\endcsname{472}
\expandafter\def\csname leanline@Formal/Stab/Magic.lean:card_orderedSel_eq_selCount\endcsname{187}
\expandafter\def\csname leanline@Formal/Stab/Magic.lean:card_powerset_magicSet\endcsname{367}
\expandafter\def\csname leanline@Formal/Stab/Magic.lean:chi_half\endcsname{139}
\expandafter\def\csname leanline@Formal/Stab/Magic.lean:chi_mode_phi\endcsname{299}
\expandafter\def\csname leanline@Formal/Stab/Magic.lean:chi_sum\endcsname{288}
\expandafter\def\csname leanline@Formal/Stab/Magic.lean:decidableACliff\endcsname{240}
\expandafter\def\csname leanline@Formal/Stab/Magic.lean:decidableAMagic\endcsname{245}
\expandafter\def\csname leanline@Formal/Stab/Magic.lean:edgeProd_mk\endcsname{151}
\expandafter\def\csname leanline@Formal/Stab/Magic.lean:exists_zeta4_of_cliff\endcsname{271}
\expandafter\def\csname leanline@Formal/Stab/Magic.lean:isQPF_bitPhase\endcsname{97}
\expandafter\def\csname leanline@Formal/Stab/Magic.lean:isQPF_finsetProd\endcsname{83}
\expandafter\def\csname leanline@Formal/Stab/Magic.lean:isQPF_modeWeight_of_cliff\endcsname{326}
\expandafter\def\csname leanline@Formal/Stab/Magic.lean:isQPF_modeWeight_of_even_b\endcsname{341}
\expandafter\def\csname leanline@Formal/Stab/Magic.lean:isQPF_modeWeight_of_even_mode\endcsname{352}
\expandafter\def\csname leanline@Formal/Stab/Magic.lean:magicSet\endcsname{363}
\expandafter\def\csname leanline@Formal/Stab/Magic.lean:modeWeight_sum_qpf\endcsname{381}
\expandafter\def\csname leanline@Formal/Stab/Magic.lean:mul_f_of_even\endcsname{456}
\expandafter\def\csname leanline@Formal/Stab/Magic.lean:mul_half_eq_zero_of_even\endcsname{442}
\expandafter\def\csname leanline@Formal/Stab/Magic.lean:omega_pow_quarter\endcsname{125}
\expandafter\def\csname leanline@Formal/Stab/Magic.lean:one_bit_decomp\endcsname{66}
\expandafter\def\csname leanline@Formal/Stab/Magic.lean:selCount_cast\endcsname{221}
\expandafter\def\csname leanline@Formal/Stab/Magic.lean:triAdj\endcsname{157}
\expandafter\def\csname leanline@Formal/Stab/Magic.lean:triAdj_apply\endcsname{161}
\expandafter\def\csname leanline@Formal/Stab/Magic.lean:triAdj_mul_eq_ite\endcsname{167}
\expandafter\def\csname leanline@Formal/Stab/Magic.lean:zeta4_sum\endcsname{55}
\expandafter\def\csname leanline@Formal/Stab/Paper.lean:clifford_collapse_spec\endcsname{86}
\expandafter\def\csname leanline@Formal/Stab/Paper.lean:hybrid_never_loses_spec\endcsname{97}
\expandafter\def\csname leanline@Formal/Stab/Paper.lean:magic_decomp_spec\endcsname{73}
\expandafter\def\csname leanline@Formal/Stab/Paper.lean:qpf_closure_spec\endcsname{44}
\expandafter\def\csname leanline@Formal/Stab/Paper.lean:stab_join_spec\endcsname{58}
\expandafter\def\csname leanline@Formal/Stab/Paper.lean:switch_optimum_spec\endcsname{108}
\expandafter\def\csname leanline@Formal/Stab/Paper.lean:the\endcsname{57}
\expandafter\def\csname leanline@Formal/Stab/QPF.lean:I_pow_four\endcsname{55}
\expandafter\def\csname leanline@Formal/Stab/QPF.lean:I_pow_mod\endcsname{61}
\expandafter\def\csname leanline@Formal/Stab/QPF.lean:IsQPF\endcsname{341}
\expandafter\def\csname leanline@Formal/Stab/QPF.lean:IsQPF.congr\endcsname{346}
\expandafter\def\csname leanline@Formal/Stab/QPF.lean:IsQPF.mul\endcsname{369}
\expandafter\def\csname leanline@Formal/Stab/QPF.lean:IsQPF.reindex\endcsname{380}
\expandafter\def\csname leanline@Formal/Stab/QPF.lean:IsQPF.smul\endcsname{375}
\expandafter\def\csname leanline@Formal/Stab/QPF.lean:QPF\endcsname{124}
\expandafter\def\csname leanline@Formal/Stab/QPF.lean:Sat\endcsname{149}
\expandafter\def\csname leanline@Formal/Stab/QPF.lean:cast_val_lift4\endcsname{116}
\expandafter\def\csname leanline@Formal/Stab/QPF.lean:const\endcsname{189}
\expandafter\def\csname leanline@Formal/Stab/QPF.lean:eval\endcsname{153}
\expandafter\def\csname leanline@Formal/Stab/QPF.lean:eval_const\endcsname{191}
\expandafter\def\csname leanline@Formal/Stab/QPF.lean:eval_mul\endcsname{261}
\expandafter\def\csname leanline@Formal/Stab/QPF.lean:eval_of_not_sat\endcsname{161}
\expandafter\def\csname leanline@Formal/Stab/QPF.lean:eval_of_sat\endcsname{156}
\expandafter\def\csname leanline@Formal/Stab/QPF.lean:eval_phase\endcsname{182}
\expandafter\def\csname leanline@Formal/Stab/QPF.lean:eval_point\endcsname{209}
\expandafter\def\csname leanline@Formal/Stab/QPF.lean:eval_reindex\endcsname{326}
\expandafter\def\csname leanline@Formal/Stab/QPF.lean:eval_smul\endcsname{278}
\expandafter\def\csname leanline@Formal/Stab/QPF.lean:isQPF_const\endcsname{351}
\expandafter\def\csname leanline@Formal/Stab/QPF.lean:isQPF_one\endcsname{356}
\expandafter\def\csname leanline@Formal/Stab/QPF.lean:isQPF_phase\endcsname{358}
\expandafter\def\csname leanline@Formal/Stab/QPF.lean:isQPF_point\endcsname{363}
\expandafter\def\csname leanline@Formal/Stab/QPF.lean:isQPF_zero\endcsname{354}
\expandafter\def\csname leanline@Formal/Stab/QPF.lean:lift4\endcsname{100}
\expandafter\def\csname leanline@Formal/Stab/QPF.lean:lift4_add\endcsname{110}
\expandafter\def\csname leanline@Formal/Stab/QPF.lean:lift4_one\endcsname{104}
\expandafter\def\csname leanline@Formal/Stab/QPF.lean:lift4_zero\endcsname{102}
\expandafter\def\csname leanline@Formal/Stab/QPF.lean:linVal\endcsname{143}
\expandafter\def\csname leanline@Formal/Stab/QPF.lean:linVal_mul\endcsname{236}
\expandafter\def\csname leanline@Formal/Stab/QPF.lean:linVal_reindex\endcsname{300}
\expandafter\def\csname leanline@Formal/Stab/QPF.lean:mul\endcsname{219}
\expandafter\def\csname leanline@Formal/Stab/QPF.lean:mul_eps\endcsname{228}
\expandafter\def\csname leanline@Formal/Stab/QPF.lean:mul_lin\endcsname{231}
\expandafter\def\csname leanline@Formal/Stab/QPF.lean:mul_quad\endcsname{234}
\expandafter\def\csname leanline@Formal/Stab/QPF.lean:neg_one_pow_mod\endcsname{65}
\expandafter\def\csname leanline@Formal/Stab/QPF.lean:phase\endcsname{168}
\expandafter\def\csname leanline@Formal/Stab/QPF.lean:point\endcsname{196}
\expandafter\def\csname leanline@Formal/Stab/QPF.lean:quadVal\endcsname{146}
\expandafter\def\csname leanline@Formal/Stab/QPF.lean:quadVal_mul\endcsname{241}
\expandafter\def\csname leanline@Formal/Stab/QPF.lean:quadVal_reindex\endcsname{306}
\expandafter\def\csname leanline@Formal/Stab/QPF.lean:reindex\endcsname{290}
\expandafter\def\csname leanline@Formal/Stab/QPF.lean:sat_mul_iff\endcsname{249}
\expandafter\def\csname leanline@Formal/Stab/QPF.lean:sat_phase\endcsname{176}
\expandafter\def\csname leanline@Formal/Stab/QPF.lean:sat_point_iff\endcsname{204}
\expandafter\def\csname leanline@Formal/Stab/QPF.lean:sat_reindex_iff\endcsname{314}
\expandafter\def\csname leanline@Formal/Stab/QPF.lean:sgn\endcsname{53}
\expandafter\def\csname leanline@Formal/Stab/QPF.lean:sgn_add\endcsname{83}
\expandafter\def\csname leanline@Formal/Stab/QPF.lean:sgn_ne_zero\endcsname{87}
\expandafter\def\csname leanline@Formal/Stab/QPF.lean:sgn_zero\endcsname{80}
\expandafter\def\csname leanline@Formal/Stab/QPF.lean:smul\endcsname{276}
\expandafter\def\csname leanline@Formal/Stab/QPF.lean:zeta4\endcsname{50}
\expandafter\def\csname leanline@Formal/Stab/QPF.lean:zeta4_add\endcsname{73}
\expandafter\def\csname leanline@Formal/Stab/QPF.lean:zeta4_ne_zero\endcsname{77}
\expandafter\def\csname leanline@Formal/Stab/QPF.lean:zeta4_two_mul\endcsname{93}
\expandafter\def\csname leanline@Formal/Stab/QPF.lean:zeta4_zero\endcsname{70}
}{%
  \GenericWarning{}{lean-anchors.tex missing: run scripts/lean_anchors.py}}

\title{Quadratic Sums-of-Powers for Fixed-Parameter Tractable Quantum-Circuit Simulation}
\author[1]{Alexis de Colnet}\email{a.e.h.de.colnet@liacs.leidenuniv.nl}
\author[2]{Floris Geerts}\email{floris.geerts@uantwerpen.be}
\author[3]{Rihan Hai}\email{r.hai@tudelft.nl}
\author[1]{Alfons Laarman}\email{a.w.laarman@liacs.leidenuniv.nl}
\author[1]{Joon Hyung Lee}\email{j.h.lee@liacs.leidenuniv.nl}
\author[2]{Guillermo A.\ P\'erez}
\email{guillermo.perez@uantwerpen.be}
\affil[1]{Leiden University, Netherlands}
\affil[2]{University of Antwerp, Belgium}
\affil[3]{TU Delft, Netherlands}
\date{}

\begin{document}

\maketitle
\begin{abstract}
Strongly simulating a quantum circuit, that is, computing an output amplitude, can be done by summing the circuit's Feynman paths: a weighted count over assignments to Boolean path variables.
The circuit's gates induce correlations among these variables, forming a graph whose structure controls several exact simulation routes.
This \emph{sum-of-powers} (SOP) viewpoint underlies recent simulators built on binary decision diagrams and weighted model counting.

For a quadratic SOP with $n$ variables, even modulus $r$, and a rank-decomposition of its variable graph of width $k$, our dynamic program (DP) computes an amplitude using only $O(4^k\poly(n))$ arithmetic operations.
For Clifford$+\T$ circuits, the amplitude is given by an SOP with modulus $8$.

Rank-width never exceeds linear rank-width, which governs some decision-diagram approaches, and is at most one greater than the Markov--Shi contraction complexity of the circuit tensor network.
Moreover, there are non-Clifford families of bounded rank-width where both competing parameters diverge.
We also present a \emph{stabilizer-rank optimization}, exploiting that the DP tables are stabilizer-type Gauss sums.
Each subtree runs at the \emph{width} price of its cut-ranks or at a \emph{magic} price that discharges the non-Clifford phases below it.
The resulting \emph{best total cost} never exceeds $O(4^k\poly(n))$, yet is polynomial on mixed families where the pure rank-width and pure $\T$-count guarantees are both exponential.
Clifford amplitudes take polynomial time on any graph, the exact-amplitude consequence of Gottesman--Knill.

A prototype evaluation on standard circuit benchmarks finds treewidth bucket elimination the strongest baseline, with the new rank-width DP complementary: it wins on structured families where treewidth blows up.
\end{abstract}

\clearpage
\tableofcontents
\clearpage

\section{Introduction}
Classical strong simulation of a quantum circuit $C$ supports compilation, optimization, and noise studies, and helps identify classically intractable regimes.
Because a $q$-qubit state has $2^q$ amplitudes, useful simulators must exploit structure rather than store it naively.

Tensor-network contraction is one such route.
For the tensor-network graph $N_C$ and its line graph $L(N_C)$ (\autoref{sec:tensor-network}), the relevant parameter is the \emph{contraction complexity} $\cc(N_C)=\tw(L(N_C))$, and the Markov--Shi bound is single-exponential in $\cc(N_C)$ and linear in circuit size~\cite{MarkovShi08}.

We instead use a \emph{sum-of-powers} (SOP)~\cite{WangChengYuanJi25}: a Feynman-path sum~\cite{Dawson+2005,Amy19PathSum} rather than a sum-of-products.
This representation underlies recent simulators based on binary decision diagrams~\cite{WangChengYuanJi25,DBLP:journals/corr/abs-2510-06775} and model counting~\cite{QPWMC,MeiBonsangueLaarman24,DBLP:conf/fm/QuistMCL24,HuangEtAl26}.
Pinning input $\by$ and output $\bz$ gives
\[
\braket{\bz|C|\by}=R_C^{-1}\delta_C(\by,\bz)\sum_{\bx\in\{0,1\}^V}\omega_r^{f(\bx)},
\]
a normalized root-of-unity sum over unpinned Boolean path variables $V$, where $\omega_r=e^{2\pi i/r}$ for an even modulus $r$ and $R_C=(\sqrt2)^{m_\H}$, with $m_\H$ the number of Hadamard gates in $C$.
Here $\delta_C$ is zero when boundary values disagree on a Hadamard-free wire.
The exponent $f$ is quadratic, and its quadratic terms $x_ux_v$ define the edges of the \emph{SOP variable graph} $G_C=(V,E)$.
We write $n=|V|$ for the number of unpinned path variables, each a wire segment strictly between two Hadamards, so $n\le m_\H$ over $\{\H,\T,\CZ\}$.
The modulus is $r=8$ there (\autoref{sec:sop}).
Diagonal single-qubit gates add only unary phases.
Other fixed finite gate sets admit constant-size gadgets and a refined modulus (\autoref{thm:gadget-universality}), and expansion can raise rank-width, so it is measured after expansion.

Our dynamic program (DP) evaluates this SOP in time exponential only in the \emph{rank-width} $k$ of $G_C$ and polynomial in $n$, given a rank-decomposition.
Three graphs recur: the tensor-network graph $N_C$, its line graph $L(N_C)$, whose treewidth is the contraction complexity, and the SOP variable graph $G_C$ on which our DP runs.
Treewidth transfers from $N_C$ to $G_C$ through $L(N_C)$ (\autoref{lem:sop-minor-line}), and rank-width directly from $G_C$.

\paragraph{Contributions.}
\begin{enumerate}
    \item An explicit DP runs on the variable graph of an $n$-variable SOP: given a rank-decomposition of width $k$, it computes a single amplitude in time $O(4^k\poly(n))$ (\autoref{cor:single-amplitude}), and on a linear layout of width $k$ in time $O(2^k\poly(n))$ (\autoref{thm:linear-layout-fourier}).
This amounts to rank-width fixed-parameter strong simulation for every circuit over Hadamard and diagonal gates, and in particular for Clifford$+\T$.
    \label{itm:contribution-1}
    \item We prove that $G_C$ is a minor of the tensor line graph, hence $\tw(G_C)\le\cc(N_C)$ (\autoref{lem:sop-minor-line}).
Since rank-width never exceeds branch-width, which is at most treewidth plus one (\autoref{sec:preliminaries})~\cite{oum2008rank,Oum17Survey}, this gives $\rw(G_C)\le\cc(N_C)+1$.
    \item We construct circuit families whose rank-width stays bounded while both linear rank-width and contraction complexity grow without bound, separating our guarantee from the decision-diagram and tensor-network ones (\autoref{cor:separating-family}).
    \item The DP tables carry stabilizer structure, which a \emph{stabilizer-rank optimization} exploits to refine \autoref{itm:contribution-1}.
Each subtree runs at the \emph{width} price of its cut-ranks or, at most once per branch, at a \emph{magic} price that discharges the non-Clifford phases below it.
The cost never exceeds \autoref{itm:contribution-1}'s $O(4^k\poly(n))$, stays polynomial on mixed families where both pure guarantees are exponential, and recovers Gottesman--Knill stabilizer-decomposition simulation as a special case (\autoref{sec:stabjoin}).
    \item A prototype is evaluated on standard benchmarks against treewidth, statevector, decision-diagram, and weighted-model-counting simulation (\autoref{sec:experiments}).
\end{enumerate}

\paragraph{Parameters versus running times.}
The exponential bases in \autoref{tab:route-comparison} matter as much as the widths.
The three routes that run on $G_C$ are charged on the same $n$-variable instance: branching rank-width, linear layout, and bucket elimination cost $O(2^{2\rw(G_C)}\poly(n))$, $O(2^{\lrw(G_C)}\poly(n))$, and $O(2^{\tw(G_C)+1}\poly(n))$, respectively.
The branching route therefore overtakes the others only when $2\rw(G_C)$ is below their exponents.
The generic width inequalities do not ensure that it is, but \autoref{cor:separating-family} exhibits families where it is.
The base $4$ is a cost of the branching join, and two routes avoid paying the cost.
A linear layout has no join at all and runs in base $2$ (\autoref{thm:linear-layout-fourier}).
Alternatively the stabilizer-rank optimization of \autoref{sec:stabjoin} reparameterizes the cost, charging the magic of a subtree before continuing at its width price wherever that is cheaper.

\begin{table}[t]
\caption{Structural simulation routes compared to the SOP DPs proposed here.
All bounds compute one amplitude on the parameter object named in the row, and the rank-width bounds assume a supplied decomposition.
Every row is parameterized by structure alone, except the stabilizer decomposition, by magic alone, and the last row, which combines the two over one decomposition through one-way switches from the magic price to the width price.
Each polynomial factor is in the size of the object its row runs on: a quadratic SOP with $n$ variables for the four SOP routes (\autoref{sec:treewidth}, \autoref{cor:single-amplitude}, \autoref{thm:linear-layout-fourier}, \autoref{sec:stabjoin}), the tensor network $N_C$ for contraction, and the elementary gate count for stabilizer decomposition.}
\label{tab:route-comparison}
\centering \footnotesize
\renewcommand{\arraystretch}{1.12}
\begin{tabular}{@{}L{0.28\linewidth}L{0.25\linewidth}L{0.42\linewidth}@{}}
\toprule
\bf Route & \bf Parameter object & \bf Running-time bound \\
\midrule
Treewidth bucket elimination & tree decomposition of $G_C$ & $2^{\tw(G_C)+1}\poly(n)$ (\autoref{sec:treewidth}) \\
Tensor-network contraction & line graph $L(N_C)$ & $2^{O(\cc(N_C))}\poly(|N_C|)$~\cite{MarkovShi08} \\
Low-rank stabilizer decomposition & $\T$-count $\tau$ of the circuit & $\mathrm{sr}(\tau)\poly(\cdot)$ with $\mathrm{sr}(\tau)=O(2^{0.3963\tau})$ for Clifford$+\T$~\cite{BravyiGosset16,BravyiEtAl19,QassimPashayanGosset21} \\
\midrule
Branching SOP DP & rank-decomposition of $G_C$ & $4^k\poly(n)$ (\autoref{cor:single-amplitude}) \\
Linear-layout SOP DP & linear layout of $G_C$ & $2^k\poly(n)$ (\autoref{thm:linear-layout-fourier}) \\
Stabilizer-rank optimization & rank-decomposition of $G_C$ with magic counts $\tau_t$ & \emph{best total cost} (\autoref{sec:stabjoin}), never exceeding $4^k\poly(n)$ \\
\bottomrule
\end{tabular}
\end{table}

\paragraph{Other related work.}
Weighted model counting has been proposed for quantum computations including partition functions~\cite{QPWMC}.
Here it serves as one practical backend for the SOP on which our rank-width algorithm runs.

Up to normalization and conjugation, a quadratic SOP is a complex-weight degree-two sum, or a graph-state/product-state overlap.
Van den Nest et al. connected graph-state simulation with Schmidt-rank width (the underlying graph's rank-width) and optimal tree tensor networks~\cite{VandenNestMDB06,VandenNestDVB07}.
Hiraishi and Imai bounded decision diagrams for these overlaps by linear rank-width~\cite{HiraishiImai14}.
Later work gave an $O^*(4^{\rw(\widetilde G)})$ tree-tensor-network algorithm on a subdivided graph $\widetilde G$ and an $O^*(2^{\mmomega\operatorname{bw}(G)/2})$ branch-decomposition algorithm~\cite{HiraishiImaiIwataLin18}, where $O^*$ suppresses polynomial factors.
Our exact finite-modulus DP instead uses the unsubdivided $G_C$ and returns how many assignments give each value of $f$, not only the amplitude (\autoref{thm:dp-correct}).
We also prove $\tw(G_C)\le\cc(N_C)$ (\autoref{lem:sop-minor-line}).

Independent work also uses rank-width through ZX calculus and stabilizer decompositions.
Before any rewriting, a quadratic SOP mode whose cross-terms contribute only signs, a complex two-spin partition function, a graph-state/product-state overlap, and a closed graph-like ZX scalar all share one interaction graph (\autoref{prop:four-guises}), so the coarse rank-width dependence necessarily coincides.
Kuyanov and Kissinger accordingly contract graph-like ZX diagrams with the same base $4$ on branching and base $2$ on linear decompositions~\cite{KuyanovKissinger26}.
Their local exponent takes the cheapest of the three cut-rank pairs at each unrooted tree vertex, so on a fixed decomposition it can be finer than the rooted SOP join.
Feynman-path decision diagrams~\cite{WangChengYuanJi25,DBLP:journals/corr/abs-2510-06775} are controlled by the same cut-signature phenomenon as our linear SOP DP, which the branching DP escapes by replacing the global variable order with a rank-decomposition.
Codsi and Laakkonen obtain slice-wise polynomial algorithms parameterized by treewidth or rank-width of an associated tensor network and by the number of non-Clifford operations, introducing \emph{focused} rank-width to account for their placement~\cite{CodsiLaakkonen26}.
We instead run directly on the SOP variable graph $G_C$, and obtain a magic-focused-width guarantee on the shared sign-edge subclass.
Since $\T$ and other diagonal gates add only unary phases, their number enters our bounds polynomially through circuit size, not exponentially.
\appref{app:kk} works out all three comparisons.

\paragraph{Formalization in Lean.}
Statements carrying a \textsf{\color{cyan}[Lean\,\doublecheck]} badge are machine-checked in \href{https://lean-lang.org/}{Lean}~4 (v4.32.0, Mathlib \texttt{81a5d25}), under \texttt{docs/lean} in the paper's code and data artifact~\cite{Artifact26}, with axiom profile \texttt{[propext, Classical.choice, Quot.sound]} and no \texttt{sorry}.

\section{Preliminaries}

This section collects the ingredients the rest of the paper relies on: the graph width parameters that will measure simulation cost, the quantum circuits we simulate, and the quadratic sum-of-powers those circuits compile to, which is the object our algorithm evaluates.
Tensor networks, against which we compare in \autoref{sec:tensor-network}, are recalled there.

\subsection{Graph notation and parameters}\label{sec:preliminaries}

We first collect the graph notation and width parameters used throughout the paper.
All graphs are finite, undirected, and simple unless explicitly stated otherwise.
For a graph $G=(V,E)$ and sets $S,T\subseteq V$, $\bA_G[S,T]$ denotes the submatrix of the adjacency matrix $\bA_G$ of $G$ with rows indexed by $S$ and columns indexed by $T$.
A singleton $S=\{v\}$ is written $\bA_G[v,T]$, the row vector indicating the neighbors of $v$ inside $T$.
Whenever ranks of adjacency submatrices are used, the field is $\F_2$.
When the graph is clear from context we drop the subscript and write $\bA$.
For a positive integer $n$ we write $[n]=\{1,\ldots,n\}$.
Elements of every graph and tree, decomposition trees included, are called \emph{vertices}.

\paragraph{Treewidth.}

A tree decomposition of a graph $G=(V,E)$ is a pair $(T,\mathcal B)$ where $T$ is a tree with vertex set $V(T)$ and $\mathcal B=(B_t)_{t\in V(T)}$ is a family of bags $B_t\subseteq V$ such that:
\begin{enumerate}[leftmargin=2em]
    \item Every vertex appears in some bag: $V=\bigcup_{t\in V(T)}B_t$.
    \item Every edge is covered by some bag: for every $uv\in E$, there is a vertex $t$ with $u,v\in B_t$.
    \item Bags containing a fixed vertex form a connected subtree: for every $v\in V$, the set $\{t\in V(T)\mid v\in B_t\}$ is connected in $T$.
\end{enumerate}
The width of $(T,\mathcal B)$ is $\max_{t\in V(T)}|B_t|-1$.
The treewidth of $G$ is
\[
\tw(G)=\min_{(T,\mathcal B)}\left(\max_{t\in V(T)}|B_t|-1\right),
\]
where the min ranges over all tree decompositions of $G$.
We use two standard facts: treewidth is monotone under taking minors, and $\tw(K_t)=t-1$, where $K_t$ denotes the complete graph on $t$ vertices.
Here a \emph{minor} of $G$ is any graph arising from $G$ by vertex deletions, edge deletions, and edge contractions.

Computing treewidth exactly is \NP-hard in general~\cite{Arnborg87}: deciding whether $\tw(G)\le k$ is \NP-complete when $k$ is part of the input.
The parameterized picture is far better.
For every fixed integer $k$, Bodlaender's algorithm~\cite{Bodlaender96} decides whether $\tw(G)\le k$ and, if the answer is positive, constructs a tree decomposition of $G$ of width at most $k$ in linear time.
We only need the weaker consequence that this is polynomial time for fixed $k$.

\paragraph{Rank-width.}
Let $G=(V,E)$ be a graph.
For $X\subseteq V$, define the cut-rank of $X$ in $G$ by
\[
\rho_G(X)=\rank_{\F_2}\bA_G[X,V\setminus X].
\]
A rank-decomposition of $G$ is a pair $(T,\mu)$ where $T$ is a subcubic tree and $\mu$ is a bijection from the leaves of $T$ to $V$.
Subcubic means every vertex of $T$ has degree at most three.

Each edge $e$ of $T$ induces a bipartition, or cut, of $V$.
Indeed, deleting $e$ splits $T$ into two connected components.
The leaves in these two components correspond, through $\mu$, to two disjoint subsets $X_e$ and $\overline{X}_e$ of $V$.
Thus, $V = X_e \uplus \overline{X}_e$, with $\uplus$ the disjoint union.

The width of an edge $e$ in $T$ is then $\rho_G(X_e)$, the width of the decomposition is the maximum width of its edges, and finally, the rank-width of $G$ is
\[
\rw(G)=\min_{(T,\mu)}\max_{e\in E(T)}\rho_G(X_e),
\]
where $E(T)$ is the edge set of $T$.

Fixed-width rank-decompositions can be found, if they exist, in polynomial time for each fixed width~\cite{OumSeymour06,HlinenyOum08,Oum17Survey}.

Linear rank-width is the linear-layout analog of rank-width.
A linear layout of $G$ is an ordering $\pi=(v_1,\ldots,v_n)$ of $V$.
Its width is $ \max_{1\le i<n}\rho_G(\{v_1,\ldots,v_i\}). $
The linear rank-width of $G$ is
\[
\lrw(G)=\min_{\pi}\max_{1\le i<n}\rho_G(\{v_1,\ldots,v_i\}).
\]
Equivalently, it is the restriction of rank-width to caterpillar-shaped rank-decompositions, up to the standard identification of a caterpillar spine with a linear layout~\cite{Oum17Survey}.
That is, the decomposition tree $T$ in the rank-decomposition $(T,\mu)$ is restricted to be a caterpillar: a tree whose nonleaf vertices form a path.
Hence
\[
\rw(G)\le \lrw(G).
\]
The inequality can be strict on natural graph families.
By Theorem~2 of Adler and Kant\'e, the linear rank-width of every forest equals its pathwidth, the width of an optimal tree decomposition whose tree is a path.
Complete binary trees therefore have bounded rank-width but linear rank-width growing with their height~\cite{AdlerK15}.
We use this later to show our algorithm provably beats previous approaches on some circuit families.

Two further standard width parameters appear only in comparisons: \emph{branch-width} $\operatorname{bw}(G)$, the branch-decomposition analog of treewidth, for which we use only the inequalities $\rw(G)\le\operatorname{bw}(G)\le\tw(G)+1$~\cite{oum2008rank,Oum17Survey}, and the pathwidth just mentioned.

\subsection{Quantum circuits and gate sets}\label{sec:quantum-c}
We recall quantum circuits and fix the elementary gate set on which we phrase our results.
For a $q$-qubit system, the state space is $(\C^2)^{\otimes q} \cong \C^{2^q}$, with computational basis states $\ket{x} = \ket{x_1\cdots x_q}$, $x_i\in\{0,1\}$.
A general $q$-qubit state is therefore $\ket{\psi}=\sum_{x\in\{0,1\}^q} a_x \ket{x}$ with $a_x\in\C$, and for normalized states $\sum_{x\in\{0,1\}^q} |a_x|^2 = 1$.
A quantum gate is a unitary linear operator acting on one or more qubits.
When a gate acts on only a subset of the $q$ qubits, it is understood to act as the identity on the remaining ones.
A quantum circuit is a finite sequence of such gates acting on a fixed number of qubits.

\paragraph{Elementary gate set.}
The circuits we consider are built over the set of elementary gates $\{\H,\T,\CZ\}$, referring to Hadamard, $\T$ gate and controlled-$\text{Z}$ gate, respectively, and with standard matrices
\[
\H=\frac{1}{\sqrt2}\begin{bmatrix}1&1\\1&-1\end{bmatrix},
\qquad
\T=\begin{bmatrix}1&0\\0&e^{i\pi/4}\end{bmatrix},
\text{ and }
\CZ=\begin{bmatrix}
1&0&0&0\\
0&1&0&0\\
0&0&1&0\\
0&0&0&-1
\end{bmatrix}.
\]
The gate set $\{\H,\T,\CZ\}$ is universal for quantum computation.
This gate set compactly presents the standard Clifford$+\T$ gate set~\cite{NielsenChuang10}, whose Clifford part is $\{\H,\text{S},\CZ\}$ with $\text S$ the phase gate with diagonal elements $1$ and $i$.
The $\T$ gate renders one Clifford gate redundant, namely $\text S = \T^2$.

\paragraph{Circuits.}
We write a $q$-qubit circuit as
\[
C = U_m\cdots U_2U_1,
\]
where the $U_p$ are unitaries ordered as applied, $U_1$ first and $U_m$ last, and we refer to the index $p$ as the \emph{position} of $U_p$ in the circuit.
A circuit over a fixed finite gate set is built from \emph{elementary gates}.
The action of $C$ on an input state $\ket{\psi}$ is $C\ket{\psi} = U_m\cdots U_2U_1\ket{\psi}$, and the matrix of the whole circuit is obtained by multiplying the gate matrices in reverse textual order, since the rightmost gate acts first.

To specify which qubits a gate acts on, we use the standard subscript convention.
Here $\H_i$ and $\T_i$ denote a single-qubit Hadamard or $\T$ gate applied to qubit $i$ (and the identity on every other wire), and $\CZ_{i,j}$ denotes a $\CZ$ gate applied to qubits $i$ and $j$ (and the identity elsewhere).
Since $\CZ$ is symmetric in its two arguments, the order of $i$ and $j$ does not matter.
For a qubit $i$, the \emph{Hadamard depth} $h_i$ is the total number of $\H$ gates that act on $i$ in a given circuit $C$.
This count will index the path variables on wire $i$ in \autoref{sec:sop}.

\paragraph{Strong simulation.}
Throughout, classical \emph{strong simulation} of a circuit $C$ means computing the amplitude
\[
\braket{\bz|C|\by}\in\C,
\]
for given input and output computational basis states $\ket{\by}$ and $\ket{\bz}$ with $\by,\bz\in\{0,1\}^q$.
Equivalently, it is the entry of the unitary associated with $C$ in row $\bz$ and column $\by$.
This is the quantity whose squared modulus $|\braket{\bz|C|\by}|^2$ is the probability of observing outcome $\bz$ when $C$ is run on input $\by$ and measured in the computational basis.
Exact computation of such an amplitude, and hence of the associated fixed-output probability, is the strong-simulation task considered throughout.
We also call it \emph{exact amplitude simulation} when the distinction from sampling-based (weak) simulation matters.

\begin{figure}[!tbp]
\centering
\begin{quantikz}[row sep=0.45cm, column sep=0.35cm]
\lstick{$\ket{0}_1$} & \gate{\H} & \qw & \qw & \ctrl{1} & \qw & \qw & \gate{\H} & \qw & \qw & \rstick{$\bra{0}_1$} \\ \lstick{$\ket{0}_2$} & \qw & \gate{\H} & \qw & \control{} & \ctrl{1} & \gate{\T} & \qw & \gate{\H} & \qw & \rstick{$\bra{0}_2$} \\ \lstick{$\ket{0}_3$} & \qw & \qw & \gate{\H} & \qw & \control{} & \qw & \qw & \qw & \gate{\H} & \rstick{$\bra{0}_3$}
\end{quantikz}
\caption{Running example: a $3$-qubit circuit $C$ over $\{\H,\T,\CZ\}$ with the all-zeros input and output basis states pinned.}
\label{fig:ex-circuit}
\end{figure}

\begin{example}[Running example  circuit]\label{ex:circuit}
We illustrate our constructions using the $3$-qubit circuit $C$ shown in \autoref{fig:ex-circuit}: an initial layer of $\H$ on every qubit (positions $1$--$3$), then $\CZ_{1,2}$ and $\CZ_{2,3}$ (positions $4$, $5$), a single $\T$ on qubit $2$ (position $6$), and a final layer of $\H$ on every qubit (positions $7$--$9$).
Each qubit has a Hadamard depth of two.
Six $\H$ gates, one $\T$, and two $\CZ$'s give nine gates total.
Our running target is the amplitude $\braket{000|C|000}$.
\end{example}

\subsection{Quadratic sums-of-powers for circuits}\label{sec:sop}\label{subsec:SOP_graph}
In \autoref{sec:quantum-c}, we defined strong simulation as the task of evaluating the amplitude $\braket{\bz|C|\by}$.
We now recall how, for circuits over $\{\H,\T,\CZ\}$, this amplitude is a quadratic \emph{sum-of-powers} (SOP)~\cite{WangChengYuanJi25,DBLP:journals/corr/abs-2510-06775}: a scalar sum of roots of unity over Feynman paths~\cite{Dawson+2005,Amy19PathSum}.
The SOP and the SOP variable graph $G_C$ it induces are the central objects of the paper.
Deriving them is the only prerequisite for the algorithm of \autoref{sec:dp}, which follows immediately and needs nothing else about circuits: it evaluates the general quadratic SOP of \autoref{eq:sop} on any graph.
How this compares with the tensor-network route, and which further gate sets it covers, are taken up afterwards, in \autoref{sec:treewidth} and \autoref{sec:gadget-universality}.

Notation from now on: $\Z_d=\Z/d\Z$ denotes the integers modulo $d$, and $\omega_d=e^{2\pi i/d}$ for any positive integer $d$.
For a set $V$, $\{0,1\}^V$ denotes the set of assignments $V\to\{0,1\}$.

We repeatedly use the computational-basis resolution of the identity
\(
I=\sum_{\bx\in\{0,1\}^q}\ket{\bx}\bra{\bx}.
\) For a circuit $C=U_m\cdots U_1$, inserting this identity between consecutive gates gives
\[
\langle \bz|C|\by\rangle= \sum_{\by_1,\ldots,\by_{m-1}} \prod_{j=1}^m \langle \by_j|U_j|\by_{j-1}\rangle .
\]
Here $\by_0=\by$ and $\by_m=\bz$.

\paragraph{Local gate factors.}
For $p,p'\in\{0,1\}$, let the Kronecker delta $\delta_{p,p'}$ be $1$ if $p=p'$ and $0$ otherwise.
It enforces equality of wire values.
For the gate set $\{\H, \T, \CZ\}$, every nonzero local matrix entry is a power of $\omega_8 = e^{2\pi i / 8}$, up to the $\frac1{\sqrt{2}}$ normalization carried by each Hadamard.
For $p, p', p_1, p_2, p'_1, p'_2 \in \{0,1\}$,
\begin{equation}\label{eqn:deltas}
\begin{aligned}
  \braket{p | \H | p'} &{}= 2^{-1/2} (-1)^{pp'} = 2^{-1/2} \omega_8^{4pp'},\\
  \braket{p | \T | p'} &{}= \delta_{p,p'} \omega_8^{p}, \text{ and }\\
  \braket{p_1 p_2 | \CZ | p'_1 p'_2} &{}= \delta_{p_1, p'_1} \delta_{p_2, p'_2} \omega_8^{4 p_1 p_2}.
\end{aligned}
\end{equation}
Thus a Hadamard contributes a quadratic phase $4pq$ and a factor $\frac1{\sqrt{2}}$, a $\T$-gate a unary phase $p$, and a $\CZ$-gate a quadratic phase $4 p_1 p_2$.
All exponents are taken modulo $8$.

\paragraph{Path variables.}
The Kronecker deltas in the $\T$ and $\CZ$ entries (see \autoref{eqn:deltas}) force each summand to vanish unless the value on every affected wire is unchanged across the gate, so the incoming and outgoing indices across such a diagonal gate may be identified.
A Hadamard has no such delta, so the values immediately before and after it remain independent.
Equivalently, cut each qubit wire at every Hadamard: on wire $i$ this yields $h_i+1$ maximal segments, where $h_i$ is the number of Hadamards on wire $i$, carrying Boolean variables $x_{i,0},\ldots,x_{i,h_i}$.
We write $x_s$ for the variable of a segment $s$, so $x_{i,j}=x_s$ when $s$ is the $j$th segment of wire $i$, and a gate lies \emph{on} the segment of its wire that contains its position.
Writing $S$ for the set of all segment variables and reading off the contributions above,
\begin{itemize}
  \item a $\T$-gate on segment $s$ contributes $x_s$;
  \item a $\CZ$-gate touching segments $s, t$ contributes $4 x_s x_t$;
  \item a Hadamard from segment $s$ to its successor $s'$ contributes
$4 x_s x_{s'}$ and one factor $\frac1{\sqrt{2}}$.
\end{itemize}
Summing these local exponents gives a quadratic phase polynomial
\[
f(\bx) = \sum_{\T \text{ on } s} x_s + 4 \sum_{\H : s \to s'} x_s x_{s'} + 4 \sum_{\CZ \text{ on } s, t} x_s x_t \pmod{8},
\]
with all sums counted with multiplicity and all coefficients reduced modulo $8$.
This polynomial records only the exponent of $\omega_8$.
The normalization factors $\frac1{\sqrt{2}}$ contributed by Hadamard gates are collected separately, as will become clear shortly.

\paragraph{Pinning the boundary.}
The next two steps construct the SOP variable graph $G_C=(V,E)$ from the phase polynomial $f$: pinning the boundary segments leaves the unpinned variables as the vertex set $V$, and collecting the surviving quadratic terms gives the edge set $E$.
To compute $\braket{\bz | C | \by}$, let $s_{\mathrm{in}}(i)$ and $s_{\mathrm{out}}(i)$ be the first and last segments on wire $i$.
Pin $x_{s_{\mathrm{in}}(i)}=y_i$, and, when $h_i\ge1$, also pin the distinct last segment $x_{s_{\mathrm{out}}(i)}=z_i$.
If $h_i=0$, the sole segment is fixed by $y_i$ and the output constraint contributes a Kronecker delta.
Collect these constraints in
\[
\delta_C(\by,\bz)=\prod_{i:h_i=0}\delta_{y_i,z_i}.
\]
Thus $\delta_C$ is zero exactly for inconsistent pins on a Hadamard-free wire.
Let $V\subseteq S$ be the remaining segment variables and write $f|_{\by,\bz}$ for the substituted polynomial.

\paragraph{The pinned SOP.}
Collecting terms in $f|_{\by, \bz}$, the constant terms combine into a single phase $c \in \Z_8$, the surviving one-variable terms become coefficients $b_v \in \Z_8$ (the $\T$-gate phases together with any quadratic interaction having exactly one pinned endpoint).
We put an edge $\{u,v\}$ in $E$ exactly when, after substituting the pinned boundary variables, the monomial $x_ux_v$ has coefficient $4$ modulo $8$.
Parallel quadratic terms are added modulo $8$, so two identical terms cancel.
Writing $f_C := f|_{\by, \bz}$,
\begin{gather*}
  f_C(\bx)
  = c + \sum_{v \in V} b_v\, x_v + 4 \sum_{\{u,v\} \in E} x_u x_v
  \pmod{8},\\
  \braket{\bz | C | \by}
  = R_C^{-1}\delta_C(\by,\bz)
    \sum_{\bx \in \{0,1\}^V} \omega_8^{f_C(\bx)},
\end{gather*}
with $R_C = (\sqrt{2})^{m_\H}$ and $m_\H$ the number of Hadamards in $C$.
We call $G_C = (V, E)$ the \emph{SOP variable graph} of the pinned amplitude $(C, \by, \bz)$.\footnote{The pinning $(\by,\bz)$, the import of the gates, and any rewriting applied before it all determine $G_C$.
Every width below is therefore a width of the pinned SOP instance that a given import produces.} \hfill \lean{Formal/Paper.lean}{amplitude_eq_sop_normalized}

\paragraph{General quadratic SOP.}
We now state the SOP and its associated counting problem in their general form, decoupled from any particular circuit.
Let $G=(V,E)$ be a graph, $r$ an even positive integer, $c\in\Z_r$, $b_v\in\Z_r$ for every $v\in V$, and edge coefficient $\eta=r/2$.
Consider the unnormalized quadratic SOP
\begin{equation}\label{eq:sop}
\begin{aligned}
S(f)
={}&
\sum_{\bx\in\{0,1\}^V}\omega_r^{f(\bx)}, \\
f(\bx)={} &c+\sum_{v\in V}b_vx_v+\eta\sum_{\{u,v\}\in E}x_ux_v\pmod r,
\end{aligned}
\end{equation}
where $\omega_r=e^{2\pi i/r}$ is the principal primitive $r$th root of unity.
Throughout the paper the edge coefficient is fixed as $\eta=r/2$.
This is the value produced by the construction above for $\{\H,\T,\CZ\}$ (where $r=8$ and $\eta=4$), and it is what makes each cross-term contribute only a sign $(-1)^{x_ux_v}$.

\begin{definition}[Quadratic \SopCount]\label{def:sopcount}\lean{Formal/Core/SopInstance.lean}{N}
Given an even integer $r$, a graph $G=(V,E)$, and a polynomial $f$ as in \autoref{eq:sop}, the \SopCount problem asks for all counts
\[
N_j=\#\{\bx\in\{0,1\}^V\mid f(\bx)=j\pmod r\}, \qquad j\in\Z_r.
\]
We call the tuple $(N_0,\ldots,N_{r-1})$ the \emph{residue histogram} of $f$: the histogram, over the $r$ residue classes, of the values $f(\bx)$ as $\bx$ ranges over $\{0,1\}^V$.
\end{definition}

\noindent The unnormalized SOP value is recovered from it via
\begin{align}
S(f)=\sum_{j=0}^{r-1}N_j\omega_r^j. \label{eq:sopvalue}
\end{align}

For a circuit $C$ over $\{\H,\T,\CZ\}$ this is the instance with $r=8$ and $G=G_C$, the coefficients $b_v$ and constant $c$ being those identified above.
The physical amplitude is obtained afterward as
\[
\braket{\bz|C|\by}=R_C^{-1}\delta_C(\by,\bz)S(f_C),
\]
where $R_C=(\sqrt2)^{m_\H}$ is the product of the Hadamard normalizations.

\paragraph{Parameters.}
An instance is measured by its number of variables $n=|V|$ and its modulus $r$.
The circuit it came from, when there is one, has $q$ qubits and $m_\H$ Hadamards.
Complexities are stated on the instance rather than on the circuit, since circuit size is not fixed until an elementary gate set is.
A decomposition of $G$ is supplied with the instance, and $k$ denotes the width of that supplied decomposition.
The magic count $\tau$ is the number of magic vertices below a cut, the quantity the stabilizer-rank bound $\mathrm{sr}(\tau)$ of \autoref{sec:stabjoin} is charged on.

\begin{figure}[!tbp]
\centering
\begin{tikzpicture}[
  scale=0.8,
  var/.style={draw, circle, fill=green!18, minimum size=0.9cm, inner sep=0pt, font=\scriptsize}
]
\node[var] (v1) at (0, 0)    {$x_{1,1}$};
\node[var] (v2) at (2.4, 0)  {$x_{2,1}$};
\node[var] (v3) at (4.8, 0)  {$x_{3,1}$};
\draw (v1) -- node[above, font=\scriptsize] {$-1$} (v2);
\draw (v2) -- node[above, font=\scriptsize] {$-1$} (v3);
\node[font=\scriptsize] at (2.4, -0.85) {$b_{x_{2,1}}=1$};
\end{tikzpicture}
\caption{The SOP variable graph $G_C$ for the circuit $C$ of \autoref{ex:circuit} after pinning $\ket{000}\to\ket{000}$.
Every edge carries the constant cross-term sign $\omega_8^{4}=-1$, contributed as $(-1)^{x_ux_v}$.
Vertices carry unary phases $\omega_8^{b_v}$, the only nontrivial one being $b_{x_{2,1}}=1$ from the $\T$ gate.}
\label{fig:ex-gc}
\end{figure}

\begin{example}[Deriving $G_C$]\label{ex:gc}
Continuing \autoref{ex:circuit}, cutting each wire at its Hadamard gates yields one segment variable per wire segment.
Since each qubit has Hadamard depth $h_i=2$, every wire contributes three variables $x_{i,0},x_{i,1},x_{i,2}$, with $x_{i,0}$ and $x_{i,2}$ fixed by the input and output basis states.
Pinning $\ket{000}\to\ket{000}$ therefore leaves $V=\{x_{1,1},x_{2,1},x_{3,1}\}$ as the only free variables.
Quadratic cross-terms come from the two $\CZ$ gates: $\CZ_{1,2}$ acts on $x_{1,1}$ and $x_{2,1}$ (the current variables on the two wires at position $4$), and $\CZ_{2,3}$ acts on $x_{2,1}$ and $x_{3,1}$.
Each such edge carries the constant cross-term sign $\omega_8^{4}=-1$, contributed as $(-1)^{x_ux_v}$.
The $\T$ on qubit~$2$ at position~$6$ contributes the unary phase coefficient $b_{x_{2,1}}=1$, so the phase at $x_{2,1}=1$ is $e^{i\pi/4}=\omega_8$.
The resulting SOP variable graph $G_C$ is the path of~\autoref{fig:ex-gc}.
\end{example}

\section{A rank-width algorithm for quadratic sums-of-powers}\label{sec:dp}
This section gives the paper's main algorithm: an exact evaluation of a quadratic SOP that is fixed-parameter tractable in the \emph{rank-width} of its variable graph.
It is stated for the general instance of \autoref{eq:sop} and uses nothing about circuits beyond it.
For a circuit-derived instance, the graph is the $G_C$ of the previous section.
Since rank-width can be much smaller than treewidth or linear rank-width, this is what allows the SOP route to improve the corresponding structural guarantees for the circuit families we construct in \autoref{subsec:separation-linear-tensor}.

Everything rests on one count.
Two assignments to the variables below a cut behave the same way above it when every outside variable sees the same parity of selected neighbors, so a cut of rank $k$ admits at most $2^k$ interaction patterns.
A branching join has to pair the patterns of two children, which is where the base $4$ comes from, whereas a linear layout extends one side by a single variable, which is where the base $2$ comes from.

\begin{figure}[!tb]
\centering
\begin{tikzpicture}[
  scale=0.8,
  var/.style={draw, circle, fill=green!18, minimum size=0.9cm, inner sep=0pt, font=\scriptsize}
]
\node[var] (v1) at (0, 0)    {$x_{1,1}$};
\node[var] (v2) at (2.4, 0)  {$x_{2,1}$};
\node[var] (v3) at (4.8, 0)  {$x_{3,1}$};
\draw (v1) -- (v2);
\draw (v2) -- (v3);
\draw[dashed, thick] (3.6, -0.95) -- (3.6, 0.95);
\node[font=\scriptsize] at (1.2, 1.15) {$X$};
\node[font=\scriptsize] at (4.8, 1.15) {$\overline X$};
\node[font=\scriptsize, align=center] at (3.6, -1.45) {cut-rank $1$};
\end{tikzpicture}
\caption{A cut of the running graph of \autoref{fig:ex-gc}.
One edge crosses it, so the cut matrix $\bA[X,\overline X]$ has rank $1$ and the four assignments to $X$ fall into two classes, told apart by the parity $x_{2,1}$ that the outside vertex $x_{3,1}$ sees.
The dynamic program keeps one table entry per class rather than per assignment.}
\label{fig:cut-signature}
\end{figure}
\autoref{sec:sopbyrank} computes the full residue histogram of \autoref{def:sopcount} by processing a rooted rank-decomposition of $G_C$ from the leaves to the root over $\Z$: every table entry counts assignments, and $\omega_r$ enters nowhere until  %
the finished residue histogram is turned into $S(f)$ via \autoref{eq:sopvalue}.
\autoref{sec:fourier} applies a discrete Fourier transform to that table and runs the same recursion over $\Z[\omega_r]$ for $r$ Fourier modes.
This saves a factor $r$ in the count of ring operations, at the price of leaving the integers.
Writing $k$ for the width of the rank-decomposition (\autoref{tab:route-comparison}), the branching DP runs in base $4^k$ per Fourier mode, while the guarantee of a linear layout removes the joins and runs in base $2^k$ (\autoref{thm:linear-layout-fourier}).

Both rings are exact, and for fixed $r$ both are representable in $\poly(n)$ bits.
A residue table holds $r$ counts per 
row,
each at most $2^n$.
The transformed table contains,  per Fourier mode, one element of $\Z[\omega_r]$, which can be represented by $\varphi(r)$ integers of that size in the power basis $1,\omega_r,\ldots,\omega_r^{\varphi(r)-1}$, where $\varphi$ is Euler's totient.
Exactness therefore does not separate the two rings.
At $r=8$ the Fourier mode tables lie in the ring over which Clifford$+\T$ circuits are already synthesized exactly~\cite{GilesSelinger13}, and the dyadic-cyclotomic ring $\Z[\omega_r,1/\sqrt2]$ of \autoref{thm:gadget-universality} covers the final normalization as well.
The two rings differ in the cost of a single operation.
Combining two residue tables at a fixed pair of signatures is a cyclic convolution over $\Z_r$, which is $r^2$ integer multiplications.
One product in $\Z[\omega_r]$ is a polynomial product modulo the cyclotomic $\Phi_r$, which is $\varphi(r)^2$.
An operation count is therefore a time bound only up to a factor that depends on the ring, and the factor $r$ between \autoref{thm:dp-correct} and \autoref{thm:fourier-speedup} is a change of unit rather than a saving in bit operations.
\autoref{cor:histogram-one-mode} makes it a saving when $r$ is a power of two.
Large moduli are the common case in practice: most of the benchmarks in \autoref{sec:experiments} admit no exact finite-modulus import, and those carry moduli between $4.0\mathbin{\cdot}10^{10}$ and $9.8\mathbin{\cdot}10^{14}$ (\autoref{tab:experiment-coverage}).
Both rings stay large there, since $\varphi(r)$ remains within a factor of two of $r$ across all of them, so what makes the regime tractable is evaluating a single Fourier mode rather than all $r$ of them (\autoref{cor:single-amplitude}).

\subsection{Counting over the integers by rank-decomposition}
\label{sec:sopbyrank}
The algorithm processes a rank-decomposition of $G$ from the leaves upward, keeping at each vertex a compact table that records how the partial assignments below it interact with the rest of the graph.
Each cut of a width-$k$ decomposition has rank at most $k$, so by the count above each table holds at most $2^k$ patterns.

\paragraph{Setup.}
If $V=\varnothing$, the empty assignment is unique, so $N_c=1$, all other counts vanish, and $S(f)=\omega_r^c$.
We handle this case directly.
Hence assume $V\ne\varnothing$.
We work with a rooted rank-decomposition $(\mathcal{T},\mu)$ of $G$, where $\mathcal T$ is a tree of maximum degree $3$ whose leaves are in bijection with $V$ via $\mu$.
Its root $r_{\mathcal T}$ is the sole leaf when $|V|=1$ and otherwise an internal vertex, obtained if necessary by subdividing an edge of an unrooted decomposition.
Suppressing the degree-two internal vertices of that unrooted decomposition changes no cut and leaves every internal vertex of degree exactly three, so after rooting each internal vertex of $\mathcal T$ has exactly two children, as \autoref{alg:rw-dp} assumes.
For a vertex $u$ of $\mathcal T$, let $X_u\subseteq V$ be the set of vertices corresponding to leaves of the subtree rooted at $u$, and let $\overline X_u=V\setminus X_u$.
At the root, $X_{r_{\mathcal T}}=V$ and $\overline X_{r_{\mathcal T}}=\varnothing$.
At a leaf $u$ with $\mu(u)=v$, $X_u=\{v\}$.
At each vertex $u$, the dynamic-programming (DP) table records how assignments inside $X_u$ look to the outside.

Throughout this section we identify $\{0,1\}$ with $\F_2$, so an assignment $\bz\in\{0,1\}^{X_u}$ is at the same time a column vector over $\F_2$, with $\bz^\top$ its row transpose.
Parities are computed over $\F_2$ and residues over $\Z_r$.
Here $\bA$ denotes the adjacency matrix, over $\F_2$, of the graph $G=(V,E)$ of the SOP instance being evaluated, and $\bA[S,T]$ its submatrix with rows $S$ and columns $T$, as fixed in \autoref{sec:preliminaries}.
Thus the inputs of the algorithm determine $\bA$: its entries are the coefficients ($0$ or $1$, before scaling by $\eta$) of the quadratic terms $x_ux_v$ of $f$ in \autoref{eq:sop}.
For $\bz\in\{0,1\}^{X_u}$, define the boundary signature
\[
\sigma_u(\bz)=\bz^\top\bA[X_u,\overline X_u]\in\F_2^{\overline X_u}.
\]
This vector records, for each outside vertex, the parity of its selected neighbors inside $X_u$.
Define the internal residue
\[
\phi_u(\bz)= \sum_{v\in X_u} b_vz_v + \eta\sum_{\{v,v'\}\in E,\ v,v'\in X_u} z_vz_{v'} \pmod r.
\]
Recall that $\eta=r/2$ throughout, with $r=8$ for $\{\H,\T,\CZ\}$.
The constant $c$ is not included in $\phi_u$ and is added back when we deal with the root.

The table at vertex $u$ is stored sparsely, with keys $(\bsigma,s)$ for a boundary signature $\bsigma$ that equals $\sigma_u(\bz)$ for some $\bz\in\{0,1\}^{X_u}$, and $s\in\Z_r$.
\autoref{alg:rw-dp} builds $D_u$ bottom-up through the leaf initialization and join update given below.
The point of that construction, proved as \autoref{thm:dp-correct}, is that the resulting entries are the counts
\[
D_u[\bsigma,s] = \#\{\bz\in\{0,1\}^{X_u}\mid \sigma_u(\bz)=\bsigma,\ \phi_u(\bz)=s\}.
\]
Although $\sigma_u(\bz)$ is written as a vector in $\F_2^{\overline X_u}$, it always lies in the image of the linear map
\[
\{0,1\}^{X_u}\rightarrow \F_2^{\overline X_u}, \qquad \bz\mapsto \bz^\top\bA[X_u,\overline X_u],
\]
namely the row space of $\bA[X_u,\overline X_u]$ over $\F_2$.
For every nonroot vertex $u$, the edge from $u$ to its parent in the rooted rank-decomposition induces the cut $(X_u,\overline X_u)$.
Since the decomposition has width $k$, $ \rank_{\F_2}\bA[X_u,\overline X_u]\le k $ and therefore at most $2^k$ boundary signatures can occur at $u$.
At the root, $\overline X_u=\varnothing$, so there is only the empty signature.
Thus each table has at most $r2^k$ nonzero states.

\paragraph{Leaves.}
If $u$ is a leaf corresponding to vertex $v$, then $X_u=\{v\}$, and the assignments $z_v=0$ and $z_v=1$ contribute
\[
D_u[0,0]\mathrel{+}=1, \qquad D_u[\bA[v,V\setminus\{v\}],b_v]\mathrel{+}=1.
\]
The use of $\mathrel{+}=$ matters: the two states may coincide.

\paragraph{Joins.}
Let $t$ be an internal vertex with children $L$ and $R$.
Write $X_t=X_L\uplus X_R$ and $Y=\overline X_t=V\setminus X_t$, so that $\overline X_L=X_R\cup Y$ and $\overline X_R=X_L\cup Y$.
Take a left state $(\balpha,p)$ and a right state $(\bbeta,q)$, with $\balpha\in\F_2^{X_R\cup Y}$ and $\bbeta\in\F_2^{X_L\cup Y}$, and let $\ba\in\{0,1\}^{X_L}$ and $\bb\in\{0,1\}^{X_R}$ be assignments with these signatures.
The new quadratic terms are exactly the edges between $X_L$ and $X_R$.
Their parity is
\[
\chi(\balpha,\bbeta)=\ba^\top\bA[X_L,X_R]\bb\in\F_2.
\]

\begin{lemma}\label{lem:chi-well-defined}\lean{Formal/Core/Signature.lean}{chi_well_defined}
The value $\chi(\balpha,\bbeta)$ depends only on the signatures $\balpha$ and $\bbeta$, and not on the chosen assignments $\ba$ and $\bb$.
\end{lemma}
\begin{proof}
Suppose $\ba'$ is another assignment in $\{0,1\}^{X_L}$ with the same left signature, $\sigma_L(\ba')=\sigma_L(\ba)=\balpha$.
Then $(\ba-\ba')^\top\bA[X_L,V\setminus X_L]=\bzero$ over $\F_2$, and since $X_R\subseteq V\setminus X_L$, restricting to $X_R$ gives $(\ba-\ba')^\top\bA[X_L,X_R]=\bzero$.
Multiplying by $\bb$ yields $\ba^\top\bA[X_L,X_R]\bb={\ba'}^\top\bA[X_L,X_R]\bb$, so changing the left representative leaves the cross term fixed.

The same argument applies on the right.
If $\bb'$ satisfies $\sigma_R(\bb')=\sigma_R(\bb)=\bbeta$, then $(\bb-\bb')^\top\bA[X_R,V\setminus X_R]=\bzero$, and since $X_L\subseteq V\setminus X_R$ and $\bA$ is symmetric, $\bA[X_L,X_R](\bb-\bb')=\bzero$.
Therefore $\ba^\top\bA[X_L,X_R]\bb=\ba^\top\bA[X_L,X_R]\bb'$, and $\chi(\balpha,\bbeta)$ is independent of all choices of representatives.
\end{proof}

The parent signature and residue are
\[
\bgamma=\balpha|_Y+\bbeta|_Y\in\F_2^Y, \qquad s=p+q+\eta\chi(\balpha,\bbeta)\pmod r.
\]
For the residue, the cross-sum $\eta \sum_{v,v'} z_v z_{v'}$ over $\{v,v'\} \in E$ with $v \in X_L$ and $v' \in X_R$ equals $\eta$ times the number of selected cut edges.
Because $\eta=r/2$, this value modulo $r$ depends only on the parity of that number, namely $\chi(\balpha,\bbeta)=\ba^\top\bA[X_L,X_R]\bb\bmod 2$, so the cross-sum reduces to $\eta\chi(\balpha,\bbeta)\pmod r$.
The join does not create an unbounded set of parent signatures.
If $\balpha$ and $\bbeta$ are the signatures of assignments $\ba\in\{0,1\}^{X_L}$ and $\bb\in\{0,1\}^{X_R}$, then
\[
\bgamma = \ba^\top\bA[X_L,Y]+\bb^\top\bA[X_R,Y] = (\ba,\bb)^\top\bA[X_t,Y] = \sigma_t(\ba,\bb).
\]
Thus, every parent signature produced by the join lies in the row space of $\bA[X_t,Y]$.
Since $Y=\overline X_t$, the cut $(X_t,Y)$ is again a cut induced by an edge of the rooted rank-decomposition unless $t$ is the root.
Therefore this row space has size at most $2^k$ at every nonroot internal vertex, and size $1$ at the root.

The join update is $D_t[\bgamma,s]\mathrel{+}=D_L[\balpha,p]D_R[\bbeta,q]$.

\paragraph{Representatives and precomputed maps.}
The update above is independent of representative assignments, but an implementation still needs representatives to compute the join tables.
For each occurring signature $\bsigma$ at a vertex $u$, keep one representative $\operatorname{rep}_u(\bsigma)\in\{0,1\}^{X_u}$ with $\sigma_u(\operatorname{rep}_u(\bsigma))=\bsigma$.
Leaves have the representatives $0$ and $1$ when the corresponding signatures occur.
At an internal vertex $t$ with children $L,R$, precompute for each occurring pair $(\balpha,\bbeta)$ the parent signature and the crossing parity,
\[
\Gamma_t(\balpha,\bbeta)=\balpha|_Y+\bbeta|_Y, \qquad \chi_t(\balpha,\bbeta) = \operatorname{rep}_L(\balpha)^\top \bA[X_L,X_R] \operatorname{rep}_R(\bbeta).
\]
When a parent signature is first produced, store the concatenation of the two child representatives as its representative.
By \autoref{lem:chi-well-defined}, all choices with the same signatures give the same crossing parity.

\begin{algorithm}[t]
\caption{Rank-decomposition DP for quadratic \SopCount. \lean{Formal/Core/DP.lean}{SopInstance.Dalg}}\label{alg:rw-dp}
\renewcommand{\algorithmicrequire}{\textbf{Input:}}
\renewcommand{\algorithmicensure}{\textbf{Output:}}
\begin{algorithmic}[1]
\Require Quadratic SOP polynomial $f(\bx)=c+\sum_v b_vx_v+\eta\sum_{\{u,v\}\in E}x_ux_v$ and a rooted rank-decomposition of $G=(V,E)$ of width $k$.
\Ensure Counts $N_j$ for $j\in\Z_r$.
\For{each leaf $u$ with vertex $v$}
    \State $D_u\gets$ empty sparse table keyed by occurring pairs $(\bsigma,s)$
    \State $D_u[0,0]\mathrel{+}=1$
    \State $D_u[\bA[v,V\setminus\{v\}],b_v\bmod r]\mathrel{+}=1$
\EndFor
\For{each internal vertex $t$ in post-order}
    \State Let $L,R$ be the children of $t$ and $Y\gets\overline X_t$.
    \State $D_t\gets$ empty sparse table keyed by occurring pairs $(\bsigma,s)$
    \ForAll{$(\balpha,p)$ with $D_L[\balpha,p]>0$}
        \ForAll{$(\bbeta,q)$ with $D_R[\bbeta,q]>0$}
            \State $\bgamma\gets\balpha|_Y+\bbeta|_Y$
            \State $\chi\gets\chi(\balpha,\bbeta)$
            \State $s\gets(p+q+\eta\chi)\bmod r$
            \State $D_t[\bgamma,s]\mathrel{+}=D_L[\balpha,p]\cdot D_R[\bbeta,q]$
        \EndFor
    \EndFor
\EndFor
\State \Return $\big(D_{r_{\mathcal T}}[\varnothing,j-c\bmod r]\big)_{j\in\Z_r}$
\end{algorithmic}
\end{algorithm}

\begin{theorem}[Correctness and cost]\label{thm:dp-correct}
The table computed by \autoref{alg:rw-dp} satisfies, for every vertex $u$, signature $\bsigma$, and residue $s$,
\[
D_u[\bsigma,s] = \#\{\bz\in\{0,1\}^{X_u}\mid \sigma_u(\bz)=\bsigma,\ \phi_u(\bz)=s\}.
\]

\vspace{-2em}
\hfill\lean{Formal/Core/DPCorrect.lean}{Dalg_eq_Dspec}
\vspace{1em}

Consequently, \autoref{alg:rw-dp} returns the counts $N_j$.
On a rank-decomposition of width $k$, each table has at most $r2^k$ occupied entries, and the total number of leaf initializations and joined child-state pairs is at most $n\bigl(2+(r2^k)^2\bigr)$.
All counts are therefore computed in time $O(r^24^k\poly(n))$, which for fixed $r$ is $2^{O(k)}\poly(n)$, and in storage $O(r2^k\poly(n))$.\lean{Formal/Paper.lean}{rank_dp_spec}
\end{theorem}

\begin{proof}
The proof is by induction over the rooted decomposition tree.
The leaf case is exactly the two assignments $z_v=0$ and $z_v=1$.

For the join step, every assignment to $X_t$ decomposes uniquely as $(\ba,\bb)$ with $\ba\in\{0,1\}^{X_L}$ and $\bb\in\{0,1\}^{X_R}$.
The boundary signature satisfies
\[
\sigma_t(\ba,\bb) =(\ba,\bb)^\top\bA[X_t,Y] =\ba^\top\bA[X_L,Y]+\bb^\top\bA[X_R,Y] =\balpha|_Y+\bbeta|_Y.
\]
For the residue, the only terms not already counted inside the two children are edges crossing from $X_L$ to $X_R$, and the integer count of such selected crossing edges satisfies $m\bmod 2 = \ba^\top\bA[X_L,X_R]\bb$.
Since $\eta=r/2$, we have $\eta m\equiv \eta(m\bmod 2)\pmod r$ for every integer $m$, so
\[
\phi_t(\ba,\bb)=\phi_L(\ba)+\phi_R(\bb)+\eta\chi(\balpha,\bbeta)\pmod r.
\]
By \autoref{lem:chi-well-defined}, this depends only on the two child states.
The child variable sets are disjoint, so the number of pairs of assignments giving the two child states is $D_L[\balpha,p]D_R[\bbeta,q]$.
The join rule adds exactly this quantity to the correct parent state.
Summing over child-state pairs proves the invariant.

At the root $X_{r_{\mathcal T}}=V$ and $\overline X_{r_{\mathcal T}}=\varnothing$.
Also $f(\bx)=c+\phi_{r_{\mathcal T}}(\bx)\pmod r$.
Hence $f(\bx)=j$ iff $\phi_{r_{\mathcal T}}(\bx)=j-c$, and the returned vector is $(N_j)_{j\in\Z_r}$.

For the cost, at most $2^k$ boundary signatures occur at every vertex, as established above, so a sparsely stored table has at most $r2^k$ nonzero states.
A naive join scans all pairs of child states, at most $r^24^k$ of them.
Each pair performs one parent update, and restricting signatures, adding them, and reading the precomputed cross term $\chi$ cost only polynomial time in $n$ under explicit bit-vector representations and representatives.
The decomposition tree has linearly many vertices, and keeping all its tables and representatives gives the stated storage.
\end{proof}

If child tables are discarded after use, a depth-first evaluation uses $O(r2^k\operatorname{depth}(\mathcal T))$ working table entries, apart from polynomial-size bit-vector data.

\subsection{A faster implementation in the arithmetic model}
\label{sec:fourier}

The residue table of the previous subsection carries a factor $r$ in the join.
Working with the discrete Fourier transform of the residue counts on $\Z_r$ removes it from the operation count, replacing the $r$ integer counts held at each boundary signature by one element of $\Z[\omega_r]$ and bringing the full histogram to $O(r4^k\poly(n))$ ring operations plus one exact $r$-point inverse transform.
Since the two dynamic programs count operations in different rings, that factor is counted in the arithmetic model.
It is unambiguous in the case simulation needs: one amplitude requires only the single Fourier mode $a=1$, and that is the primitive below.
The full residue histogram $N_j$ of \autoref{def:sopcount} is a strictly stronger algebraic object: it is the complete distribution of phase residues, and that distribution determines every Fourier character, namely every scaled sum $S(af)=\sum_{j\in\Z_r}N_j\omega_r^{aj}$.
It is also the one computed entirely over $\Z$.
When $r$ is a power of two, \autoref{cor:histogram-one-mode} recovers it from the single mode $a=1$, so the stronger object costs no more than one amplitude.

\paragraph{Residue counts and their Fourier transform.}
The vector of residue counts $N=(N_0,N_1,\ldots,N_{r-1})$ is a function on the cyclic group $\Z_r$.
Using the characters $\chi_a(j)=\omega_r^{aj}$ for $a,j\in\Z_r$, define the discrete Fourier transform
\begin{equation}\label{eq:residue-dft}
\widehat N(a)=\sum_{j=0}^{r-1}N_j\chi_a(j)
=\sum_{j=0}^{r-1}N_j\omega_r^{aj}.
\end{equation}
This is the usual finite Fourier transform, which, together with its inverse, can be computed by the fast Fourier transform~\cite{CooleyTukey65} (FFT).
The group is $\Z_r$, but the same principle is useful later: instead of storing coefficients indexed by residues, we store evaluations against $\chi_a$.

\autoref{eq:residue-dft} is also exactly an unnormalized SOP evaluation.
Grouping the assignments $\bx$ by their residue $f(\bx)\pmod r$ gives
    \[
\widehat N(a) =\sum_{j=0}^{r-1}\#\{\bx\in\{0,1\}^V\mid f(\bx)=j\pmod r\}\cdot\omega_r^{aj} =\sum_{\bx\in\{0,1\}^V}\omega_r^{a f(\bx)},
    \]
since every assignment $\bx$ contributes to exactly one residue class.
In particular, the unnormalized SOP value satisfies $\widehat N(1)=\sum_{\bx}\omega_r^{f(\bx)}=S(f)$, so SOP evaluation gives one Fourier coefficient and evaluating the scaled polynomials $af$ for all $a\in\Z_r$ gives all of them.
By orthogonality of characters of $\Z_r$, the transform inverts as
\begin{equation}\label{eq:inverse-residue-dft}
N_j=\frac1r\sum_{a=0}^{r-1}\omega_r^{-aj}\widehat N(a),
\qquad j\in\Z_r.
\end{equation}
The division by $r$ is not an operation of $\Z[\omega_r]$.
Each $N_j$ is an integer, so the numerator is formed in $\Z[\omega_r]$ and then divided exactly, coefficient by coefficient in the power basis.
Computing all residue counts and computing all scaled amplitude evaluations are therefore equivalent up to an $r$-point Fourier transform.
Exact inversion costs $O(r^2)$ ring operations when done directly, or $O(r\log r)$ when an exact FFT is available for the chosen representation, so $\poly(n)$ time for fixed $r$.
For a fixed finite gate set, $r$ is a constant and this cost is immaterial.

\paragraph{One mode, three guises.}
A single odd mode is already the object the closest related methods manipulate.
Writing $\lambda_v=\omega_r^{ab_v}$, the next proposition fixes the direct translations.

\begin{proposition}[Direct representation correspondence]\label{prop:four-guises}
For every quadratic SOP in the sign-edge form of \autoref{eq:sop} and every odd $a\in\Z_r$, the mode sum
\[
\widehat N(a) \;=\;\omega_r^{ac}\sum_{\bx\in\{0,1\}^{V}}(-1)^{\sum_{\{u,v\}\in E}x_ux_v}\prod_{v\in V}\lambda_v^{x_v}
\]
is a complex two-spin partition function on $G$ with edge coupling $-1$ and vertex fields $\lambda_v$.
Equivalently, it is
\begin{enumerate}[label=(\roman*),ref=(\roman*),leftmargin=2.4em,itemsep=1pt,topsep=2pt]
  \item\label{guise-graphstate} the graph-state overlap $\omega_r^{ac}\,2^{|V|/2}\bigl(\bigotimes_{v\in V}(\bra0+\lambda_v\bra1)\bigr)\ket{G}$, and
  \item\label{guise-zx} a closed graph-like ZX scalar whose Hadamard edges are those of $G$ and whose spider phases are $\arg\lambda_v$.
\end{enumerate}
\end{proposition}

\begin{proof}
Since $\eta=r/2$ and $a$ is odd, each quadratic term of $af$ contributes $(-1)^{ax_ux_v}=(-1)^{x_ux_v}$, which gives the display.
For \ref{guise-graphstate}, expand the graph state in the computational basis, $\ket G=2^{-|V|/2}\sum_{\bx}(-1)^{\sum_{\{u,v\}\in E}x_ux_v}\ket{\bx}$, and pair it with the covector whose factors have coefficients $1,\lambda_v$, taken without complex conjugation.
For \ref{guise-zx}, read each variable as a $Z$-spider with phase $\arg\lambda_v$ and each edge as a Hadamard edge.
\end{proof}

These translations preserve the interaction graph, and hence its rank-width, exactly.
Subsequent circuit or ZX rewrites may change it.
\appref{app:kk} matches the normalization conventions.

\paragraph{Per-mode dynamic program.}
This suggests running the rank-decomposition DP \emph{one Fourier mode at a time}: for each $a\in\Z_r$, track a single element of $\Z[\omega_r]$ per boundary signature, computing $\widehat N(a)$ in a single pass, and recover all $N_j$ at the root via~\autoref{eq:inverse-residue-dft}.
Concretely, for each $a\in\Z_r$ define
\[
A_u^{(a)}[\bsigma] = \sum_{\substack{\bz\in\{0,1\}^{X_u}\\ \sigma_u(\bz)=\bsigma}} \omega_r^{a\phi_u(\bz)}.
\]

\begin{example}[The running cut table]\label{ex:cut-table}
For the path in \autoref{ex:gc}, take $X=\{x_{1,1}\}$ and order the outside coordinates as $(x_{2,1},x_{3,1})$.
The cut matrix is $(1\ 0)$ and has rank $1$, so the two assignments to $x_{1,1}$ have signatures $(0,0)$ and $(1,0)$.
Neither assignment accumulates an internal or unary phase, hence
\[
A_X^{(a)}[(0,0)]=1,
\qquad
A_X^{(a)}[(1,0)]=1,
\]
and every other entry is zero.
\end{example}

\begin{lemma}[Mode tables are character images]\label{lem:mode-from-counts}
\lean{Formal/Core/Fourier.lean}{Aalg_root}
For every vertex $u$, signature $\bsigma$, and $a\in\Z_r$,
\[
A_u^{(a)}[\bsigma]=\sum_{s\in\Z_r}D_u[\bsigma,s]\,\omega_r^{as}.
\]
The mode-$a$ leaf initialization and join update are the images of those of \autoref{alg:rw-dp} under the ring homomorphism $\Z[\Z_r]\to\Z[\omega_r]$ carrying a residue $s$ to $\omega_r^{as}$, and at the root $\widehat N(a)=\omega_r^{ac}A_{r_{\mathcal T}}^{(a)}[\varnothing]$.
\end{lemma}

\begin{proof}
Grouping the assignments $\bz\in\{0,1\}^{X_u}$ with $\sigma_u(\bz)=\bsigma$ by their residue $\phi_u(\bz)$ and applying \autoref{thm:dp-correct} gives
\[
\sum_{s\in\Z_r}D_u[\bsigma,s]\,\omega_r^{as} =\sum_{\substack{\bz\in\{0,1\}^{X_u}\\ \sigma_u(\bz)=\bsigma}}\omega_r^{a\phi_u(\bz)} =A_u^{(a)}[\bsigma].
\]
Read the counts at a fixed signature as the element $\sum_sD_u[\bsigma,s]\,\delta_s$ of the group ring $\Z[\Z_r]$.
The join of \autoref{alg:rw-dp} multiplies two such elements and shifts the product by $\eta\chi$, and $\delta_s\mapsto\omega_r^{as}$ is a ring homomorphism $\Z[\Z_r]\to\Z[\omega_r]$, so it carries the leaf and join rules to their mode-$a$ counterparts.
At the root $X_{r_{\mathcal T}}=V$ and $f=c+\phi_{r_{\mathcal T}}$, so $\widehat N(a)=\sum_{\bx}\omega_r^{af(\bx)}=\omega_r^{ac}A_{r_{\mathcal T}}^{(a)}[\varnothing]$.
\end{proof}

The counts are recovered at the root by
\begin{equation}\label{eq:inverse-dft}
N_j=\frac1r\sum_{a=0}^{r-1}\omega_r^{-aj}\widehat N(a)
=\frac1r\sum_{a=0}^{r-1}\omega_r^{a(c-j)}A_{r_{\mathcal T}}^{(a)}[\varnothing].
\end{equation}

For a fixed mode $a$, the leaf initialization and the join become
\[
A_u^{(a)}[0]\mathrel{+}=1, \qquad A_u^{(a)}[\bA[v,V\setminus\{v\}]]\mathrel{+}=\omega_r^{ab_v}, \qquad A_t^{(a)}[\bgamma] \mathrel{+}= A_L^{(a)}[\balpha] A_R^{(a)}[\bbeta]\, \omega_r^{a\eta\chi(\balpha,\bbeta)}.
\]
Since $\eta=r/2$, we have $\omega_r^{a\eta\chi}=\omega_r^{a(r/2)\chi}=(-1)^{a\chi}$.
Thus the interaction kernel depends only on the parity of $a$: even modes have no signed cross term, and odd modes have the same $(-1)^{\chi}$ kernel.
For even modes the absence of the cross term collapses the whole computation:

\begin{proposition}[Even modes factorize]\label{prop:even-modes}\lean{Formal/Stab/Magic.lean}{Nhat_even}
For every even $a$ the mode sum factorizes over the vertices:
\[
\widehat N(a)=\omega_r^{ac}\prod_{v\in V}\bigl(1+\omega_r^{ab_v}\bigr).
\]
\end{proposition}

\begin{proof}
Write $a=2a'$.
Then $a\eta=a'r\equiv0\pmod r$, so every quadratic term vanishes from $af$, leaving $\omega_r^{af(\bx)}=\omega_r^{ac}\prod_v\omega_r^{ab_vx_v}$, over which the sum distributes.
\end{proof}

Every even mode is therefore computable in $O(n)$ ring operations, hence in $\poly(n)$ time, on every graph and at every even modulus, regardless of width.
Only the odd modes need the dynamic program.
Evaluating the even half by the product improves constants, not the coarse bound of \autoref{thm:fourier-speedup}.
\autoref{cor:histogram-one-mode} uses the same vanishing at modulus $r/2$ to replace the whole even half by a single convolution of $O(nr)$ integer additions.

\begin{theorem}[Fourier speedup]\label{thm:fourier-speedup}
\lean{Formal/Paper.lean}{fourier_mode_spec}
Given a rank-decomposition of an $n$-vertex SOP graph $G$ of width $k$, the mode-$a$ DP returns $\widehat N(a)$ and performs at most $n(2+4^k)$ leaf initializations and joined child-signature pairs.
Consequently all counts $N_j$ can be computed in time $O\bigl((r4^k+r^2)\poly(n)\bigr)$.
\end{theorem}

\begin{proof}
Correctness is \autoref{lem:mode-from-counts}.
For each fixed $a\in\Z_r$, the table holds one element of $\Z[\omega_r]$ per boundary signature.
As above, the signatures at a nonroot vertex $u$ form the row space of $\bA[X_u,\overline X_u]$, whose dimension is at most $k$.
At the root there is only the empty signature.
Hence each table has at most $2^k$ states per vertex.
The naive join scans at most $4^k$ pairs of child signatures, so running the DP for all $r$ modes costs $O(r4^k\poly(n))$.
By \autoref{prop:even-modes} the even half needs only $O(rn)$ ring operations of that.
The inverse transform in \autoref{eq:inverse-dft} costs $O(r^2)$ ring operations, which is immaterial when $r$ is fixed.
\end{proof}

\begin{corollary}[Single-amplitude mode]\label{cor:single-amplitude}\lean{Formal/Core/Fourier.lean}{single_amplitude}
Given a rank-decomposition of $G$ of width $k$, the unnormalized SOP value $S(f)=\sum_{\bx}\omega_r^{f(\bx)}$ can be computed in time $O(4^k\poly(n))$ by running only the Fourier-mode DP for $a=1$.
\end{corollary}

\begin{proof}
By \autoref{eq:residue-dft}, the mode $a=1$ is $ \widehat N(1)=\sum_{\bx}\omega_r^{f(\bx)}=S(f). $
The per-mode DP computes it without the other $r-1$ modes or the inverse transform.
\end{proof}

\noindent
For a circuit-derived instance, multiply by $\delta_C(\by,\bz)$ and divide by $R_C$ to obtain the amplitude.

\begin{corollary}[Full histogram at a fixed power-of-two modulus]\label{cor:histogram-one-mode}\lean{Formal/Paper.lean}{histogram_one_mode_spec}
Fix $r=2^m$ with $m\ge1$.
Given a rank-decomposition of $G$ of width $k$, all residue counts $N_j$ are computable in time $O(4^k\poly(n))$, by the mode-$1$ DP of \autoref{cor:single-amplitude} together with $O(nr)$ integer additions.
\end{corollary}

\begin{proof}
Since $r$ is a power of two, $\Phi_r(x)=x^{r/2}+1$.
Hence $1,\omega_r,\ldots,\omega_r^{r/2-1}$ is a $\Z$-basis of $\Z[\omega_r]$ and $\omega_r^{j+r/2}=-\omega_r^{j}$, so splitting \autoref{eq:residue-dft} at $a=1$ into its two halves gives
\[
S(f)=\widehat N(1)=\sum_{j=0}^{r/2-1}\bigl(N_j-N_{j+r/2}\bigr)\,\omega_r^{j}.
\]
The coordinates of $S(f)$ in that basis are therefore the differences $\Delta_j=N_j-N_{j+r/2}$, which \autoref{cor:single-amplitude} produces in time $O(4^k\poly(n))$.
They are read off the representation rather than computed from it.

For the sums, reduce the modulus to $r/2$.
The coefficient $\eta=r/2$ vanishes there, so every quadratic term of \autoref{eq:sop} drops out and $f\equiv c+\sum_{v\in V}b_vx_v\pmod{r/2}$ is linear.
The two residues modulo $r$ that reduce to $j$ modulo $r/2$ are $j$ and $j+r/2$, so the residue histogram of the reduced polynomial is
\[
M_j=\#\{\bx\in\{0,1\}^V\mid f(\bx)\equiv j\!\!\pmod{r/2}\}=N_j+N_{j+r/2}.
\]
Because that polynomial is linear, $M$ is the coefficient vector of $\delta_c\prod_{v\in V}(\delta_0+\delta_{b_v})$ in the group ring $\Z[\Z_{r/2}]$, with all indices taken modulo $r/2$.
Multiplying the $n$ factors in turn costs $O(r)$ integer additions per factor and uses neither the graph nor the decomposition.

Finally $N_j=(M_j+\Delta_j)/2$ and $N_{j+r/2}=(M_j-\Delta_j)/2$ for $0\le j<r/2$, which is exact over $\Z$ because $M_j+\Delta_j=2N_j$.
\end{proof}

The corollary fixes $r$, which is what puts the $O(nr)$ additions and the $r$ outputs inside $O(4^k\poly(n))$.
When $r$ is part of the input instead, the computation uses $O(4^k\poly(n))$ operations in $\Z[\omega_r]$, then $O(nr)$ integer additions, and writing the $r$ counts down takes $\Omega(r)$ time on its own.
Its bit complexity depends further on the representation cost of an element of $\Z[\omega_r]$, which is $\varphi(r)$ integers in the power basis.

The hypothesis on $r$ is forced by a count.
Mode $1$ carries $\varphi(r)$ integers and the histogram modulo $r/2$ carries $r/2$, so together they can determine all $r$ counts only when $\varphi(r)+r/2\ge r$, that is only when $r$ is a power of two.
At $r=12$ they carry $4+6$ numbers where $12$ are needed, and \autoref{thm:fourier-speedup} applies instead.
Only the first step uses the hypothesis, since the quadratic terms drop out modulo $r/2$ for every even $r$.
The power-of-two hypothesis is useful because it separates the two settings the experiments meet: every exact import in \autoref{sec:experiments} has a power-of-two modulus, whereas rounding phases to a common modulus produces no power of two at all.

Exact arithmetic is what makes the first step possible.
A floating-point evaluation returns an approximation of $S(f)$, from which the integers $\Delta_j$ can be recovered only by rounding, and only at a precision above $n$ bits.
Over $\Z[\omega_r]$ they are the representation itself.

\paragraph{Linear layouts have no join.}\label{sec:fast-join}
At a branching join two child tables must be combined, and the direct implementation scans up to $4^k$ pairs of child signatures.
Linear layouts extend one prefix by one vertex at a time, so there is no join.
This removes one~$2^k$ factor.

\begin{restatable}[Linear-layout Fourier DP]{theorem}{thmlinearlayoutfourier}\label{thm:linear-layout-fourier}
\lean{Formal/Paper.lean}{linear_fourier_mode_spec}
Let $f$ be a quadratic SOP instance of the form in \autoref{eq:sop}, with $n$-vertex SOP variable graph $G=(V,E)$, and let $\pi=(v_1,\ldots,v_n)$ be a linear layout of cut-rank at most $k$.
For each $a\in\Z_r$, the left-to-right DP computes $\widehat N(a)=\sum_{\bx}\omega_r^{af(\bx)}$.
Its operation count is at most $n(2+2\cdot2^k)$.
Consequently a single amplitude is computable in time $O(2^k\poly(n))$, and all residue counts in time $O\bigl((r2^k+r^2)\poly(n)\bigr)$.
The same bounds hold whenever $\lrw(G)\le k$ and a witnessing layout is supplied.
\end{restatable}

\begin{proof}[Proof sketch]
At each prefix, store one coefficient for each signature toward the unprocessed vertices, of which there are at most $2^k$.
Adding one vertex updates each signature for its two Boolean values and contributes its unary phase and the signs from earlier neighbors.
No join occurs, so $n$ such scans give the stated count, and the empty root signature yields the mode sum.
\appref{app:linear-layout-proof} gives the tables and the update rule in full.
\end{proof}

\subsection{Stabilizer-rank optimization}\label{sec:stabjoin}

So far a cut costs what its rank says, whatever sits below it.
The tables the dynamic program produces are not arbitrary, though: each entry is the sum of a quadratic root-of-unity weight over an affine $\F_2$-subspace, the form stabilizer-decomposition simulators manipulate~\cite{BravyiSmithSmolin16,BravyiGosset16,BravyiEtAl19}.
All the Clifford phases below a cut collapse into one such Gauss sum, and only the non-Clifford unary phases have to be decomposed.
A subtree can therefore be charged for something other than its width before its table enters the ordinary dynamic program.
Two quantities govern it per cut instead of one: the cut-rank and the number of \emph{magic} vertices below the cut.
It never does worse than the $O(4^k\poly(n))$ of \autoref{thm:fourier-speedup}, collapses to polynomial time on Clifford instances, and beats both pure guarantees on mixed instances.
(Even Fourier modes are polynomial for a more elementary reason that needs no stabilizer structure: \autoref{prop:even-modes}.)

We fix a Fourier mode $a\in\Z_r$ and keep the notation of the join paragraph of \autoref{sec:sopbyrank}: an internal vertex $t$ has children $L_t$ and $R_t$, and the outside of its cut is $Y=\overline X_t$.
A left signature $\balpha$ and a right signature $\bbeta$ combine into $\bgamma=\balpha|_Y+\bbeta|_Y$ with crossing parity $\chi(\balpha,\bbeta)$.
The signatures at a vertex form an $\F_2$-space, the row space of its cut matrix, of dimension the cut-rank.
A table is a function on that space, and we may ask what structure it has.

\paragraph{Quadratic phase functions.}
Identify $x\in\F_2^w$ with its $0/1$ integer coordinate vector, where $w$ is the dimension of whatever $\F_2$-space the function lives on.
Three such spaces arise below, and $w$ always names the dimension of whichever one is in play.
A stored table sits on a cut's signature space, so there $w$ is that cut's rank, at most $k$.
A join works on the direct sum of the two child signature spaces, so there $w\le2k$.
A switch works on the assignment space of the variables below a vertex, so there $w\le n$.

\begin{definition}[Quadratic phase function]\label{def:qpf}\lean{Formal/Stab/QPF.lean}{QPF}
A \emph{quadratic phase function} (QPF) on $\F_2^w$ is a function
\[
g(x)=\theta\;i^{\ell(x)}\,(-1)^{q(x)}\ \text{for }x\in K, \qquad g(x)=0\ \text{otherwise},
\]
given by a scalar $\theta\in\C$, a degree-one $\Z_4$-valued polynomial $\ell(x)=\sum_j c_jx_j\bmod 4$ in the Boolean coordinates, an $\F_2$-quadratic form $q(x)=\sum_{i<j}q_{ij}x_ix_j\bmod 2$, and an affine subspace $K=\{x\in\F_2^w:\bm Mx=\bm d\}$ cut out by a matrix $\bm M$ over $\F_2$ with $w$ columns, which row reduction lets us take to have at most $w$ rows, together with a right-hand side $\bm d$ of matching length.
\end{definition}

\noindent
A QPF is thus given by a pair: the phase $\theta\,i^{\ell(x)}(-1)^{q(x)}$, which is defined on all of $\F_2^w$, and the affine support $K$ outside which the function is zero.
Everything but $\theta$ is a bounded amount of data: $w$ coefficients $c_j\in\Z_4$, the $\binom{w}{2}$ bits $q_{ij}$, and the at most $w(w+1)$ bits of $(\bm M,\bm d)$, so $O(w^2)$ bits in all.
The scalar is counted separately, since an arbitrary element of $\C$ carries no bit bound.
In our use it is one element of a cyclotomic coefficient ring, hence $O(\varphi(r))$ integers: the mode-$a$ tables of \autoref{sec:fourier} lie in $\Z[\omega_r]$, and the two-term split of \autoref{lem:magic-decomp} inverts $\left(\begin{smallmatrix}1&1\\1&i\end{smallmatrix}\right)$, which adjoins $1/2$.

\begin{example}[A point table as a QPF]\label{ex:point-qpf}
Fix a signature $\bsigma\in\F_2^w$ and take $\theta=1$, $\ell=q=0$, $\bm M=I_w$, and $\bm d=\bsigma$.
Then $K=\{\bsigma\}$ is a single point, so $g(\bsigma)=1$ and $g$ vanishes elsewhere.
Thus the point table stored by the ordinary dynamic program is a QPF.
\end{example}

Up to normalization these are the amplitude functions of stabilizer states supported on affine subspaces, in the phase-sensitive form of CH-form simulators~\cite{BravyiEtAl19}.
Two elementary identities drive everything below.
For bits $x,x'\in\{0,1\}$ we have $x\oplus x'=x+x'-2xx'$, so for the degree-one part
\begin{equation}\label{eq:xor-lift}
i^{\ell(x\oplus x')}=i^{\ell(x)}\,i^{\ell(x')}\,(-1)^{\sum_j (c_j\bmod 2)\,x_jx'_j};
\end{equation}
passing to $\oplus$-coordinates creates only quadratic sign corrections.
Likewise a parity $\lambda(x)=x_{j_1}\oplus\cdots\oplus x_{j_s}$ satisfies $\lambda\equiv\sum_t x_{j_t}-2\sigma_2\pmod 4$, with $\sigma_2$ the elementary symmetric polynomial of degree two in those $s$ bits, hence
\begin{equation}\label{eq:parity-lift}
i^{\mu\lambda(x)}
=\Bigl(i^{\sum_t x_{j_t}}\Bigr)^{\mu}\,(-1)^{\mu\,\sigma_2(x_{j_1},\ldots,x_{j_s})},
\end{equation}
again inside the class.
\hfill\lean{Formal/Stab/Affine.lean}{isQPF_affinePhase}

\begin{restatable}[Closure]{lemma}{lemqpfclosure}\label{lem:qpf-closure}\lean{Formal/Stab/Paper.lean}{qpf_closure_spec}
QPFs on $\F_2^w$ are closed under
\begin{enumerate}[label=(\roman*),ref=(\roman*),leftmargin=2.4em,itemsep=1pt,topsep=2pt]
  \item\label{qpf-prod} pointwise products,
  \item\label{qpf-affine} affine substitutions $x\mapsto Sx'\oplus x_0$, and
  \item\label{qpf-sum1} summation over a single variable,
\end{enumerate}
and the total sum $\sum_xg(x)$ is computable from the data above.
A product stays on $\F_2^w$, an affine substitution moves to the domain of $S$, and summing out one variable returns a QPF on the remaining $\F_2^{w-1}$, equivalently a QPF on $\F_2^w$ that is constant in the summed variable.
A join needs one more operation, which iterates \ref{qpf-sum1}: it sums $g$ over each fiber of a linear map, moving it to the target of that map.
\begin{enumerate}[label=(\roman*),ref=(\roman*),leftmargin=2.4em,start=4,topsep=2pt]
  \item\label{qpf-push} for every linear map $\Pi:\F_2^w\to\F_2^p$, the function
$y\mapsto\sum_{x:\,\Pi x=y}g(x)$ on $\F_2^p$ is a QPF (the pushforward of $g$ along $\Pi$).
\end{enumerate}
Operations \ref{qpf-prod}--\ref{qpf-sum1} and the total sum cost $O(w^3)$ time, and \ref{qpf-push} costs $O((w+p)^3)$, the extra term being the cost of writing the result down on $\F_2^p$.
\end{restatable}

\noindent
Every use we make of \ref{qpf-push} has $p\le w$, so it also costs $O(w^3)$.
Closure of quadratic forms over $\F_2$ on an affine support under products, affine substitution and marginalization is the standard stabilizer formalism~\cite{DehaeneDeMoor03,VandenNest10}.
The $\Z_4$-valued degree-one part retained here records the phase explicitly alongside the $\F_2$ quadratic form, and \appref{app:qpf-closure} proves the lemma in that phase-sensitive form.

\paragraph{The join preserves stabilizer structure.}
The mode-$a$ join of \autoref{sec:fourier} sends two child tables $F$ and $G$ to
\[
H(\bgamma)=\sum_{\balpha|_Y+\bbeta|_Y=\bgamma}F(\balpha)\,G(\bbeta)\, (-1)^{a\,\chi(\balpha,\bbeta)} .
\]
Written abstractly, the two restrictions $\balpha\mapsto\balpha|_Y$ and $\bbeta\mapsto\bbeta|_Y$ are linear maps into the parent signature space, and the crossing parity is a bilinear form on the two child spaces.
It is bilinear because \autoref{lem:chi-well-defined} lets us pick representatives linearly: fix a basis of each child's signature space, choose one representative assignment per basis vector, and extend by linearity.
For arbitrary $F$ and $G$ this join is a scan over all $4^k$ pairs of child signatures, and \appref{app:join-barrier} places a barrier in front of the black-box route to a better base.
Driving that route to $O(2^k\poly(k))$ would imply $\mmomega=2$, and an $O(2^{\alpha k}\poly(k))$ join would imply $\mmomega\le\max\{2,2\alpha\}$.
The tables our dynamic program produces are not arbitrary, and on them the same join is cubic.

\begin{lemma}[Stabilizer-rank join]\label{lem:stab-join}\lean{Formal/Stab/Join.lean}{isQPF_join}
Let both child cut-ranks be at most $k$.
If the two child tables are QPFs, the joined table $H$ is a single QPF, computable in $O(k^3)$ ring operations.
\end{lemma}

\begin{proof}
On the direct sum of the two child signature spaces the summand is a product of two QPFs with the sign of a bilinear form, hence a QPF by \autoref{lem:qpf-closure}\ref{qpf-prod}.
The join sums it along $(\balpha,\bbeta)\mapsto\balpha|_Y+\bbeta|_Y$, hence a QPF by \autoref{lem:qpf-closure}\ref{qpf-push}.
For the cost, that direct sum has dimension at most $2k$ and the parent signature space at most $k$, so \autoref{lem:qpf-closure} applies with $w\le2k$ and $p\le k$.
Since $k\le n$, this is a $\poly(n)$ operation.
\end{proof}

A point table is a QPF (\autoref{ex:point-qpf}), and an arbitrary table on a rank-$k$ cut is a weighted sum of up to $2^k$ of them.
Applying \autoref{lem:stab-join} to every pair of point terms therefore reproduces the naive join of \autoref{thm:fourier-speedup} at its $4^k$ cost.
The saving comes when a whole table is a \emph{short} sum of QPFs, which is what the rest of this section supplies.

\paragraph{Magic vertices.}
We are back on the SOP variable graph of \autoref{sec:sop}, where each vertex $v$ carries the unary phase $\omega_r^{b_v}$ and each edge the cross-term coefficient $\eta=r/2$.
In mode $a$ the vertex phase is $\omega_r^{ab_v}$, and a QPF can absorb exactly the powers of $i$.
Call a vertex $v$ \emph{$a$-magic} when $\omega_r^{ab_v}$ is not a power of $i$, which presupposes $4\mid r$.
For $r=8$ the condition says that $ab_v$ is odd.
The name follows the magic states of stabilizer-rank simulation: at $r=8$ an $a$-magic vertex carries exactly a $\T$ phase, the resource those decompositions are charged for~\cite{BravyiSmithSmolin16,BravyiGosset16}.
Let $\tau_u$ be the number of $a$-magic vertices below vertex $u$ of the decomposition.

Fix such a $u$ and an assignment $\bz$ to the variables below it, and read off the mode-$a$ weight term by term.
Because $\eta=r/2$, each edge contributes $\omega_r^{a\eta z_vz_{v'}}=(-1)^{az_vz_{v'}}$, so the edges below $u$ collect into the $\F_2$-quadratic form $q(\bz)=a\sum_{\{v,v'\}}z_vz_{v'}\bmod 2$.
A vertex that is not $a$-magic contributes a power of $i$, say $\omega_r^{ab_v}=i^{c_v}$ with $c_v\in\Z_4$, and these collect into the degree-one part $\ell(\bz)=\sum_{v\ \text{not }a\text{-magic}}c_vz_v$.
What is left is one factor per magic vertex, so the weight factors as
\[
\omega_r^{a\bigl(\sum_v b_vz_v+\eta\sum_{\{v,v'\}}z_vz_{v'}\bigr)} =\underbrace{\;i^{\ell(\bz)}(-1)^{q(\bz)}\;}_{\text{one QPF}} \;\prod_{v\ a\text{-magic}}\omega_r^{ab_vz_v},
\]
both sums ranging over the variables and edges below $u$.
The product over the magic vertices is the one factor that is not already a QPF, and it is what a stabilizer-rank decomposition attacks: such a decomposition rewrites that product as a sum of QPFs, and the number of terms is the price.
Write $\mathrm{sr}(\tau)$ for a bound on that number of terms for $\tau$ magic factors.
Our algorithm consumes the terms themselves, so $\mathrm{sr}$ has to be \emph{constructive}: a decomposition of length at most $\mathrm{sr}(\tau)$ must be \emph{generated} in $\mathrm{sr}(\tau)\poly(n)$ time.
The two-term split in the proof below is constructive and gives $\mathrm{sr}(\tau)\le 2^{\tau}$ whenever $4\mid r$, and for $r=8$ the explicit decompositions of~\cite{BravyiSmithSmolin16,QassimPashayanGosset21} give $\mathrm{sr}(\tau)=O(2^{0.3963\,\tau})$.
Any other explicitly generable decomposition substitutes for $\mathrm{sr}$ verbatim.
The minimum stabilizer rank is at most every such $\mathrm{sr}$, but knowing that a short decomposition exists does not produce one, so every cost below is charged on the constructive bound.

\begin{lemma}[Magic decomposition]\label{lem:magic-decomp}\lean{Formal/Stab/Magic.lean}{modeWeight_sum_qpf}
The mode-$a$ table at a vertex $u$ is a sum of at most $\mathrm{sr}(\tau_u)$ QPFs, computable in $\mathrm{sr}(\tau_u)\poly(n)$ time from the instance alone.
\end{lemma}

\begin{proof}
A magic factor $\omega_r^{ab_vz_v}$ depends on the single bit $z_v$, so it is determined by the pair of values it takes at $z_v=0$ and $z_v=1$.
The two QPFs $z\mapsto1$ and $z\mapsto i^{z}$ take the value pairs $(1,1)$ and $(1,i)$, and $\left(\begin{smallmatrix}1&1\\1&i\end{smallmatrix}\right)$ is invertible, so they are a basis for the functions of one bit.
Each magic factor is therefore a combination of those two QPFs, and expanding the $\tau_u$ factors independently gives $2^{\tau_u}$ products, each of them a QPF.
For $r=8$ the product of $\tau$ magic factors is, up to a Clifford factor, the amplitude function of $\ket{T}^{\otimes\tau}$, of stabilizer rank $O(2^{0.3963\tau})$~\cite{QassimPashayanGosset21}.
Multiplying each term into the Clifford QPF and summing along the assignment-to-signature map (\autoref{lem:qpf-closure}\ref{qpf-push}) decomposes the table term by term.
\end{proof}

We can now rederive the polynomial-time exact-amplitude consequence of the Gottesman--Knill theorem, and identify the tables of Clifford instances as stabilizer-state amplitude vectors.
It costs nothing beyond the two lemmas above.
\begin{corollary}[Gottesman--Knill through the SOP lens]\label{cor:gottesman-knill}\lean{Formal/Stab/Paper.lean}{clifford_collapse_spec}
Let the instance be Clifford in the native sign-edge form of \autoref{eq:sop}.
Then for every mode $a$ the mode-$a$ weight on the assignment space is a single QPF, and its total sum evaluates $\widehat N(a)$ in $O(n^3)$ time, on any graph and at any rank-width.
\end{corollary}

\begin{proof}
Every unary phase $\omega_r^{b_v}$ of such an instance is a power of $i$, so no vertex is $a$-magic for any mode $a$ and $\tau=0$ at every vertex.
\autoref{lem:magic-decomp} at the root therefore presents the weight as $\mathrm{sr}(0)=1$ QPF on the assignment space, where $w\le n$.
Its total sum is $\widehat N(a)$ and costs $O(n^3)$ by \autoref{lem:qpf-closure}.
\end{proof}

Gadgetized representations of the same gates need not be in native form.
Sums of this shape are the tractable quadratic exponential sums of~\cite{CaiCLL10}, though the cost above needs nothing beyond \autoref{lem:qpf-closure}.
\paragraph{Where to switch.}
The two prices are combined by an antichain $B$ of decomposition vertices, with no member a strict ancestor of another.
Call a vertex \emph{active} if it does not lie strictly below a member of $B$.
At each $b\in B$, \autoref{lem:magic-decomp} constructs the mode-$a$ table from the instance as a sum of at most $\mathrm{sr}(\tau_b)$ QPFs, which are evaluated on the $2^{\rho_b}$ signatures of the cut.
Nothing below $b$ is computed.
At every other active internal vertex, the ordinary per-mode DP of \autoref{thm:fourier-speedup} joins the two child tables.
The switch from the magic price to the width price therefore occurs at most once on each rootward branch, as in \autoref{fig:switch-antichain}.

The \emph{total cost} of $B$ is
\begin{equation}\label{eq:total-cost}
T(B)=\poly(n)\left(
\sum_{b\in B}\mathrm{sr}(\tau_b)\bigl(1+2^{\rho_b}\bigr)
+\!\!\!\sum_{\substack{t\ \mathrm{active\ internal}\\t\notin B}}
2^{\rho_{L_t}+\rho_{R_t}}
\right),
\end{equation}
and the \emph{best total cost} of the decomposition is $\min_B T(B)$ over all such antichains.

\begin{figure}[!tbp]
\centering
\begin{tikzpicture}[
  x=1.35cm,
  y=0.88cm,
  activejoin/.style={draw, circle, fill=green!18, minimum size=5.5mm, inner sep=0pt, font=\scriptsize},
  switch/.style={draw, circle, fill=cyan!22, minimum size=6mm, inner sep=0pt, font=\scriptsize},
  pruned/.style={draw=gray!65, circle, fill=gray!18, minimum size=4.5mm, inner sep=0pt},
  charge/.style={font=\scriptsize, align=center}
]
\draw[dashed, very thick, cyan!60!black] (-3.45,0.9) -- (-0.15,0.9) -- (1.35,2.1) -- (2.75,2.1);
\node[activejoin] (t) at (0,3.2) {$t$};
\node[activejoin] (u) at (-1.8,2.1) {$u$};
\node[switch] (b3) at (1.8,2.1) {$b_3$};
\node[switch] (b1) at (-2.7,0.9) {$b_1$};
\node[switch] (b2) at (-0.9,0.9) {$b_2$};
\node[pruned] (l11) at (-3.15,-0.25) {};
\node[pruned] (l12) at (-2.25,-0.25) {};
\node[pruned] (l21) at (-1.35,-0.25) {};
\node[pruned] (l22) at (-0.45,-0.25) {};
\node[pruned] (l31) at (1.35,0.9) {};
\node[pruned] (l32) at (2.25,0.9) {};
\draw (t) -- (u);
\draw (t) -- (b3);
\draw (u) -- (b1);
\draw (u) -- (b2);
\draw[gray!65] (b1) -- (l11);
\draw[gray!65] (b1) -- (l12);
\draw[gray!65] (b2) -- (l21);
\draw[gray!65] (b2) -- (l22);
\draw[gray!65] (b3) -- (l31);
\draw[gray!65] (b3) -- (l32);
\node[charge, right=4mm of t] {ordinary join\\$2^{\rho_{L_t}+\rho_{R_t}}$};
\node[charge] at (0,-0.8) {each $b\in B$: $\mathrm{sr}(\tau_b)\bigl(1+2^{\rho_b}\bigr)$};
\end{tikzpicture}
\caption{The antichain $B=\{b_1,b_2,b_3\}$ is a one-way switching frontier.
The gray subtrees below them are pruned.
Above the frontier, every ordinary join pays the width charge, while each switch pays the stabilizer-rank charge.
These are the two sums in \autoref{eq:total-cost}.}
\label{fig:switch-antichain}
\end{figure}

It is found by one postorder sweep: with $C(u)$ denoting the minimum below $u$ after removing the common polynomial factor,
\begin{equation}\label{eq:switch-recurrence}
C(u)=\min\left(
\mathrm{sr}(\tau_u)\bigl(1+2^{\rho_u}\bigr),
C(L_u)+C(R_u)+2^{\rho_{L_u}+\rho_{R_u}}
\right),
\qquad C(\text{leaf})=O(1).
\end{equation}
One number and one choice bit per vertex suffice, so an optimal antichain is computable in $O(n)$ arithmetic operations for any constructive price $\mathrm{sr}$ (\appref{app:stabjoin-proofs}).

Choosing $B=\varnothing$ recovers \autoref{thm:fourier-speedup} exactly, so the best total cost never exceeds $O(4^k\poly(n))$.\lean{Formal/Stab/Paper.lean}{hybrid_never_loses_spec}
Choosing only the root recovers global stabilizer-decomposition simulation in $\mathrm{sr}(\tau)\poly(n)$~\cite{BravyiGosset16,BravyiEtAl19}.
Thus the optimization refines the global $\T$-count parameterization of those simulators to the finite-modulus SOP setting and complements rank-width-parameterized ZX contraction, whose parameter is the width alone~\cite{KuyanovKissinger26,CodsiLaakkonen26}.

\paragraph{Mixed instances.}
Cut-rank and magic count are incomparable, so the two prices separate on instances that mix them: a magic-dense part of low cut-rank, attached by a single edge to a part of high rank-width with no magic vertices at all.

\begin{restatable}[Mixed instances]{proposition}{propmixedseparation}\label{prop:mixed-separation}
For every $s\ge2$ there is a circuit $D_s$ over $\{\H,\CZ,\T\}$ whose pinned SOP has $n_s=\Theta(s^2)$ variables and $\T$-count $\tau_s=s$, and a rank-decomposition of $G_{D_s}$ on which the stabilizer-rank optimization evaluates the amplitude in $\poly(n_s)$ operations.
On that family the two pure routes are bounded only by $4^{\rw(G_{D_s})}\poly(n_s)$, where $\rw(G_{D_s})\ge s-1$, and by $\mathrm{sr}(s)\poly(n_s)$.
\end{restatable}

Both are exponential in $s=\Theta(\sqrt{n_s})$, the second at the constructive stabilizer ranks of \autoref{tab:route-comparison}.
Only the decomposition is given.
The postorder recurrence of \autoref{eq:switch-recurrence} chooses the switching antichain, and \appref{app:stabjoin-proofs} exhibits the family and its antichain.

\section{Comparison with other parameters and methods}\label{subsec:separation-linear-tensor}

Having given the algorithm, we now compare the parameter it runs on with those governing the competing routes.
We first recall tensor networks and their contraction complexity, and show that the SOP variable graph is never harder than the circuit's tensor network, so the SOP route is competitive with tensor contraction before rank-width is used at all.
We then place rank-width, linear rank-width, and contraction complexity side by side (\autoref{fig:rw-lrw-tw-separation}), show that arbitrary simple graphs can be realized as SOP variable graphs of small circuits, and finally instantiate this with a graph family that has bounded rank-width but growing competing parameters.
The comparison is made at two levels: structural graph parameters, and, for the closest algorithms, decomposition-sensitive costs on a common representation.

\subsection{Tensor networks}\label{sec:tensor-network}

We recall tensors, tensor networks, and Markov and Shi's contraction complexity \cite{MarkovShi08}.

A \emph{tensor} of \emph{order} $\alpha$ is an $\alpha$-dimensional array of complex numbers, that is, a map $A\colon [d_1]\times\cdots\times[d_\alpha]\to\C$, where $[d]=\{0,\ldots,d-1\}$.
Each of the $\alpha$ argument positions is an \emph{index}, and we write $A[i_1,\ldots,i_\alpha]$ for the entry of $A$ at $(i_1,\ldots,i_\alpha)$.
A \emph{tensor network} is a finite collection of tensors $\mathcal T=\{A_1,\ldots,A_\mu\}$ with some of their index positions paired across tensors.
Each matched pair (whose two positions must share a dimension) is a single \emph{bond} (or \emph{contracted}) \emph{index}, and the unmatched positions are the \emph{open} (or \emph{free}) \emph{indices}.
We write $e_1,\ldots,e_p$ for the bond indices and $o_1,\ldots,o_\ell$ for the open ones.

A \emph{bond configuration} $\bm\beta=(\beta_1,\ldots,\beta_p)$ assigns a value $\beta_i\in[d_{e_i}]$ to every bond index $e_i$.
Together with a choice $\bm j=(j_1,\ldots,j_\ell)$ of values for the open indices, it fixes a value for \emph{every} index position of \emph{every} tensor, so each $A_t$ has a well-defined entry.
The \emph{value} of the network is the order-$\ell$ tensor obtained by summing, over all bond configurations, the product of these entries:
\[
V[j_1,\ldots,j_\ell] =\sum_{\beta_1,\ldots,\beta_p} \prod_{t=1}^\mu A_t[\bm\beta,\bm j],
\]
where $A_t[\bm\beta,\bm j]$ denotes the entry of $A_t$ read at the values that $(\bm\beta,\bm j)$ assigns to the indices of $A_t$.
For circuits every index is a qubit wire, so each dimension is $d=2$.

\paragraph{Tensor-network graph.}
Given a tensor network $\mathcal T$, its \emph{tensor-network graph} (a multigraph) $N_{\mathcal T}$ has one vertex per tensor and one edge $\{t,t'\}$ for each bond index shared between $A_t$ and $A_{t'}$.
Open indices are not represented by edges.

It is sometimes convenient to represent a quantum circuit $C=U_m\cdots U_1$ over $\{\H,\T,\CZ\}$ as a tensor network.
In this context, each gate is a tensor: $\H$ and $\T$ are order-$2$ tensors (one input, one output qubit index) and $\CZ$ is order-$4$ (two input, two output qubit indices).
Adjacent gates sharing a wire are connected by a bond index.
The circuit's input and output wires are the open indices.
For a circuit $C$, we write $N_C$ for this tensor-network graph.
Thus $N_C$ has one vertex per gate and one edge per wire segment between consecutive gates.

\begin{example}[Deriving $N_C$]\label{ex:nc}
The circuit of \autoref{ex:circuit} has nine gate tensors $g_1,\ldots,g_9$, one for each gate in left-to-right order.
Each $\H$ and $\T$ gate contributes an order-$2$ tensor, with one input index and one output index, while each $\CZ$ gate contributes an order-$4$ tensor, with two input indices and two output indices.
We define $N_C=(V,E)$ by taking $V$ to be the set of gate tensors and $E$ to be the set of internal bonds, that is, the wire segments connecting consecutive gates on the same qubit wire.
In \autoref{fig:ex-nc}, the eight such internal bonds are drawn as the eight edges of $N_C$.

The input and output indices are not edges of $N_C$.
They are boundary indices, handled separately by the boundary assignment $(\by,\bz)$ when considering the pinned amplitude $\braket{\bz|C|\by}$.
Equivalently, one may obtain a fully contracted network for this amplitude by adjoining rank-$1$ boundary tensors, which fix the input values $\by$ and project onto the output values $\bz$.
\autoref{fig:ex-nc} shows them as dotted open legs.
\end{example}

\begin{figure}[ht]

\centering

    \begin{subfigure}[b]{0.45\textwidth}

\centering
\begin{tikzpicture}[
  scale=0.83,
  gate/.style={draw, circle, fill=blue!12, minimum size=0.65cm, inner sep=0pt, font=\scriptsize},
  bdy/.style={font=\scriptsize, gray}
]
\node[gate] (g1) at (0, 4)     {$g_1$};
\node[gate] (g4) at (1, 3) {$g_4$};
\node[gate] (g7) at (4.5, 4)   {$g_7$};
\node[gate] (g2) at (0, 2)     {$g_2$};
\node[gate] (g5) at (2, 1) {$g_5$};
\node[gate] (g6) at (3, 2)   {$g_6$};
\node[gate] (g8) at (4.5, 2)   {$g_8$};
\node[gate] (g3) at (0, 0)     {$g_3$};
\node[gate] (g9) at (4.5, 0)   {$g_9$};

\node[bdy] (i1) at (-0.9, 4) {$\iota_1$};
\node[bdy] (i2) at (-0.9, 2) {$\iota_2$};
\node[bdy] (i3) at (-0.9, 0) {$\iota_3$};
\node[bdy] (o1) at (5.4, 4) {$o_1$};
\node[bdy] (o2) at (5.4, 2) {$o_2$};
\node[bdy] (o3) at (5.4, 0) {$o_3$};

\draw[gray, dotted, thick] (i1) -- (g1);
\draw[gray, dotted] (i2) -- (g2);
\draw[gray, dotted] (i3) -- (g3);
\draw[gray, dotted] (g7) -- (o1);
\draw[gray, dotted] (g8) -- (o2);
\draw[gray, dotted] (g9) -- (o3);

\draw (g1) -- (g4);
\draw (g4) -- (g7);
\draw (g2) -- (g4);
\draw (g4) -- (g5);
\draw (g5) -- (g6);
\draw (g6) -- (g8);
\draw (g3) -- (g5);
\draw (g5) -- (g9);
\end{tikzpicture}
\caption{The tensor-network graph $N_C$ for the circuit of \autoref{ex:circuit}.
Solid edges are bond indices (wire segments between consecutive gates).
Dotted edges are boundary indices.}
\label{fig:ex-nc}
    \end{subfigure}
\hfill
    \begin{subfigure}[b]{0.45\textwidth}
\centering
\begin{tikzpicture}[
  scale=0.83,
  idx/.style={draw, circle, fill=red!12, minimum size=0.65cm, inner sep=0pt, font=\tiny}
]
\node[idx] (e14) at (-1.4,  2)    {$e_{1,4}$};
\node[idx] (e47) at ( 1.4,  2)    {$e_{4,7}$};
\node[idx] (e24) at (-1.4,  0.4)  {$e_{2,4}$};
\node[idx] (e45) at ( 0,   -0.2)  {$e_{4,5}$};
\node[idx] (e56) at ( 1.6, -1.6)  {$e_{5,6}$};
\node[idx] (e68) at ( 3.6, -1.6)  {$e_{6,8}$};
\node[idx] (e35) at (-1.4, -1.6)  {$e_{3,5}$};
\node[idx] (e59) at ( 0,   -3.2)  {$e_{5,9}$};

\draw (e14) -- (e47);
\draw (e14) -- (e24);
\draw (e14) -- (e45);
\draw (e47) -- (e24);
\draw (e47) -- (e45);
\draw (e24) -- (e45);

\draw (e45) -- (e56);
\draw (e45) -- (e35);
\draw (e45) -- (e59);
\draw (e56) -- (e35);
\draw (e56) -- (e59);
\draw (e35) -- (e59);

\draw (e56) -- (e68);
\end{tikzpicture}
\caption{The line graph $L(N_C)$ for the circuit of \autoref{ex:circuit}.
Note how bond indices become exactly the vertex set.
In particular, there are no (open) boundary vertices.}
\label{fig:ex-lnc}
\end{subfigure}
\caption{A tensor-network graph and its corresponding line graph.}
\end{figure}

\paragraph{Contraction complexity.}
A \emph{contraction order} is a sequence of pairwise tensor contractions that reduces the network to a single tensor.
A contraction step merges two tensors and sums out their shared indices, producing a new tensor whose bond indices with the rest of the network are the union of those of the two.
The cost of the step is exponential in the number of indices of this intermediate tensor.
This process is captured by the \emph{line graph} $L(N_C)$, whose vertices are the bond indices of $N_C$ (its edges) and whose edges record co-occurrence in a common tensor.
In this graph, a contraction step corresponds exactly to a sequence of vertex-elimination steps, in which the shared indices are removed and the remaining indices on the merged tensor form a clique.
A full contraction order is therefore an elimination order on $L(N_C)$, and Markov and Shi~\cite{MarkovShi08} showed that the minimum-cost order is governed by $\tw(L(N_C))$.
This motivates the definition
\[
\cc(N_C)\coloneqq\tw(L(N_C))
\]
of the contraction complexity of a circuit.

\begin{example}[Deriving $L(N_C)$]\label{ex:lnc}
Continuing \autoref{ex:nc}, each tensor of $N_C$ contributes a clique on its incident bond indices in $L(N_C)$.
The two $\CZ$'s have all four of their indices internal, so each gives a $K_4$: $\{e_{1,4},e_{4,7},e_{2,4},e_{4,5}\}$ from $g_4=\CZ_{1,2}$ and $\{e_{4,5},e_{5,6},e_{3,5},e_{5,9}\}$ from $g_5=\CZ_{2,3}$, sharing the bond index $e_{4,5}$ on the wire of qubit $2$ between them.
The $\T$ at position $6$ has two bond indices and so contributes the single edge $e_{5,6}\!-\!e_{6,8}$.
Each of the six boundary Hadamards has only one bond index.
Its other index is an open input/output index, so it contributes no edge.
These open indices are the leaves omitted from the bond-restricted $L(N_C)$ of \autoref{fig:ex-lnc}.
The result is two $K_4$'s sharing the single vertex $e_{4,5}$, with $e_{6,8}$ attached to $e_{5,6}$ as a pendant, so $\tw(L(N_C))=\cc(N_C)=3$.
\end{example}

\subsection{SOP evaluation matches contraction complexity}\label{sec:treewidth}

We now relate the SOP variable graph $G_C$ of \autoref{sec:sop} to the contraction complexity $\cc(N_C)=\tw(L(N_C))$ just defined, which governs the Markov--Shi tensor-contraction bound~\cite{MarkovShi08}.
The resulting bound $\tw(G_C)\le\cc(N_C)$ is what lets the rank-width dynamic program of \autoref{sec:dp} be compared with tensor contraction, in the separations that follow, and it yields a treewidth FPT dynamic program for SOP evaluation.

\begin{restatable}[SOP graph is a minor of the tensor line graph]{lemma}{lemsopminorline}\label{lem:sop-minor-line}
For circuits $C$ over $\{\H,\T,\CZ\}$, the SOP variable graph $G_C$ is a minor of $L(N_C)$.
Hence $ \tw(G_C)\le \tw(L(N_C))=\cc(N_C). $
\end{restatable}

\begin{proof}[Proof sketch]
Collapse the connected bond path of every Hadamard-free wire segment, delete the outer pinned segments, and delete the edges whose sign contributions cancel modulo $8$.
These minor operations leave exactly $G_C$, and minor-monotonicity of treewidth gives the inequality.
The three steps are carried out in \appref{app:parameter-proofs}.
\end{proof}

The contraction complexity is defined through the line graph $L(N_C)$.
Since our tensor networks have fixed arity, one may also relate it to the treewidth of $N_C$ itself~\cite{MarkovShi08}.
If every tensor has arity at most $\Delta$, then a tree decomposition of $N_C$ of width $t$ can be converted into one of $L(N_C)$ of width at most $\Delta(t+1)-1$ by replacing each bag with all incident tensor indices.
Without bounded arity, this implication is false: a star with one high-arity tensor has treewidth one as a tensor graph, but a line graph that is a clique.

By \autoref{lem:sop-minor-line}, therefore, $\tw(G_C)\le\cc(N_C)$: the SOP variable graph is never harder to contract than the tensor network.
Quadratic SOP evaluation is moreover a special case of the standard sum--product problem on a graphical model: the variables are Boolean and the unary and binary factors are exactly the coefficients of \autoref{eq:sop}.
The SOP variable graph $G_C$ is therefore the \emph{primal graph} of the instance.
By bucket elimination (equivalently, tree-decomposition dynamic programming) such instances can be evaluated in time $2^{O(\tw(G_C))}\poly(n)$~\cite{Dechter99,DBLP:journals/jair/Ganian0SS22}.
The link between path-sum formulas and graphical-model primal graphs is also noted in~\cite{DBLP:journals/corr/abs-2510-06775}.
Here \autoref{lem:sop-minor-line}, together with the observation that $G_C$ is the primal graph, already yields a treewidth fixed-parameter algorithm for SOP evaluation \emph{for free}, by off-the-shelf bucket elimination.
Its parameter dependence matches the Markov--Shi contraction-complexity guarantee, though the algorithms differ: bucket elimination runs on the primal graph $G_C$, whereas Markov--Shi contracts the tensor network.
This treewidth algorithm is not our contribution.
It is the baseline that \autoref{sec:dp} improves on, by refining the parameter from treewidth to the rank-width of $G_C$.

\subsection{Separations with linear rank-width and contraction complexity}

The dynamic program of \autoref{sec:dp} is parameterized by the rank-width of the SOP variable graph.
This is a branching parameter: the decomposition tree may split the variables recursively.
Linear decision-diagram approaches \cite{DBLP:journals/corr/abs-2510-06775} use a variable ordering and are therefore governed by a linear layout parameter, such as linear rank-width.
Tensor-network contraction is governed by the contraction complexity $\cc(N_C)=\tw(L(N_C))$.
The construction below separates these parameters at the level of the SOP variable graph, giving explicit circuit families where the rank-width guarantee stays bounded while the linear-layout and tensor-contraction guarantees grow.
The separation is built in two steps.
We first realize any simple graph as the SOP variable graph of a small circuit, then instantiate a blow-up family on which rank-width stays bounded while both competing parameters diverge.

\begin{figure}[h]
\centering
\begin{tikzpicture}[
  q/.style={font=\normalsize},
  op/.style={font=\large},
  tag/.style={align=center, font=\footnotesize}
]
\node[q]  (lrw) at (0,   0) {$\lrw(G_C)$};
\node[op] (oa)  at (2.7, 0) {$\ge$};
\node[q]  (rw)  at (5.0, 0) {$\rw(G_C)$};
\node[op] (ob)  at (7.3, 0) {$\le$};
\node[q]  (tw)  at (9.6, 0) {$\tw(G_C)+1$};

\node[tag, above=4mm] at (oa) {\emph{general}};
\node[tag, above=4mm] at (ob) {\emph{family} $C_{h,t}$\\ $O(1)$ vs ${\ge}\,2t{-}1$};
\node[tag, below=4mm] at (ob) {$\Rightarrow \cc(N_C)=\Omega(t)$ by \autoref{lem:sop-minor-line}};
\node[tag] at (4.8, -1.55) {and in general $\rw(G_C)\le\tw(G_C)+1\le\cc(N_C)+1$};
\end{tikzpicture}
\caption{The structural parameters governing the competing routes: rank-width ($\rw$, this work), linear rank-width ($\lrw$, linear decision diagrams), and treewidth ($\tw$, tensor contraction, via $\rw(G)\le\operatorname{bw}(G)\le\tw(G)+1$).
The last chain runs through \autoref{lem:sop-minor-line}, which exhibits $G_C$ as a minor of the tensor line graph $L(N_C)$, whose treewidth is the contraction complexity $\cc(N_C)$.
The family $C_{h,t}$ of \autoref{cor:separating-family} keeps $\rw=O(1)$ while $\lrw=\Omega(h)$ and $\tw\ge 2t-1$.}
\label{fig:rw-lrw-tw-separation}
\end{figure}

\paragraph{Realizing arbitrary SOP variable graphs by circuits.}

We first establish that for an arbitrary simple graph we can construct a quantum circuit whose SOP variable graph is exactly the input graph.
This avoids relying on a special feature of a particular benchmark family (cf. the separation from \cite{DBLP:journals/corr/abs-2510-06775}).

\begin{restatable}[Realizing a graph as an SOP variable graph]{lemma}{lemrealizeanygraph}\label{lem:realize-any-graph}
Let $\Gamma=(W,F)$ be a simple graph.
There is a circuit $C_\Gamma$ over $\{\H,\CZ\}$ with $|W|$ qubits and $2|W|+|F|$ gates whose SOP variable graph, after fixing the input and output basis states and substituting the pinned variables, is exactly $\Gamma$.
\end{restatable}

\begin{proof}[Proof sketch]
Place one $\H$ gate on every qubit, one $\CZ$ gate for every edge, and a final $\H$ layer, so pinning the outer wire segments leaves one free variable per vertex and exactly the quadratic interactions of $\Gamma$.
\appref{app:parameter-proofs} checks the segment count and the cross terms.
\end{proof}

\paragraph{A separating circuit family.}

We next define a family of graphs with bounded rank-width, unbounded linear rank-width, and arbitrarily large treewidth.
Let $B_h$ be the complete binary tree of height $h$.
For an integer $t\ge 1$, let
\[
\Gamma_{h,t}=B_h[K_t]
\]
be the graph obtained from $B_h$ by replacing every vertex $v$ by a clique $K_t^v$ of $t$ twins, and by replacing every tree edge $\{u,v\}$ by the complete bipartite graph between $K_t^u$ and $K_t^v$.
Equivalently, $\Gamma_{h,t}$ is the lexicographic blow-up of $B_h$ by a $t$-clique.
The bracket notation $G[H]$ denotes this substitution of a copy of $H$ for every vertex of $G$.

\begin{restatable}{lemma}{lemblowupwidths}\label{lem:blowup-widths}
For all $h$ and $t\ge 1$,
\begin{itemize}
\item $\rw(\Gamma_{h,t})\le 1$,
\item $\lrw(\Gamma_{h,t})=\Omega(h)$, and
\item
$\tw(\Gamma_{h,t})\ge 2t-1$ for $h\ge1$ (and $\ge t-1$ for $h=0$).
\end{itemize}
\end{restatable}

\begin{proof}[Proof sketch]
A width-$1$ decomposition of $B_h$ extends over the true-twin cliques, one representative per clique induces $B_h$ and gives the linear-rank-width lower bound, and two adjacent cliques induce $K_{2t}$ for the treewidth bound, with the lone $K_t$ giving the case $h=0$.
\appref{app:parameter-proofs} gives the three arguments.
\end{proof}

\begin{remark}
There is also a one-line way to see the rank-width bound.
The graph $\Gamma_{h,t}$ is obtained from a tree by true-twin additions.
Equivalently, it is distance-hereditary, and distance-hereditary graphs are exactly the graphs of rank-width at most $1$~\cite{Oum05RankWidthVertexMinors,Oum17Survey}.
\end{remark}

\begin{corollary}[A non-Clifford separation family]\label{cor:separating-family}
For every $h$ and $t\ge 1$ there is a non-Clifford circuit $C_{h,t}$ over $\{\H,\T,\CZ\}$ such that its SOP variable graph $G_{C_{h,t}}$ satisfies
\begin{itemize}
\item $\rw(G_{C_{h,t}})\le1$,
\item
$\lrw(G_{C_{h,t}})=\Omega(h)$, and
\item
$\cc(N_{C_{h,t}})=\Omega(t)$.
\end{itemize}
Consequently, the rank-width DP has bounded parameter on this family, while linear rank-width grows with $h$ and contraction complexity grows with $t$.
\end{corollary}

\begin{proof}
Apply \autoref{lem:realize-any-graph} to $\Gamma_{h,t}$.
Then insert a $\T$ gate on every qubit between the $\CZ$ layer and the final $\H$ layer, and call the resulting circuit $C_{h,t}$.
Thus the $\T$-count is exactly $|V(\Gamma_{h,t})|=t|V(B_h)|=t(2^{h+1}-1)$, one $\T$ gate per qubit.
Each inserted $\T$ gate is diagonal on the middle segment variable $x_{v,1}$, so it contributes only the unary coefficient $b_{x_{v,1}}=1$ and creates no additional edge.
After the input and output pins are substituted, the SOP variable graph is still
\[
G_{C_{h,t}}=\Gamma_{h,t}.
\]
The circuit contains non-Clifford $\T$ gates in Hadamard-separated wire segments, so the witnessing family is not confined to Clifford circuits.
The rank-width and linear-rank-width claims follow from \autoref{lem:blowup-widths}.

For contraction complexity, \autoref{lem:sop-minor-line} gives
\[
\tw(G_{C_{h,t}})\le \tw(L(N_{C_{h,t}}))=\cc(N_{C_{h,t}}).
\]
Since $G_{C_{h,t}}=\Gamma_{h,t}$ contains $K_{2t}$ for $h\ge1$, \autoref{lem:blowup-widths} gives $\tw(G_{C_{h,t}})\ge 2t-1$.
Hence $\cc(N_{C_{h,t}})=\Omega(t)$.
\end{proof}

The case $t=1$ already separates rank-width from linear rank-width: the SOP variable graph is a complete binary tree, so the rank-width is bounded while the linear rank-width grows with the height.
Increasing $t$ makes the tensor-network contraction complexity grow as well, without changing the rank-width bound.
For instance, choosing $t=2^h$ gives instances with $n=t(2^{h+1}-1)=\Theta(t^2)$ variables, so that $t=\Theta(\sqrt n)$ and $h=\Theta(\log n)$.
On this diagonal the rank-width parameter stays $O(1)$ while linear rank-width grows as $\Omega(\log n)$ and contraction complexity as $\Omega(\sqrt n)$, so both competing parameters are unbounded functions of the instance size while ours is not.
This separates the structural parameters used by the three algorithmic routes.
Since the $\T$-count is linear in the number of qubits, stabilizer-decomposition simulators whose worst-case dependence is exponential in $\T$-count do not get a constant-$\T$ shortcut on this family.
What separates is the structural guarantee each route can prove in advance.

The construction has the familiar $\H$, diagonal, $\H$ shape of IQP-style circuits, although here it is used for exact fixed-boundary evaluation and structural-parameter separation rather than approximate-sampling hardness~\cite{BremnerJozsaShepherd10}.

\paragraph{Decomposition-sensitive comparison.}
The parameters above are one way to compare the routes.
A second way holds the decomposition fixed and counts what each algorithm charges on it.
Our rooted DP charges one cut-rank pair per internal vertex, the pair its rooting selects.
Kuyanov and Kissinger work on the unrooted tree and charge the cheapest of the three pairs at each vertex~\cite{KuyanovKissinger26}.
No rooting selects all of those pairs at once, so on a fixed decomposition their bound can be the finer one, and neither route dominates the other asymptotically.
The per-cut objects are nonetheless the same information in two bases: at each cut our table holds the Walsh coefficients of their boundary state.
That correspondence is also what gives Feynman-path decision diagrams their linear-rank-width guarantee.
On the sign-edge subclass the two settings share, the same representation is what yields our magic-focused-width guarantee, in the spirit of Codsi and Laakkonen~\cite{CodsiLaakkonen26}.
All three comparisons read the direct translations of \autoref{prop:four-guises}, before any rewriting.
After simplification a pinned SOP graph need not coincide with a graph-like ZX diagram.

\section{Implementation and experiments}\label{sec:experiments}

We validate the implementation and survey where the routes help: whether they behave as predicted on controlled families, and how they fare on a standard corpus.
Our prototype simulator, \url{https://github.com/gaperez64/dlx4sop}, implements treewidth bucket elimination, the branching rank-width DP of \autoref{thm:fourier-speedup}, and the linear-layout DP of \autoref{thm:linear-layout-fourier} directly on $G_C$.
Beyond these width-only routes it provides two more: one through sign variables and weighted model counting (WMC), and a quadratic-phase-function route that realizes the root-switch endpoint of \autoref{sec:stabjoin}, though not the antichain optimization described there.
\autoref{sec:results} compares that endpoint with the width-only routes and reports which route is selected on the corpus.

A circuit imports \emph{exactly} if its gates lie in the implementation's explicit symbolic catalog.
Such entries share a common dyadic-cyclotomic ring, the setting of the gadgets of \autoref{sec:gadget-universality}.
Every other circuit imports by \emph{rounding}, moving the remaining phases to a common modulus under an error bound the importer verifies (\autoref{sec:circ2sop}).

We compare our implementation with general-purpose statevector and decision-diagram simulators, asking:
\begin{description}
  \item[RQ1 (Corpus).] How broadly do exact and rounded imports (see \autoref{sec:bench}) cover fixed-boundary amplitude instances, and can the native routes handle the large moduli rounding introduces?
  \item[RQ2 (SOP baselines).] Which SOP route is the strongest broad baseline, and which benefits on families that separate the width parameters?
  \item[RQ3 (External comparison).] How do the SOP routes compare with statevector and decision-diagram simulation?
  \item[RQ4 (WMC scaling).] Does weighted model counting scale better end-to-end than the native SOP routes?
\end{description}

\subsection{From circuits to pinned SOP instances}\label{sec:circ2sop}
For each $q$-qubit circuit, we remove final measurements, preserve global phase, and pin the all-zero input and output as in \autoref{sec:sop}.
This gives
\[
A_C=\braket{0^q|C|0^q} =R_C^{-1}\!\!\sum_{\bx\in\{0,1\}^V}\!\!
\omega_r^{f(\bx)}, \text{ where }f(\bx) = c+\sum_{v}b_vx_v+\frac{r}2\sum_{\{u,v\}\in E}x_ux_v
\]
where $G_C=(V,E)$ contains the structure.
Its free variables are the $n=|V|$ of \autoref{sec:sop}, distinct from the qubit count $q$.
The import is exact whenever every phase has an exact finite-modulus representation.
Otherwise the importer moves each remaining phase to the nearest tick of a common modulus, at a target rounding error of $\varepsilon=10^{-8}$.
That is finer than the amplitude precision used in large-scale supremacy verification~\cite{PanChenZhang22}.

\paragraph{Import and arithmetic: exact or not?}
The \emph{import} is exact or rounded as just described.
The \emph{arithmetic} that then evaluates the SOP is exact when it counts residues over $\Z$, and numerical when it evaluates a single Fourier mode using
extended-precision complex types, whose mantissa is $64$ bits.
The treewidth, rank-width and linear-layout configurations measured below fix a single Fourier mode, so they use the numerical arithmetic on both cohorts.
The branch strategy keeps the exact residue count whenever the vector of $r$ counts fits its memory budget.

Beyond the width of the type, two devices hold that evaluation in range, and only the rounding they leave behind is measured.
Each $\omega_r^{ak}$ is evaluated from its angle rather than read from a table, so a modulus near $10^{15}$ costs no space and no root of unity inherits the error of an earlier one.
Every table the dynamic program passes upward is rescaled so that its largest entry lies in $[1,2)$, with the discarded power of two accumulated in an integer exponent.
That exponent also carries the final $R_C^{-1}=2^{-m_\H/2}$, as an integer halving and at most one factor $2^{-1/2}$.
Scaling by a power of two is exact in binary floating point, so the mantissas are the ones an unscaled evaluation would compute, and the amplitude stays in range at every instance size in the corpus.
What is left is rounding, and the prototype carries a worst-case bound for it beside the amplitude, adding a fixed multiple of the unit roundoff per complex operation.
The weighted-model-counting route is numerical in a different way: the exporter writes double-precision complex literal weights and Ganak returns a complex count, with no exponent kept beside it. (This is the source of the reconstructions that underflow to exactly zero in \autoref{tab:experiment-approx}.)

\paragraph{The error bound of a rounded import.}
Rounding the phase $\theta$ of one diagonal gate to $\tilde\theta$ replaces that gate by another at operator-norm distance $|e^{i\theta}-e^{i\tilde\theta}|$, and these distances telescope along the circuit, so
\begin{equation}\label{eq:rounding-bound}
\bigl|\braket{0^q|C|0^q}-\braket{0^q|\widetilde C|0^q}\bigr|
\ \le\ \Delta\ =\!\!\sum_{g\ \text{rounded}}\!\!\bigl|e^{i\theta_g}-e^{i\tilde\theta_g}\bigr| .
\end{equation}
The importer accumulates $\Delta$ as it goes.
It starts at the least multiple of $16$ above $\pi P/\varepsilon$, for $P$ the number of phases it expects to round, then doubles the modulus and imports again until $\Delta\le\varepsilon$, so the bound is verified rather than predicted.
This is why every rounded modulus below is divisible by $16$ and none is a power of two.
The bound $\Delta$ covers the rounding alone, and the roundoff of the evaluation is the separate bound of the previous paragraph.
We therefore report a bound on the phase rounding and a separate bound on the evaluation roundoff, instead of one error bound on the amplitude.
A run counts as correct when its error against the reference amplitude falls within the tolerance the benchmark protocol fixes for every case, $10^{-8}$ absolute plus $10^{-8}$ relative.

\subsection{Evaluation strategies}
\paragraph{Native SOP routes.}
These evaluate the SOP directly on $G_C$.
\emph{Treewidth bucket elimination} uses a tree decomposition.
The \emph{branching rank-width} DP uses a rank-decomposition, and the \emph{linear-layout} DP processes a layout from left to right.
The \emph{branch} strategy first applies the exact $[\H\H]$ and $[\omega]$ path-sum simplifications of \cite{HuangEtAl26} and splits connected components.
Width and cost estimates select treewidth, rank-width, or the \emph{QPF} route for each residual component, with bounded cutset conditioning as a fallback.
The QPF route evaluates a component entirely as quadratic phase functions (\autoref{def:qpf}), which is the stabilizer-rank side of \autoref{sec:stabjoin} taken on its own.
The prototype picks one route per component.
Route and arithmetic are independent choices: each route runs either over the residue counts or on a single Fourier mode, in the sense of \autoref{sec:circ2sop}.
A reversible trail carries sparse incremental state in the spirit of Knuth's dancing cells~\cite{Knuth2025DancingCells}.

\paragraph{Decompositions.}
Tree decompositions use standard min-fill and min-degree orders~\cite{BodlaenderKoster10}.
Rank-decompositions use layout- and cut-based heuristics over GF$(2)$.
The resulting generated decomposition widths are upper bounds, while exact rank-width is computed only for a calibration subset with $n\le12$~\cite{OumSeymour06}.
We separate decomposition time from kernel time, the time spent computing the amplitude, because search can dominate small instances.

\paragraph{Weighted model counting.}
The WMC route encodes the pinned SOP as a single weighted model-counting formula $F_C$.
For each edge $e=\{u,v\}$, it introduces an auxiliary sign variable $s_e$ and the implications $s_e\to x_u$ and $s_e\to x_v$.
It gives the true literals of $x_v$ and $s_e$ weights $\omega_r^{b_v}$ and $\omega_r^\eta-1$, respectively, with all false literals weighted by $1$.
Unless $x_u=x_v=1$, the implications force $s_e=0$.
When both endpoints are $1$, summing over $s_e$ gives $1+(\omega_r^\eta-1)=\omega_r^\eta=-1$, and hence supplies the factor $(-1)^{x_ux_v}$.
The weighted model counter Ganak evaluates this complex weighted count, after which the normalization gives the amplitude directly~\cite{DBLP:conf/ijcai/SharmaRSM19}.
Before export, degree-zero and degree-one variables are eliminated analytically.
On three of the $410$ solved instances, a dense-block variant instead represents a complete bipartite block by one parity per side and a single two-clause sign factor.
A $54$-case diagnostic sweep covers the instances the primary configuration leaves unsolved.
On it the three-clause strict encoding and the two-clause encoding above each solve $41$ cases and the block variant solves $44$, all at the process-time floor.
This differs from circuit-level and hybrid path-sum/WMC encodings~\cite{MeiBonsangueLaarman24,MeiEtAl26QuokkaSharp,HuangEtAl26}.

\paragraph{External baselines.}
For the same %
amplitude, the external baselines are the Qiskit Aer dense statevector simulator~\cite{JavadiAbhariEtAl24Qiskit} and the MQT DDSIM decision-diagram simulator~\cite{ZulehnerWille19DDSIM}.

\subsection{Benchmarks and protocol}
\label{sec:bench}
The corpus has four source families.
ALG85 combines algorithm-query circuits from QASMBench~\cite{LiEtAl23QASMBench}, MQT Bench, SupermarQ, and several handcrafted cases.
InferQ603 is a deterministic stratum-balanced InferQ sample from work on quantum data management~\cite{HaiEtAl25QuantumDataNISQ}.
The MQT cohorts contain unitary MQT Bench circuits~\cite{QuetschlichBurgholzerWille23MQTBench,WilleEtAl24MQTHandbook} from \url{https://www.cda.cit.tum.de/mqtbench/}.
SupermarQ contributes optimized application circuits~\cite{TomeshEtAl22SupermarQ}.
Of $974$ entries, import produces $422$ exact and $552$ rounded instances.
Deduplication leaves $420$ exact cases for the primary comparison (\autoref{tab:experiment-coverage}) and $550$ rounded cases.
Their $n$ distribution is strongly bimodal.
Of the exact cases, $397$ have $n\le32$, while $23$ have $n>32$, and the maximum is $163{,}359$.
Eight definite-basis mid-measurement ALG85 cases support coherent SOP import but cannot be evaluated by the external simulators for the same fixed-boundary amplitude.

A fixed selection rule picks a 200-case cohort from the rounded imports, stratified by logarithmic $n$, min-fill width, and modulus magnitude, and covering all 21 unique cases in the width-20--40 tail.
It reaches $n=3320$ and min-fill width $40$, at moduli between $4.0\mathbin{\cdot}10^{10}$ and $9.8\mathbin{\cdot}10^{14}$, with $\Delta$ never worse than $7.8\mathbin{\cdot}10^{-10}$.

A fifth family is synthetic and targets the WMC route directly.
It places Hadamard layers around random doubly-controlled-$Z$ ($\CCZ$) gates and $\T$ phases, with five seeds per size, and grows the circuit while holding this structure fixed.
A $\CCZ$ is not itself quadratic, so the importer expands it by the standard seven-phase Clifford$+\T$ identity: a $\T$ phase on each of the three qubits, a $\T^\dagger$ phase on each of the three pairwise parities, and one more $\T$ phase on the triple parity.
Those parities come from $\CX$ conjugation, which substitutes affinely into the wire expression and adds no variable, while a phase landing on a parity of two or more variables materializes that parity with two fresh variables and sign edges alone.
Every exported instance is therefore the quadratic sign-edge SOP of \autoref{eq:sop}, since the exporter refuses any cross-term coefficient other than $\eta=r/2$.
The family thus raises the density and the variable count of a quadratic polynomial rather than its degree, and every width reported for it is measured on the expanded $G_C$.
That is the regime where a compact model-counting formula and a mature counter with component caching were conjectured to scale better end-to-end than our prototype width-driven dynamic programs.
RQ4 tests that conjecture.

For the primary exact comparison, every configuration computes $\braket{0^q|C|0^q}$ in a fresh process, and times are end to end.
Each case has an $8$-GiB memory cap and a $120$-s timeout, increased to $1800$~s for larger ALG85 circuits.
We report PAR-2, which charges an unsolved case twice its timeout.\footnote{The machine is an eight-core Intel i7-11850H running one job and one software thread.
The repeated-run median coefficient of variation is $0.028$.}
Most cases are trivial after pinning, so the roughly $0.26$-s median mainly measures process and import overhead.
PAR-2 distinguishes the unsolved tail.
The rounded cohort uses the same per-process caps, but the native routes read pinned SOP instances imported once in advance, so their times exclude the import that Aer and DDSIM pay on every run.
For that cohort we report solve coverage, not times.
The measurements behind every table and figure below, together with the scripts that derive them, are archived on Zenodo~\cite{Artifact26}.

\begin{table}[tbp]
\centering \small
\begin{tabular}{lrrrrrrr}
\toprule
Source & Raw & Unique & Exact & Rounded & Exact \% & Med. $n$ & Max $n$ \\
\midrule
ALG85     & 154 & 150 &  76 &  78 & 49.4 &      1 &      98 \\
InferQ603 & 150 & 150 &  36 & 114 & 24.0 &      0 &     282 \\
MQT-easy  & 271 & 271 & 126 & 145 & 46.5 &      1 & 163,359 \\
MQT-fixed &  33 &  33 &   2 &  31 &  6.1 & 24,332 &  25,396 \\
MQT2040   & 258 & 258 &  97 & 161 & 37.6 &      1 &       1 \\
SupermarQ & 108 & 108 &  85 &  23 & 78.7 &      1 &      35 \\
\midrule
Total     & 974 & 970 & 422 & 552 & 43.3 &      1 & 163,359 \\
\bottomrule
\end{tabular}

\vspace{6pt}
\small
\begin{tabular}{lrrrrrr}
\toprule
Free-variable stratum $n$ & 0--32 & 33--64 & 65--128 & 129--256 & 257--512 & $>512$ \\
\midrule
Unique exact cases        &   397 &      3 &       2 &        0 &        3 &     15 \\
\bottomrule
\end{tabular}

\vspace{6pt}
\small
\begin{tabular}{lrrrrr}
\toprule
Modulus stratum $r$ & $8$ & $2^4$--$2^{16}$ & $2^{17}$--$2^{32}$ & $10^{10}$--$10^{12}$ & $\ge10^{13}$ \\
\midrule
Imported (exact + rounded) cases  & 320 &              76 &                 26 &                  269 &          283 \\
\bottomrule
\end{tabular}
\caption{Benchmark provenance and exact-first import coverage for the all-zero fixed boundary.
\emph{Raw} counts benchmarks and \emph{Unique} counts them after deduplication.
Every benchmark imports: each import is either \emph{Exact} (a finite-modulus instance) or \emph{Rounded} (under the bound of \autoref{eq:rounding-bound} at $\varepsilon=10^{-8}$), with \emph{Exact \%} the exact share.
Here $n$ is the number of free SOP variables left after pinning the boundary.
The lower block distributes the $420$ unique exact cases across free-variable strata, while \autoref{tab:experiment-approx} reports a stratified rounded cohort.
The last block distributes the $974$ imports across modulus strata: every exact modulus is a power of two, no rounded one is, and all of the latter are divisible by $16$.}
\label{tab:experiment-coverage}
\end{table}
\begin{table}[tbp]
\centering \small \setlength{\tabcolsep}{7pt}
\begin{tabular}{lrrrl}
\toprule
Configuration & Solved & TO & Other & Main non-solved outcome \\
\midrule
SOP branch        & 186 &  8 &   6 & cutset budget exhausted \\
SOP treewidth     & 186 &  0 &  14 & out of memory, or killed at the cap \\
SOP rank-width    & 179 &  0 &  21 & no decomposition within the memory budget \\
SOP QPF route     &   0 &  0 & 200 & modulus above the $2^{32}$ ceiling \\
\midrule
WMC (Ganak)       & 179 &  4 &  17 & reconstruction underflowed to zero \\
\midrule
Aer (statevector) & 156 &  5 &  39 & statevector larger than the cap \\
DDSIM             & 135 & 57 &   8 & harness failure, or killed at the cap \\
\bottomrule
\end{tabular}
\caption{Solve coverage on the 200-case rounded-import cohort.
\emph{Solved} counts runs that produced an amplitude, \emph{TO} counts timeouts, and \emph{Other} counts the outcomes named in the last column.
The refusals differ in kind: treewidth fails to allocate its table, whereas the rank-width heuristic finds no decomposition it can afford.
The 17 WMC entries are cases whose reconstruction underflowed to exactly zero.}
\label{tab:experiment-approx}
\end{table}

\begin{figure}[tbp]\centering\includegraphics[width=.76\linewidth]{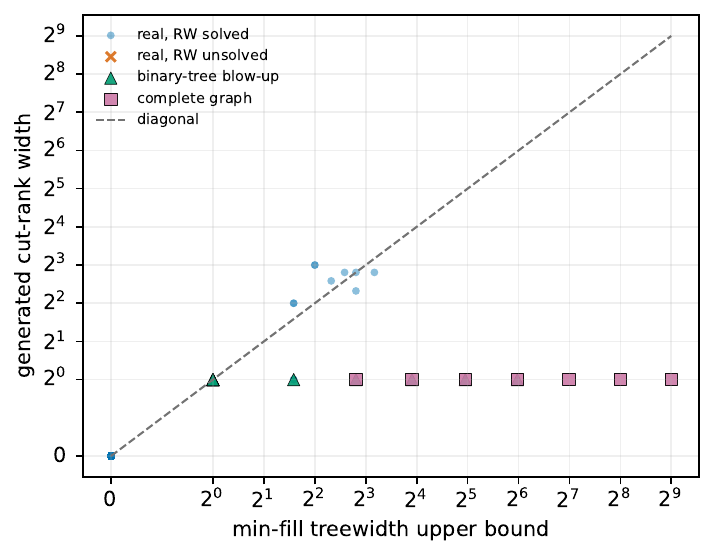}\caption{Per-instance min-fill treewidth upper bound (horizontal) against the generated rank-decomposition cut-rank width (vertical), with the $y=x$ diagonal. Real-corpus points lie on or above the diagonal, so the heuristic rank-decompositions do not go below the treewidth bound and treewidth stays preferable there. The targeted families $K_n$ (complete graphs) and $\Gamma_{3,t}$ (blow-ups) are the complementary exception, sitting far below the diagonal at rank-width $1$. Marker shape distinguishes outcome and family. Both axes are symlog with base $2$ and a linear threshold at $1$, so a label such as $2^5$ is a tick and not a width value.}\label{fig:experiment-structure}\end{figure}
\begin{figure}[tbp]\centering\includegraphics[width=.78\linewidth]{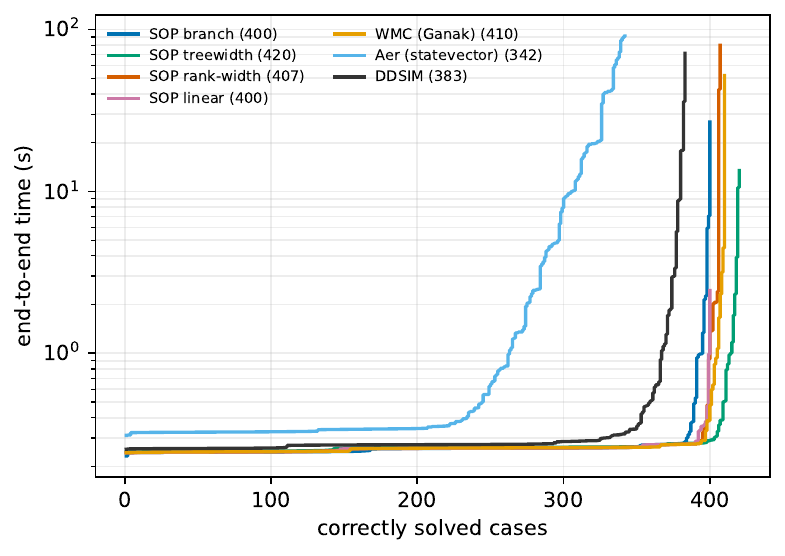}\caption{End-to-end ``cactus profile'' (a.k.a. survival plot) on the $420$-case exact corpus under the common cap: lower and farther to the right is better, and each legend count is the number of cases that configuration solved correctly within the cap.}\label{fig:experiment-cactus}\end{figure}
\begin{figure}[tbp]\centering\includegraphics[width=.76\linewidth]{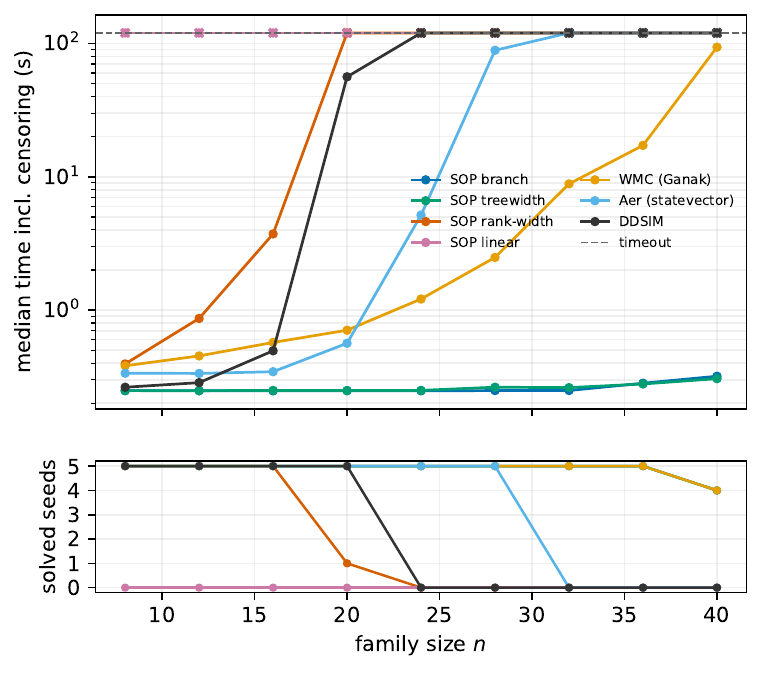}\caption{Phase-polynomial scaling study. Each timeout is replaced by the $120$-s cap before the median is taken, which is what \emph{median time incl.\ censoring} means. Crosses mark the censored observations and the lower panel reports the number of solved seeds.}\label{fig:experiment-wmc-scaling}\end{figure}

\subsection{Results}\label{sec:results}
On this corpus, \emph{treewidth is empirically dominant, while rank-width is theoretically and synthetically complementary}.
Rank-width wins on structured families with controlled decompositions but not on the real benchmarks.
Its tables can be large even at low width, while treewidth dynamic programs are simpler and easier to optimize.
\begin{description}
  \item[RQ1 (rounding supported, coverage uneven).] Exact import covers all four sources and $420$ unique cases, including instances with more than $512$ free variables.
Only one of six size strata contains at least $20$ nontrivial cases, so exact size coverage is uneven.
On the stratified rounded cohort, the branch strategy solves $186/200$ cases despite moduli of order $10^{10}$--$10^{14}$ (\autoref{tab:experiment-approx}).
Of these solves, $132$ are delegated to a treewidth program and $54$ are evaluated whole on a single Fourier mode.
All $1087$ available cross-route checks and $931$ external-reference checks fall within the tolerance of \autoref{sec:circ2sop}.

  \item[RQ2 (SOP baselines).] Treewidth solves all $420$ cases, and neither rank-width nor linear layout solves an additional corpus case.
Among the $23$ cases with $n>32$, treewidth solves $23$, rank-width $10$, and linear layout $3$.
Complete graphs $K_n$ and clique blow-ups $\Gamma_{3,t}$ expose the complementary regime: both have rank-width $1$, but treewidth $n-1$ and $2t-1$, and rank-width continues after bucket elimination exhausts memory (\autoref{tab:experiment-families}).
Such parameter separation is rare in the corpus.
There, the generated rank-decompositions never fall below the min-fill treewidth upper bound, so the branch strategy assigns its few nontrivial residual components to treewidth (\autoref{tab:experiment-structure}, \autoref{fig:experiment-structure}).

  \item[RQ3 (complementary).] Aer and DDSIM solve $342$ and $383$ cases, respectively.
The union of native SOP routes solves all $420$, including nine solved by neither external simulator.
When both routes succeed, the paired median ratio of the best external to best SOP time is $1.11$ under an oracle best-of-each comparison.
This median is dominated by a common per-process floor, and Aer and DDSIM compute richer state representations than one fixed amplitude, so we read it only as parity on the easy mass.
The native advantage is in the unsolved tail (\autoref{tab:experiment-primary}, \autoref{fig:experiment-cactus}).

  \item[RQ4 (not supported).] WMC solves $410$ of $420$ corpus cases but none beyond the native SOP routes.
On the synthetic scaling family no sustained crossover occurs (\autoref{fig:experiment-wmc-scaling}), so compact formulas do not improve end-to-end scaling.
\end{description}

\paragraph{On stabilizer rank.}
The QPF route never applies here: our implementation caps the phase order at $2^{32}$ while rounding puts every modulus in the cohort above $10^{10}$.
A quadratic phase function also carries only fourth roots of unity, so a vertex is magic unless $4b_v\equiv0\pmod r$, and a rounded phase almost never is: $92\%$ of the cohort's variables are magic, at a median of $98\%$ per instance, against the handful of $\T$ phases an exact Clifford$+\T$ import leaves behind.
Rounding destroys the very structure the route exploits, so $\mathrm{sr}(\tau)$ is exponential in the instance size whatever the arithmetic permits.
The branch strategy therefore uses a treewidth program or a whole-instance single-Fourier evaluation throughout.
With the magic price out of reach at every vertex, the recurrence of \autoref{eq:switch-recurrence} chooses the empty antichain, which is the branching rank-width DP measured here.
The optimization therefore adds no further route on this corpus, while \autoref{prop:mixed-separation} identifies a family where a nonempty switching antichain pays off.

\paragraph{Where the two prices cross.}
The two prices can still be separated on a family built to move them independently: a dense random Clifford graph on $n$ qubits, whose width grows with $n$, carrying $\tau$ added $\T$ gates (\autoref{tab:experiment-magic}).
The QPF route wins the high-width, low-magic corner outright.
At $n\ge44$ it is the only route that finishes, since bucket elimination exhausts the memory cap and the rank-width heuristic returns no decomposition inside the timeout, while the QPF route answers in under $12$ seconds through $\tau=30$.
It cedes once its $6^{\tau/6}$ terms outgrow the cheaper of the two width costs, and the route that takes over is the one with the smaller exponent: treewidth at $n=28$, where $2^{18}<4^{10}$, and rank-width at $n=36$, where $4^{12}<2^{26}$.
That is the boundary $\mathrm{sr}(\tau)=\min(2^{\tw},4^{\rw})$ the two prices predict.
At $n=28$ the margins are tens of milliseconds on instances every route solves, so the crossover there is a constant-factor effect, and only the $n\ge36$ corner separates the guarantees themselves.
Like the $K_n$ and $\Gamma_{3,t}$ separations, this is a structural demonstration and not a real-world benchmark advantage.

\begin{table}[tbp]
\centering \small \setlength{\tabcolsep}{5.5pt}
\begin{tabular}{llrrrrrr}
\toprule
 & & \multicolumn{6}{c}{magic count $\tau$} \\
\cmidrule(lr){3-8}
Instance & Route & $0$ & $6$ & $12$ & $18$ & $24$ & $30$ \\
 & QPF terms $6^{\tau/6}$ & $1$ & $6$ & $36$ & $216$ & $1296$ & $7776$ \\
\midrule
$n=28$          & treewidth  & 0.03 & 0.05 & 0.02 & \textbf{0.03} & \textbf{0.04} & --- \\
$\tw\,18$--$19$ & rank-width & 0.05 & 0.05 & 0.04 & 0.05 & 0.05 & --- \\
$\rw\,10$       & QPF        & \textbf{0.01} & \textbf{0.01} & \textbf{0.01} & 0.04 & 0.11 & --- \\
\midrule
$n=36$          & treewidth  & 5.81 & 2.96 & 5.09 & 4.78 & 7.00 & 2.36 \\
$\tw\,25$--$26$ & rank-width & 1.40 & 1.22 & 1.26 & 1.18 & 1.24 & \textbf{1.27} \\
$\rw\,12$       & QPF        & \textbf{0.02} & \textbf{0.02} & \textbf{0.04} & \textbf{0.09} & \textbf{0.35} & 1.72 \\
\midrule
$n=44$          & treewidth  & OOM & OOM & OOM & OOM & OOM & OOM \\
$\tw\,32$--$33$ & rank-width & TO & TO & TO & TO & TO & TO \\
                & QPF        & \textbf{0.02} & \textbf{0.02} & \textbf{0.06} & \textbf{0.21} & \textbf{0.96} & \textbf{4.57} \\
\midrule
$n=52$          & treewidth  & OOM & OOM & OOM & OOM & OOM & OOM \\
$\tw\,40$--$41$ & rank-width & TO & TO & TO & TO & TO & TO \\
                & QPF        & \textbf{0.01} & \textbf{0.03} & \textbf{0.09} & \textbf{0.46} & \textbf{2.23} & \textbf{11.19} \\
\bottomrule
\end{tabular}
\caption{Magic against width: end-to-end seconds on a dense random Clifford graph on $n$ qubits (a Hadamard layer, each $\CZ$ present with probability $\tfrac12$, a second Hadamard layer) carrying $\tau$ added $\T$ gates, so the two prices move independently.
One instance per cell, a single Fourier mode at $r=8$, the same $8$-GiB cap and a $30$-s timeout.
\emph{OOM} covers both a refusal to allocate the table and a kill at the cap, and \emph{TO} is the timeout.
The fastest route in each column is in bold.
Here $\tw$ is the min-fill upper bound and $\rw$ the generated cut-rank width, absent at $n\ge44$ because no decomposition was produced inside the cap.
The QPF route is the root-switch endpoint of \autoref{sec:stabjoin} on the whole instance, not the antichain optimization, and it uses the elementary $6^{\tau/6}$ stabilizer decomposition rather than the sharper cited bound.
The family places at most one $\T$ per qubit, so the cell $n=28$, $\tau=30$ does not exist.}
\label{tab:experiment-magic}
\end{table}

\begin{table}[tbp]
\centering \small \setlength{\tabcolsep}{5pt}
\begin{tabular}{lrrrrrrrr}
\toprule
Configuration & Correct & TO & OOM & \shortstack{Ref./\\unsup.} & \shortstack{PAR-2\\(s)} & \shortstack{Median\\(s)} & \shortstack{P90\\(s)} & \shortstack{Med.\ RSS\\(MiB)} \\
\midrule
SOP branch        & 400 & 20 & 0 &  0 &  11.79 & 0.258 &  0.272 & 71.3 \\
SOP treewidth     & 420 &  0 & 0 &  0 &   0.34 & 0.261 &  0.276 & 71.3 \\
SOP rank-width    & 407 & 13 & 0 &  0 &   7.99 & 0.259 &  0.272 & 71.2 \\
SOP linear        & 400 & 13 & 7 &  0 &  27.68 & 0.259 &  0.272 & 71.3 \\
\midrule
WMC (Ganak)       & 410 & 10 & 0 &  0 &   6.12 & 0.260 &  0.273 & 71.3 \\
\midrule
Aer (statevector) & 342 &  6 & 0 & 72 & 112.40 & 0.340 & 11.623 & 83.4 \\
DDSIM             & 383 & 25 & 0 & 12 &  93.83 & 0.272 &  0.339 & 91.0 \\
\bottomrule
\end{tabular}
\caption{Primary exact-corpus outcomes on the $420$ cases under the common cap ($8$~GiB, $120$~s; $1800$~s for ALG85).
\emph{Correct} is the number of cases solved correctly; \emph{TO}, \emph{OOM}, and \emph{Ref./unsup.}\ count timeouts, out-of-memory, and refusals or unsupported cases (harness failures, and inputs whose dense statevector would exceed the memory cap).
PAR-2 charges an unsolved case twice its timeout; \emph{Median} and \emph{P90} are end-to-end wall-clock seconds and \emph{Med.\ RSS} is median resident memory.
All configurations attempted every case.
Every returned amplitude agreed with the reference value, and no configuration uniquely solved a case.}
\label{tab:experiment-primary}
\end{table}
\begin{table}[tbp]
\centering \footnotesize \setlength{\tabcolsep}{5pt} \textbf{Panel A: initial and generated widths}\\[2pt]
\begin{tabular}{lrr rrrr rrrr}
\toprule
 & & & \multicolumn{4}{c}{TW upper bound} & \multicolumn{4}{c}{Generated RW width} \\
\cmidrule(lr){4-7}\cmidrule(lr){8-11}
Cohort & Cases & RW obs. & min & mean & med & max & min & mean & med & max \\
\midrule
All real exact & 420 & 407 & 0.0 &   0.7 &  0.0 &   26.0 & 0.0 & 0.1 & 0.0 & 8.0 \\
Branch solved  & 400 & 387 & 0.0 &   0.7 &  0.0 &   26.0 & 0.0 & 0.1 & 0.0 & 8.0 \\
Branch hard    &  70 &  57 & 0.0 &   3.4 &  0.0 &   26.0 & 0.0 & 0.0 & 0.0 & 1.0 \\
Targeted (A/B) &  22 &  19 & 1.0 & 103.9 & 11.0 & 1023.0 & 1.0 & 1.0 & 1.0 & 1.0 \\
\bottomrule
\end{tabular}

\vspace{6pt}
\textbf{Panel B: branch events}\\[2pt]
\begin{tabular}{lrrrrc}
\toprule
Cohort & Deleg.\ (any) & $\to$TW & $\to$RW & Fallthr. & Search nodes (med/max) \\
\midrule
All real exact & 22 & 22 &  0 & 10 & 1/11 \\
Branch solved  & 22 & 22 &  0 & 10 & 1/11 \\
Branch hard    & 13 & 13 &  0 & 10 & 1/11 \\
Targeted (A/B) & 18 &  8 & 10 & 35 & 1/67 \\
\bottomrule
\end{tabular}
\caption{Structural and branch behavior across cohorts.
\emph{Panel A}: each width is summarized as min/mean/median/max; TW is the min-fill treewidth upper bound and RW the cut-rank width of the generated rank-decomposition, and \emph{RW obs.}\ is the number of cases for which a rank-decomposition was generated (missing observations are not imputed).
\emph{Panel B}: \emph{Deleg.}\ counts components the branch strategy delegated to a treewidth ($\to$TW) or rank-width ($\to$RW) program, \emph{Fallthr.}\ residual fallthroughs, and \emph{Search nodes}\ the residual-search size (median/max).
Cohorts are all real exact cases, those the branch strategy solved, the $70$ hardest, and the targeted A/B families (binary-tree/blow-up and complete-graph) of \autoref{subsec:separation-linear-tensor}.}
\label{tab:experiment-structure}
\end{table}
\begin{table}[tbp]
\centering \small \setlength{\tabcolsep}{6pt}
\begin{tabular}{lrr rr rr}
\toprule
 & & & \multicolumn{2}{c}{width} & \multicolumn{2}{c}{end-to-end (s)} \\
\cmidrule(lr){4-5}\cmidrule(lr){6-7}
Instance & $n$ & $|E|$ & $\tw$ & $\rw$ & treewidth & rank-width \\
\midrule
\multicolumn{7}{l}{\emph{Complete graphs $K_n$\quad($\rw=1$, min-fill $\tw=n-1$)}} \\
$K_{16}$   &   16 &     120 &   15 & 1 & 0.25 & 0.25 \\
$K_{32}$   &   32 &     496 &   31 & 1 & OOM  & 0.25 \\
$K_{512}$  &  512 & 130{,}816 &  511 & 1 & OOM  & 14.5 \\
$K_{1024}$ & 1024 & 523{,}776 & 1023 & 1 & OOM  & TO   \\
\midrule
\multicolumn{7}{l}{\emph{Clique blow-ups $\Gamma_{3,t}$\quad($\rw=1$, min-fill $\tw=2t-1$)}} \\
$\Gamma_{3,8}$  & 120 &  1{,}316 &  15 & 1 & 0.27 & 0.30 \\
$\Gamma_{3,16}$ & 240 &  5{,}384 &  31 & 1 & OOM  & 0.93 \\
$\Gamma_{3,32}$ & 480 & 21{,}776 &  63 & 1 & OOM  & 9.39 \\
$\Gamma_{3,64}$ & 960 & 87{,}584 & 127 & 1 & OOM  & TO   \\
\bottomrule
\end{tabular}
\caption{Where rank-width is genuinely complementary to treewidth: two families whose min-fill treewidth grows while their rank-width stays $1$.
Widths are the min-fill treewidth upper bound $\tw$ and the generated cut-rank width $\rw$.
The last two columns give end-to-end seconds for the treewidth and rank-width routes, or OOM/TO for out-of-memory (the $8$-GiB cap) and timeout ($120$~s).
Once treewidth passes $\approx30$, its route exhausts memory while the rank-width route remains at width $1$.}
\label{tab:experiment-families}
\end{table}

\paragraph{Limitations.}
We compute only the all-zero fixed-boundary amplitude, not a full state or a sweep over sampled boundaries.
The primary performance analysis uses exact finite-modulus imports.
The rounded cohort reports solve coverage under the bound of \autoref{eq:rounding-bound}, not cross-system performance.
Widths are exact only for the $n\le12$ calibration subset, and decomposition search can dominate amplitude evaluation.
The synthetic families isolate structural mechanisms rather than representative workloads.
Comparisons with tensor-network, FeynmanDD/path-sum, and rank-width ZX contractors, and scaling with supplied decompositions until treewidth fails, remain future work.

\FloatBarrier

\section{Consequences for quantum-circuit simulation}\label{sec:consequences}

\autoref{sec:dp} gave the rank-width algorithm, and \autoref{subsec:separation-linear-tensor} the circuit families on which it provably beats the linear-layout and tensor-contraction guarantees (\autoref{cor:separating-family}).
We obtain a rank-width fixed-parameter strong-simulation algorithm for Clifford$+\T$ and, through the gadgets of \autoref{thm:gadget-universality}, for a much larger gate-set family, and we place the SOP route alongside decision-diagram and weighted-model-counting simulators.

\subsection{Rank-width FPT simulation}\label{sec:rw-fpt}

Putting \autoref{sec:sop} and \autoref{sec:dp} together gives the simulation statement.
For a circuit $C$ over $\{\H,\T,\CZ\}$, the pinned amplitude $\langle\bz|C|\by\rangle$ is $R_C^{-1}\delta_C(\by,\bz)S(f_C)$, so evaluating the SOP evaluates the amplitude.
Given a rank-decomposition of $G_C$ of width $k$, \autoref{cor:single-amplitude} therefore computes one amplitude in $O(4^k\poly(n))$\lean{Formal/Paper.lean}{amplitude_by_rank_dp_normalized}, and \autoref{thm:fourier-speedup} computes the full residue histogram at a factor $r$ more~\cite{WangChengYuanJi25,DBLP:journals/corr/abs-2510-06775}, which \autoref{cor:histogram-one-mode} removes whenever $r$ is a power of two, as it is for Clifford$+\T$.
With a linear layout of cut-rank at most $k$ instead, \autoref{thm:linear-layout-fourier} improves the base to $2^k$.
For a fixed finite gate set, $r$ is a constant, so the bit complexity of this exact amplitude simulation remains $2^{O(k)}\poly(n)$.
Deriving the SOP from the circuit is one pass over its elementary gates, and nothing after that pass is specific to circuits: the same bounds hold for an arbitrary quadratic SOP, since evaluation uses only the graph, its coefficients and the decomposition.

The decomposition is part of the input.
Known algorithms find fixed-width rank-decompositions in FPT time~\cite{OumSeymour06,HlinenyOum08,Oum17Survey}, making the end-to-end problem FPT in rank-width.
Their dependence on $k$ is large, so the prototype uses heuristic decompositions and includes their construction cost (\autoref{sec:experiments}).

Decision-diagram simulators process variables in a fixed order~\cite{WangChengYuanJi25,DBLP:journals/corr/abs-2510-06775}.
The SOP analog is the base-$2$ linear-layout DP of \autoref{thm:linear-layout-fourier}.
Branching rank-width is more general because $\rw(G)\le\lrw(G)$, with an unbounded gap on the families of \autoref{cor:separating-family}.

\subsection{Diagonal gates and the \texorpdfstring{Clifford$+\T$}{Clifford+T} gate set}\label{sec:clifford-t}

The construction uses two features that extend beyond $\{\H,\T,\CZ\}$.

First, $\H$ is the only nondiagonal gate.
Since $\H_{p,q}=2^{-1/2}(-1)^{pq}$, it contributes the quadratic sign $(-1)^{xx'}=\omega_r^{(r/2)xx'}$ between consecutive path values and one factor $2^{-1/2}$.
The latter enters $R_C$, while the sign gives an edge of $G_C$.
A $\CZ$ contributes the same sign between wires.

Second, $\T$ may be replaced by a diagonal single-qubit gate
\[
\D_{\alpha,\beta}=\begin{bmatrix}\alpha&0\\0&\beta\end{bmatrix}
\]
without changing $G_C$.
It contributes $\alpha$ at $x=0$ and $\beta$ at $x=1$, so factoring out $\alpha$ leaves the unary weights $1$ and $\beta/\alpha$.
Unitarity gives $\beta/\alpha=e^{\phi i}$.
When all such ratios are powers of a common $\omega_r$, they are precisely the coefficients $\omega_r^{b_v}$ in \autoref{eq:sop}.
This includes $\mathrm S$, $\T$, $\mathrm Z$, global phases, and $\mathrm Z$-rotations by rational multiples of $\pi$.

Thus circuits over $\{\H,\D_{\alpha,\beta},\CZ\}$ yield \autoref{eq:sop}.
The $\H$ and $\CZ$ gates supply cross-terms, diagonal gates supply unary coefficients, and fully pinned interactions enter $c$.
Since local Hadamards interchange $\CZ$ and $\CX$ (\autoref{sec:quantum-c}), and $\H$, $\mathrm S=\D_{1,i}$, and $\CX$ generate Clifford circuits~\cite{NielsenChuang10}, this includes Clifford$+\T$ with $r=8$.
The bounds therefore apply directly to these native circuits.
Two special cases deserve separate mention.
On Clifford instances no vertex carries magic, so a root switch in the stabilizer-rank optimization collapses the whole computation to one Gauss sum on any graph (\autoref{sec:stabjoin}), which rederives the polynomial-time exact-amplitude consequence of the Gottesman--Knill theorem.
Independently of stabilizer structure and of graph width, every even Fourier mode of the residue histogram is polynomial by the direct factorization of \autoref{prop:even-modes}.
Other fixed finite gate sets use the refined-modulus gadgets of the next subsection (\autoref{thm:gadget-universality}), with rank-width measured after expansion.

\subsection{Expressivity beyond Clifford+T}\label{sec:gadget-universality}

The construction of \autoref{sec:sop} is native for $\{\H,\T,\CZ\}$: Hadamards and $\CZ$ gates give quadratic terms, while $\T$ gates give unary phases.
This is the clean setting for the rank-width algorithm of \autoref{sec:dp}, which needs only this native form, so this subsection may be skipped on a first reading.
SOPs in this form nevertheless reach well beyond that gate set.
In fact, gates whose unitary matrix entries lie in a common dyadic-cyclotomic ring can all be encoded into a quadratic SOP.
The same rings govern exact synthesis: for a power-of-two degree $r\ge8$, and allowing clean ancillas, the multiqubit Clifford-cyclotomic gate set of degree $r$ synthesizes exactly the unitaries whose entries lie in $\Z[1/2,\omega_r]$, which is the ring $\Z[\omega_r,1/\sqrt2]$ below~\cite{DBLP:conf/rc/AmyGKMMR24}.
The proof uses two standard finite-Fourier facts: the Walsh expansion of a Boolean truth table in the parity-character basis, and the character-orthogonality identity enforcing a parity constraint.
Combining them gives a quadratic SOP gadget.

For an even modulus $r$, write
\[
\Dcycl_r=\Z[\omega_r,1/\sqrt2].
\]
This is the finitely generated dyadic-cyclotomic coefficient ring generated by the allowed root of unity and the Hadamard normalization.

Fix an even modulus $R$.
A \emph{sign-edge phase polynomial} in Boolean variables $\bu$ is one of the form
\[
f(\bu)=c+\sum_i a_i u_i+\frac R2\sum_{i<j}\epsilon_{ij}\,u_iu_j\pmod R, \qquad \epsilon_{ij}\in\{0,1\},
\]
the natural extension to modulus $R$ of the pinned exponent of \autoref{eq:sop}: a constant, one unary phase per variable, and a sign edge $\{i,j\}$ exactly where $\epsilon_{ij}=1$.
A \emph{sign-edge SOP gadget} for a fixed $k$-qubit gate $U$ is an \emph{exact} identity
\[
\langle \bz|U|\by\rangle = 2^{-h/2}\sum_{\bw\in\{0,1\}^m}\omega_R^{f_U(\by,\bz,\bw)}
\]
for some integers $h,m\ge0$, in which $f_U$ is a sign-edge phase polynomial in $\bu=(\by,\bz,\bw)$.
Here $\by,\bz\in\{0,1\}^k$ are the boundary bits (input $\by$, output $\bz$, matching the convention of \autoref{sec:quantum-c}), and the entries of $\bw$ are auxiliary Boolean variables.
The gadget has \emph{constant size} when $h$, $m$, and the number of terms of $f_U$ depend only on the fixed gate $U$.

The construction rests on two elementary finite Fourier steps.
A finite table of dyadic-cyclotomic amplitudes becomes a finite sum of roots of unity of one common order, and an arbitrary finite phase table becomes a linear combination of Boolean parities after refining the modulus.
The parities are then enforced by quadratic constraints on auxiliary Boolean variables, by the identity below.

\begin{lemma}[Parity averaging identity]\label{lem:parity-constraint-gadget}
Let $R$ be even and let $p$ and $(u_i)_{i\in T}$ be Boolean.
Then
\[
\frac12\sum_{\lambda\in\{0,1\}} \omega_R^{(R/2)\lambda\left(p+\sum_{i\in T}u_i\right)} =
\begin{cases}
1 & \text{if } p=\bigoplus_{i\in T}u_i,\\
0 & \text{otherwise.}
\end{cases}
\]
The exponent has only unary terms and quadratic terms with coefficient $R/2$.
\end{lemma}

\begin{proof}
Since $R$ is even, $\omega_R^{R/2}=-1$, so the summand equals $(-1)^{\lambda(p+\sum_{i\in T}u_i)}$ and the average over $\lambda\in\{0,1\}$ is $\frac12\bigl(1+(-1)^{p+\sum_{i\in T}u_i}\bigr)$.
This is $1$ when $p+\sum_{i\in T}u_i$ is even, that is when $p=\bigoplus_{i\in T}u_i$, and $0$ otherwise.
Expanding the exponent gives the quadratic terms $(R/2)\lambda p$ and $(R/2)\lambda u_i$, each with coefficient $R/2$.
\end{proof}

\begin{restatable}[Gadget universality for dyadic-cyclotomic gate sets]{theorem}{thmgadgetuniv}\label{thm:gadget-universality}
A fixed finite gate set $\Gate$ of bounded arity admits exact constant-size sign-edge SOP gadgets over one finite even modulus if and only if all its gate entries lie in a common dyadic-cyclotomic ring $\Z[\omega_r,1/\sqrt2]$.
Every pinned circuit over such a $\Gate$ then has a quadratic SOP of size linear in its number of elementary gates.
\end{restatable}

\begin{proof}[Proof sketch]
For each gate, rewrite its entries over one common root list indexed by \emph{selector bits} $\bs$.
Then express the resulting finite phase table as a linear combination of Boolean parities at a refined modulus $R$.
Finally, \autoref{lem:parity-constraint-gadget} enforces each non-singleton parity through a fresh parity variable and an averaging variable, using only unary terms and sign edges.
Composing gates adds exponents over the common $R$ with disjoint auxiliaries, and pinning specializes the boundary.
Conversely, a finite sign-edge gadget evaluates to an element of $\Z[\omega_R,1/\sqrt2]$.
Those constants depend only on the gate-set description, but they need not be small.
\appref{app:gadget-wmc} bounds them and gives the construction.
\end{proof}

The equivalence is exact for finite-modulus sign-edge SOP gadgets.
Gates with entries outside every dyadic-cyclotomic ring are not covered by this finite-modulus statement.
They require approximation, symbolic weights, or a different coefficient model.
\autoref{thm:gadget-universality} proves the \emph{existence} of constant-size gadgets for a fixed symbolic gate set.
It does not by itself provide a numerical procedure that recognizes arbitrary floating-point matrices as members of a dyadic-cyclotomic ring, extracts root-list representations, and constructs gadgets.
The implementation therefore supports an explicit symbolic gate catalog rather than a generic importer (\autoref{sec:experiments}).
The statement is about expressivity only.
Some gates have much smaller native descriptions if one allows arbitrary quadratic coefficients.
A two-qubit diagonal gate acts as $\ket{ij}\mapsto\omega_r^{\alpha_{ij}}\ket{ij}$, and its exponent, as a function of $(i,j)\in\{0,1\}^2$, expands as
\[
\alpha_{ij}=\alpha_{00}+(\alpha_{10}-\alpha_{00})\,i+(\alpha_{01}-\alpha_{00})\,j+\Delta\,ij\pmod r, \qquad \Delta=\alpha_{11}-\alpha_{10}-\alpha_{01}+\alpha_{00},
\]
where the cross-term coefficient $\Delta$ is the discrete second derivative of the phase table.
The pinned SOP form of \autoref{eq:sop} admits a quadratic term only with coefficient $r/2$, so the gate is already native exactly when $\Delta\in\{0,r/2\}$: the cross term is then absent or is a single sign edge.
Otherwise no plain sign edge reproduces it, and \autoref{thm:gadget-universality} gives a constant-size gadget replacement after a fixed modulus refinement, provided the entries lie in a dyadic-cyclotomic ring.

For example, CCZ is covered by the theorem because its entries are eighth roots of unity.
One may encode it by introducing parity variables for the relevant XORs, but this gadgetization can change rank-width, as the next paragraph makes precise.

Replacing a non-native gate by the gadget of \autoref{thm:gadget-universality} adds only $O(1)$ auxiliary vertices per gate, expanding $G_C$ into a larger \emph{gadget-expanded} graph $G_C^{\mathrm{gad}}$.
Treewidth tolerates this, but the right route is through the tensor line graph: a bounded set of attachment endpoints need not occur together in any bag of an arbitrary tree decomposition of $G_C$, so inserting the auxiliaries there is not sound in general.

\begin{restatable}[Gadget-expanded treewidth]{corollary}{corgadgettw}\label{cor:gadget-treewidth}
For every fixed finite gate set $\Gate$ admitting the gadgets of \autoref{thm:gadget-universality} there is a constant $c_{\Gate}$ such that
\[
\tw\bigl(G_C^{\mathrm{gad}}\bigr)\le\max\{\cc(N_C),\,c_{\Gate}\}.
\]
\end{restatable}

\begin{proof}[Proof sketch]
Work in the line graph $L(N_C)$.
The bond indices incident on one gate tensor pairwise co-occur in that tensor, so they form a clique in $L(N_C)$, and every clique of a graph is contained in some bag of every tree decomposition.
Attach to such a bag a fresh leaf bag holding those $O(1)$ boundary indices together with the gate's $O(1)$ auxiliary gadget variables.
Every gadget edge then lies inside this constant-size leaf bag.
This yields a tree decomposition of an intermediate graph in which only the new leaf bags grew, to size at most $c_{\Gate}$.
The pinning, contractions, deletions, and cancellations that produce the pinned gadget-expanded SOP graph realize $G_C^{\mathrm{gad}}$ as a minor of that intermediate graph, which inherits the bound.
\appref{app:gadget-wmc} carries this out and bounds $c_{\Gate}$.
\end{proof}

This preserves the $2^{O(\cc(N_C))}\poly(n)$ simulation guarantee of \autoref{sec:treewidth} for every fixed gadgetized gate set, on the gadget-expanded instance.
Rank-width does not tolerate gadgetization in general, since the added vertices can change the rank-width by an unbounded amount.
For non-native gate sets one must therefore read rank-width on the gadget-expanded graph, not on the original $G_C$.
The caveat is vacuous for the native gate set $\{\H,\T,\CZ\}$, which already produces SOPs of the form in \autoref{eq:sop} without any gadgets.

\subsection{Simulation via WMC also matches contraction complexity}\label{sec:feynmandd}

The WMC route encodes the amplitude as a weighted count~\cite{QPWMC,MeiBonsangueLaarman24,DBLP:conf/fm/QuistMCL24,HuangEtAl26}.
To keep the theoretical statement separate from any particular tool encoding, define the \emph{strict algebraic encoding} $F_C^{\mathrm{strict}}$: one Boolean variable $x_v$ per SOP variable; for each graph edge $e=\{u,v\}$ one auxiliary variable $s_e$ with the three standard clauses of $s_e\leftrightarrow(x_u\wedge x_v)$; and the algebraic weights
\[
w(x_v{=}1)=\omega_r^{b_v},\qquad w(x_v{=}0)=1,\qquad w(s_e{=}1)=-1,\qquad w(s_e{=}0)=1,
\]
with $\omega_r^{c}$ applied as an external scalar.
On a satisfying assignment the weights multiply to $\omega_r^{f(\bx)}$ term by term, since each surviving cross term contributes its sign through $s_e$ and $\eta=r/2$.
Thus the algebraic weighted count of $F_C^{\mathrm{strict}}$ is exactly $S(f)$.
The incidence graph $I(F)$ has a vertex per variable and clause, and an edge per occurrence.
Define $F_C^{\mathrm{soft}}$ as the encoding used in \autoref{sec:experiments}: drop the ternary clause from each edge gadget of $F_C^{\mathrm{strict}}$ and replace the true weight of $s_e$ by $\omega_r^\eta-1$.
Summing over each $s_e$ shows that its algebraic weighted count is also $S(f)$.

\begin{restatable}[Algebraic-WMC incidence treewidth]{lemma}{lemvartoinc}\label{lem:variable-to-incidence}
If $E\ne\varnothing$, then
\[
\tw\bigl(I(F_C^{\mathrm{soft}})\bigr)=\max(\tw(G_C),1)
\quad\text{and}\quad
\tw\bigl(I(F_C^{\mathrm{strict}})\bigr)\le \max(\tw(G_C),3).
\]
If $E=\varnothing$, both $G_C$ and $I(F_C^{\mathrm{soft}})$ are edgeless and have treewidth $0$, while the strict bound still holds.
\end{restatable}

\noindent Together with $\tw(G_C)\le\cc(N_C)$, this lets standard model-counting DP match Markov--Shi in time $2^{O(\cc(N_C))}\poly(n)$~\cite{BacchusDP03,SamerSzeider10} for either encoding, whose size is linear in the explicit SOP size $n+|E|$.
The runs in \autoref{sec:experiments} use $F_C^{\mathrm{soft}}$ except when the exporter selects the block variant on a few dense instances.
The soft runs are therefore covered by the sharper equality, while the block runs are not covered by the lemma.

Rank-width must instead be measured on $G_C$, since sign variables can increase it.
For example, the clique $K_n$ has rank-width $1$, whereas its once-subdivided graph $S(K_n)$ has rank-width $\Omega(n)$.
The graph $S(K_n)$ occurs induced in the strict incidence graph of the $K_n$ instance (\appref{app:gadget-wmc}).
A corresponding lower bound holds directly for the soft incidence graph, which subdivides every edge of $K_n$ three times.
A balanced cut of the original clique vertices exposes an $\Omega(n)$-rank submatrix through the subdivision vertices, so that incidence graph has rank-width $\Omega(n)$ as well.
Thus WMC exposes the SOP, but rank-width belongs on its variable graph.

\section{Conclusion}\label{sec:conclusion}

Rank-width is a usable parameter for exact quantum-circuit simulation.
On the SOP variable graph it computes an amplitude with $O(4^k\poly(n))$ arithmetic operations, or $O(2^k\poly(n))$ from a linear layout, and at a fixed power-of-two modulus it returns the full residue histogram for the price of one amplitude.
It sits below linear rank-width and within one of the Markov--Shi contraction complexity, and strictly below both on families we construct.
Where the tables also carry stabilizer structure, one switch per branch trades the width price for a magic price, which recovers Gottesman--Knill without giving up the rank-width bound.
The prototype finds treewidth bucket elimination the strongest broad baseline, with the rank-width and stabilizer routes complementary on the structured corners that separate the parameters.

We have compared representations only for one amplitude.
A broader knowledge-compilation comparison should also study succinctness, composition, equivalence checking, and simplification.
FeynmanDD represents higher-degree phase polynomials as decision diagrams~\cite{WangChengYuanJi25,DBLP:journals/corr/abs-2510-06775}.
Reduction rules for more general path sums support model-counting-based equivalence checking~\cite{HuangEtAl26}.
Relating these operations to quadratic rank-width remains open.

\section*{Author Contributions}

All authors contributed equally to the development of the results and to the preparation of the manuscript.
AI was used for the initial literature survey, the writing, and implementation tuning.

\section*{Acknowledgments}

This work was supported by the FWO/NWO KR2iQS project (G076326N), a collaboration between the University of Antwerp and Leiden University.

J. H. Lee was supported by the Dutch Research Council (NWO) under the project \emph{Boosting the Search for New Quantum Algorithms with AI} (BoostQA), file number NGF.1623.23.033, research programme Quantum Technologie 2023.

A. de Colnet is supported by the Netherlands Organization for Scientific Research (NWO/OCW) as part of Quantum Limits (project number SUMMIT.1.1016).

\bibliographystyle{quantum}
\bibliography{refs}

\begin{thebibliography}{10}

\bibitem{MarkovShi08}
Igor~L. Markov and Yaoyun Shi.
\newblock ``Simulating quantum computation by contracting tensor networks''.
\newblock \href{https://dx.doi.org/10.1137/050644756}{SIAM Journal on Computing {\bf 38}, 963--981}~(2008).

\bibitem{WangChengYuanJi25}
Ziyuan Wang, Bin Cheng, Longxiang Yuan, and Zhengfeng Ji.
\newblock ``{FeynmanDD}: Quantum circuit analysis with classical decision diagrams''.
\newblock In Computer Aided Verification.
\newblock \href{https://dx.doi.org/10.1007/978-3-031-98685-7_2}{Volume 15934 of Lecture Notes in Computer Science, pages 28--52}.
\newblock Springer~(2025).

\bibitem{Dawson+2005}
Christopher~M. Dawson, Andrew~P. Hines, Duncan Mortimer, Henry~L. Haselgrove, Michael~A. Nielsen, and Tobias~J. Osborne.
\newblock ``Quantum computing and polynomial equations over the finite field $\mathbb{Z}_2$''.
\newblock \href{https://dx.doi.org/10.26421/QIC5.2-2}{Quantum Info. Comput. {\bf 5}, 102--112}~(2005).

\bibitem{Amy19PathSum}
Matthew Amy.
\newblock ``Towards large-scale functional verification of universal quantum circuits''.
\newblock In Proceedings of the 15th International Conference on Quantum Physics and Logic.
\newblock \href{https://dx.doi.org/10.4204/EPTCS.287.1}{Volume 287 of Electronic Proceedings in Theoretical Computer Science, pages 1--21}.
\newblock ~(2019).
\newblock  \href{http://arxiv.org/abs/1805.06908}{arXiv:1805.06908}.

\bibitem{DBLP:journals/corr/abs-2510-06775}
Bin Cheng, Ziyuan Wang, Ruixuan Deng, Jianxin Chen, and Zhengfeng Ji.
\newblock ``Breaking the treewidth barrier in quantum circuit simulation with decision diagrams''.
\newblock CoRR{\bf abs/2510.06775}~(2025).
\newblock  url:~\url{https://arxiv.org/abs/2510.06775}.

\bibitem{QPWMC}
Dirck van~den Ende, Joon~Hyung Lee, Alfons Laarman, and Henning Basold.
\newblock ``Quantum physics using weighted model counting''~(2026).
\newblock  \href{http://arxiv.org/abs/2508.21288}{arXiv:2508.21288}.

\bibitem{MeiBonsangueLaarman24}
Jingyi Mei, Marcello~M. Bonsangue, and Alfons Laarman.
\newblock ``Simulating quantum circuits by model counting''.
\newblock In Computer Aided Verification.
\newblock \href{https://dx.doi.org/10.1007/978-3-031-65633-0_25}{Pages 555--578}.
\newblock Lecture Notes in Computer Science. Springer~(2024).

\bibitem{DBLP:conf/fm/QuistMCL24}
Arend{-}Jan Quist, Jingyi Mei, Tim Coopmans, and Alfons Laarman.
\newblock ``Advancing quantum computing with formal methods''.
\newblock In {FM} {(2)}.
\newblock \href{https://dx.doi.org/10.1007/978-3-031-71177-0_25}{Pages 420--446}.
\newblock Lecture Notes in Computer Science. Springer~(2024).

\bibitem{HuangEtAl26}
Wei-Jia Huang, Christophe Chareton, Yu-Fang Chen, Kai-Min Chung, Min-Hsiu Hsieh, Alfons Laarman, and Jingyi Mei.
\newblock ``Equivalence checking of quantum circuits via path-sum and weighted model counting''.
\newblock In Tools and Algorithms for the Construction and Analysis of Systems (TACAS 2026).
\newblock \href{https://dx.doi.org/10.1007/978-3-032-22749-2_21}{Volume 16506 of Lecture Notes in Computer Science, pages 419--439}.
\newblock Springer~(2026).

\bibitem{oum2008rank}
Sang-il Oum.
\newblock ``Rank-width is less than or equal to branch-width''.
\newblock \href{https://dx.doi.org/10.1002/jgt.20280}{Journal of Graph Theory {\bf 57}, 239--244}~(2008).

\bibitem{Oum17Survey}
Sang il~Oum.
\newblock ``Rank-width: Algorithmic and structural results''.
\newblock \href{https://dx.doi.org/10.1016/j.dam.2016.08.006}{Discrete Applied Mathematics {\bf 231}, 15--24}~(2017).

\bibitem{BravyiGosset16}
Sergey Bravyi and David Gosset.
\newblock ``Improved classical simulation of quantum circuits dominated by {C}lifford gates''.
\newblock \href{https://dx.doi.org/10.1103/PhysRevLett.116.250501}{Physical Review Letters {\bf 116}, 250501}~(2016).

\bibitem{BravyiEtAl19}
Sergey Bravyi, Dan Browne, Padraic Calpin, Earl Campbell, David Gosset, and Mark Howard.
\newblock ``Simulation of quantum circuits by low-rank stabilizer decompositions''.
\newblock \href{https://dx.doi.org/10.22331/q-2019-09-02-181}{Quantum {\bf 3}, 181}~(2019).

\bibitem{QassimPashayanGosset21}
Hammam Qassim, Hakop Pashayan, and David Gosset.
\newblock ``Improved upper bounds on the stabilizer rank of magic states''.
\newblock \href{https://dx.doi.org/10.22331/q-2021-12-20-606}{Quantum {\bf 5}, 606}~(2021).

\bibitem{VandenNestMDB06}
Maarten~Van den Nest, Akimasa Miyake, Wolfgang D\"ur, and Hans~J. Briegel.
\newblock ``Universal resources for measurement-based quantum computation''.
\newblock \href{https://dx.doi.org/10.1103/PhysRevLett.97.150504}{Physical Review Letters {\bf 97}, 150504}~(2006).

\bibitem{VandenNestDVB07}
Maarten~Van den Nest, Wolfgang D\"ur, Guifr\'e Vidal, and Hans~J. Briegel.
\newblock ``Classical simulation versus universality in measurement-based quantum computation''.
\newblock \href{https://dx.doi.org/10.1103/PhysRevA.75.012337}{Physical Review A {\bf 75}, 012337}~(2007).

\bibitem{HiraishiImai14}
Hidefumi Hiraishi and Hiroshi Imai.
\newblock ``{BDD} operations for quantum graph states''.
\newblock In Reversible Computation (RC 2014).
\newblock \href{https://dx.doi.org/10.1007/978-3-319-08494-7_17}{Volume 8507 of Lecture Notes in Computer Science, pages 216--229}.
\newblock Springer~(2014).

\bibitem{HiraishiImaiIwataLin18}
Hidefumi Hiraishi, Hiroshi Imai, Yoichi Iwata, and Bingkai Lin.
\newblock ``Parameterized algorithms to compute {I}sing partition function''.
\newblock \href{https://dx.doi.org/10.1587/transfun.E101.A.1398}{IEICE Transactions on Fundamentals of Electronics, Communications and Computer Sciences {\bf E101.A}, 1398--1403}~(2018).

\bibitem{KuyanovKissinger26}
Fedor Kuyanov and Aleks Kissinger.
\newblock ``Efficient classical simulation of low-rank-width quantum circuits using {ZX}-calculus''~(2026).
\newblock  \href{http://arxiv.org/abs/2603.06764}{arXiv:2603.06764}.

\bibitem{CodsiLaakkonen26}
Julien Codsi and Tuomas Laakkonen.
\newblock ``Unifying graph measures and stabilizer decompositions for the classical simulation of quantum circuits''~(2026).
\newblock  \href{http://arxiv.org/abs/2603.06377}{arXiv:2603.06377}.

\bibitem{Artifact26}
Alexis de~Colnet, Floris Geerts, Rihan Hai, Alfons Laarman, Joon~Hyung Lee, and Guillermo~A. P{\'e}rez.
\newblock ``Data and source-code artifact for {Quadratic Sums-of-Powers for Fixed-Parameter Tractable Quantum-Circuit Simulation}''.
\newblock Zenodo~(2026).
\newblock \url{https://doi.org/10.5281/zenodo.21992784}.

\bibitem{Arnborg87}
Stefan Arnborg, Derek~G. Corneil, and Andrzej Proskurowski.
\newblock ``Complexity of finding embeddings in a k-tree''.
\newblock \href{https://dx.doi.org/https://doi.org/10.1137/0608024}{SIAM J. Algebraic Discrete Methods {\bf 8}, 277–284}~(1987).

\bibitem{Bodlaender96}
Hans~L. Bodlaender.
\newblock ``A linear-time algorithm for finding tree-decompositions of small treewidth''.
\newblock \href{https://dx.doi.org/10.1137/S0097539793251219}{{SIAM} J. Comput. {\bf 25}, 1305--1317}~(1996).

\bibitem{OumSeymour06}
Sang il~Oum and Paul Seymour.
\newblock ``Approximating clique-width and branch-width''.
\newblock \href{https://dx.doi.org/10.1016/j.jctb.2005.10.006}{Journal of Combinatorial Theory, Series B {\bf 96}, 514--528}~(2006).

\bibitem{HlinenyOum08}
Petr Hlin{\v{e}}n{\'y} and Sang il~Oum.
\newblock ``Finding branch-decompositions and rank-decompositions''.
\newblock \href{https://dx.doi.org/10.1137/070685920}{SIAM Journal on Computing {\bf 38}, 1012--1032}~(2008).

\bibitem{AdlerK15}
Isolde Adler and Mamadou~Moustapha Kant{\'e}.
\newblock ``Linear rank-width and linear clique-width of trees''.
\newblock \href{https://dx.doi.org/10.1016/j.tcs.2015.04.021}{Theoretical Computer Science {\bf 589}, 87--98}~(2015).

\bibitem{NielsenChuang10}
Michael~A. Nielsen and Isaac~L. Chuang.
\newblock ``Quantum computation and quantum information''.
\newblock \href{https://dx.doi.org/10.1017/CBO9780511976667}{Cambridge University Press}. ~(2010).
\newblock 10th anniversary edition.

\bibitem{GilesSelinger13}
Brett Giles and Peter Selinger.
\newblock ``Exact synthesis of multiqubit {C}lifford+{T} circuits''.
\newblock \href{https://dx.doi.org/10.1103/PhysRevA.87.032332}{Physical Review A {\bf 87}, 032332}~(2013).

\bibitem{CooleyTukey65}
James~W. Cooley and John~W. Tukey.
\newblock ``An algorithm for the machine calculation of complex {F}ourier series''.
\newblock \href{https://dx.doi.org/10.1090/S0025-5718-1965-0178586-1}{Mathematics of Computation {\bf 19}, 297--301}~(1965).

\bibitem{BravyiSmithSmolin16}
Sergey Bravyi, Graeme Smith, and John~A. Smolin.
\newblock ``Trading classical and quantum computational resources''.
\newblock \href{https://dx.doi.org/10.1103/PhysRevX.6.021043}{Physical Review X {\bf 6}, 021043}~(2016).

\bibitem{DehaeneDeMoor03}
Jeroen Dehaene and Bart~De Moor.
\newblock ``{C}lifford group, stabilizer states, and linear and quadratic operations over {GF}(2)''.
\newblock \href{https://dx.doi.org/10.1103/PhysRevA.68.042318}{Physical Review A {\bf 68}, 042318}~(2003).

\bibitem{VandenNest10}
Maarten~Van den Nest.
\newblock ``Classical simulation of quantum computation, the {G}ottesman--{K}nill theorem, and slightly beyond''.
\newblock \href{https://dx.doi.org/10.26421/QIC10.3-4-6}{Quantum Information and Computation {\bf 10}, 258--271}~(2010).

\bibitem{CaiCLL10}
Jin-Yi Cai, Xi~Chen, Richard~J. Lipton, and Pinyan Lu.
\newblock ``On tractable exponential sums''.
\newblock In Frontiers in Algorithmics (FAW).
\newblock \href{https://dx.doi.org/10.1007/978-3-642-14553-7_16}{Volume 6213 of Lecture Notes in Computer Science, pages 148--159}.
\newblock Springer~(2010).
\newblock  \href{http://arxiv.org/abs/1005.2632}{arXiv:1005.2632}.

\bibitem{Dechter99}
Rina Dechter.
\newblock ``Bucket elimination: A unifying framework for reasoning''.
\newblock \href{https://dx.doi.org/10.1016/S0004-3702(99)00059-4}{Artificial Intelligence {\bf 113}, 41--85}~(1999).

\bibitem{DBLP:journals/jair/Ganian0SS22}
Robert Ganian, Eun~Jung Kim, Friedrich Slivovsky, and Stefan Szeider.
\newblock ``Sum-of-products with default values: Algorithms and complexity results''.
\newblock \href{https://dx.doi.org/10.1613/JAIR.1.12370}{J. Artif. Intell. Res. {\bf 73}, 535--552}~(2022).

\bibitem{Oum05RankWidthVertexMinors}
Sang il~Oum.
\newblock ``Rank-width and vertex-minors''.
\newblock \href{https://dx.doi.org/10.1016/j.jctb.2005.03.003}{Journal of Combinatorial Theory, Series B {\bf 95}, 79--100}~(2005).

\bibitem{BremnerJozsaShepherd10}
Michael~J. Bremner, Richard Jozsa, and Dan~J. Shepherd.
\newblock ``Classical simulation of commuting quantum computations implies collapse of the polynomial hierarchy''.
\newblock \href{https://dx.doi.org/10.1098/rspa.2010.0301}{Proceedings of the Royal Society A {\bf 467}, 459--472}~(2011).
\newblock  \href{http://arxiv.org/abs/1005.1407}{arXiv:1005.1407}.

\bibitem{PanChenZhang22}
Feng Pan, Keyang Chen, and Pan Zhang.
\newblock ``Solving the sampling problem of the {Sycamore} quantum circuits''.
\newblock \href{https://dx.doi.org/10.1103/PhysRevLett.129.090502}{Physical Review Letters {\bf 129}, 090502}~(2022).

\bibitem{Knuth2025DancingCells}
Donald~E. Knuth.
\newblock ``The art of computer programming, volume 4, fascicle 7: Constraint satisfaction''.
\newblock Addison-Wesley Professional. ~(2025).
\newblock  url:~\url{https://www-cs-faculty.stanford.edu/~knuth/taocp.html}.

\bibitem{BodlaenderKoster10}
Hans~L. Bodlaender and Arie M. C.~A. Koster.
\newblock ``Treewidth computations i. upper bounds''.
\newblock \href{https://dx.doi.org/10.1016/j.ic.2009.03.008}{Information and Computation {\bf 208}, 259--275}~(2010).

\bibitem{DBLP:conf/ijcai/SharmaRSM19}
Shubham Sharma, Subhajit Roy, Mate Soos, and Kuldeep~S. Meel.
\newblock ``{GANAK:} {A} scalable probabilistic exact model counter''.
\newblock In Sarit Kraus, editor, Proceedings of the Twenty-Eighth International Joint Conference on Artificial Intelligence, {IJCAI} 2019, Macao, China, August 10-16, 2019.
\newblock \href{https://dx.doi.org/10.24963/IJCAI.2019/163}{Pages 1169--1176}.
\newblock ijcai.org~(2019).

\bibitem{MeiEtAl26QuokkaSharp}
Jingyi Mei, Dekel Zak, Muhammad Osama, Tim Coopmans, and Alfons Laarman.
\newblock ``Quokka\#: Quantum computing with \#{SAT}''~(2026).
\newblock  \href{http://arxiv.org/abs/2605.16509}{arXiv:2605.16509}.

\bibitem{JavadiAbhariEtAl24Qiskit}
Ali Javadi-Abhari, Matthew Treinish, Kevin Krsulich, Christopher~J. Wood, Jake Lishman, Julien Gacon, Simon Martiel, Paul~D. Nation, Lev~S. Bishop, Andrew~W. Cross, Blake~R. Johnson, and Jay~M. Gambetta.
\newblock ``Quantum computing with {Qiskit}''~(2024).
\newblock  \href{http://arxiv.org/abs/2405.08810}{arXiv:2405.08810}.

\bibitem{ZulehnerWille19DDSIM}
Alwin Zulehner and Robert Wille.
\newblock ``Advanced simulation of quantum computations''.
\newblock \href{https://dx.doi.org/10.1109/TCAD.2018.2834427}{IEEE Transactions on Computer-Aided Design of Integrated Circuits and Systems {\bf 38}, 848--859}~(2019).

\bibitem{LiEtAl23QASMBench}
Ang Li, Samuel Stein, Sriram Krishnamoorthy, and James Ang.
\newblock ``{QASMBench}: A low-level quantum benchmark suite for {NISQ} evaluation and simulation''.
\newblock \href{https://dx.doi.org/10.1145/3550488}{ACM Transactions on Quantum Computing{\bf 4}}~(2023).

\bibitem{HaiEtAl25QuantumDataNISQ}
Rihan Hai, Shih-Han Hung, Tim Coopmans, Tim Littau, and Floris Geerts.
\newblock ``Quantum data management in the {NISQ} era''.
\newblock \href{https://dx.doi.org/10.14778/3725688.3725701}{Proceedings of the {VLDB} Endowment {\bf 18}, 1720--1729}~(2025).

\bibitem{QuetschlichBurgholzerWille23MQTBench}
Nils Quetschlich, Lukas Burgholzer, and Robert Wille.
\newblock ``{MQT Bench}: Benchmarking software and design automation tools for quantum computing''.
\newblock \href{https://dx.doi.org/10.22331/q-2023-07-20-1062}{Quantum {\bf 7}, 1062}~(2023).

\bibitem{WilleEtAl24MQTHandbook}
Robert Wille, Lucas Berent, Tobias Forster, Jagatheesan Kunasaikaran, Kevin Mato, Tom Peham, Nils Quetschlich, Damian Rovara, Aaron Sander, Ludwig Schmid, Daniel Schoenberger, Yannick Stade, and Lukas Burgholzer.
\newblock ``The {MQT} handbook: A summary of design automation tools and software for quantum computing''.
\newblock In IEEE International Conference on Quantum Software (QSW).
\newblock ~(2024).
\newblock  \href{http://arxiv.org/abs/2405.17543}{arXiv:2405.17543}.

\bibitem{TomeshEtAl22SupermarQ}
Teague Tomesh, Pranav Gokhale, Victory Omole, Gokul~Subramanian Ravi, Kaitlin~N. Smith, Joshua Viszlai, Xin-Chuan Wu, Nikos Hardavellas, Margaret~R. Martonosi, and Frederic~T. Chong.
\newblock ``{SupermarQ}: A scalable quantum benchmark suite''.
\newblock In IEEE International Symposium on High-Performance Computer Architecture (HPCA).
\newblock \href{https://dx.doi.org/10.1109/HPCA53966.2022.00050}{Pages 587--603}.
\newblock ~(2022).

\bibitem{DBLP:conf/rc/AmyGKMMR24}
Matthew Amy, Andrew~N. Glaudell, Shaun Kelso, William Maxwell, Samuel~S. Mendelson, and Neil~J. Ross.
\newblock ``Exact synthesis of multiqubit clifford-cyclotomic circuits''.
\newblock In Torben~{\AE}gidius Mogensen and Lukasz Mikulski, editors, Reversible Computation - 16th International Conference, {RC} 2024, Toru{\'{n}}, Poland, July 4-5, 2024, Proceedings.
\newblock \href{https://dx.doi.org/10.1007/978-3-031-62076-8\_15}{Volume 14680 of Lecture Notes in Computer Science, pages 238--245}.
\newblock Springer~(2024).

\bibitem{BacchusDP03}
Fahiem Bacchus, Shannon Dalmao, and Toniann Pitassi.
\newblock ``Algorithms and complexity results for \#{SAT} and {Bayesian} inference''.
\newblock In Proceedings of the 44th Annual {IEEE} Symposium on Foundations of Computer Science.
\newblock \href{https://dx.doi.org/10.1109/SFCS.2003.1238208}{Pages 340--351}.
\newblock IEEE Computer Society~(2003).

\bibitem{SamerSzeider10}
Marko Samer and Stefan Szeider.
\newblock ``Algorithms for propositional model counting''.
\newblock \href{https://dx.doi.org/10.1016/j.jda.2009.06.002}{Journal of Discrete Algorithms {\bf 8}, 50--64}~(2010).

\bibitem{BodlaenderLRV10}
Hans~L. Bodlaender, Erik~Jan van Leeuwen, Johan M.~M. van Rooij, and Martin Vatshelle.
\newblock ``Faster algorithms on branch and clique decompositions''.
\newblock In Mathematical Foundations of Computer Science 2010.
\newblock \href{https://dx.doi.org/10.1007/978-3-642-15155-2_17}{Volume 6281 of Lecture Notes in Computer Science, pages 174--185}.
\newblock Springer~(2010).

\bibitem{vanRooijBodlaenderLeeuwenRossmanithVatshelle18}
Johan M.~M. van Rooij, Hans~L. Bodlaender, Erik~Jan van Leeuwen, Peter Rossmanith, and Martin Vatshelle.
\newblock ``Fast dynamic programming on graph decompositions''.
\newblock CoRR{\bf abs/1806.01667}~(2018).
\newblock  \href{http://arxiv.org/abs/1806.01667}{arXiv:1806.01667}.

\bibitem{BjorklundHusfeldtKaskiKoivisto07}
Andreas Bj{\"o}rklund, Thore Husfeldt, Petteri Kaski, and Mikko Koivisto.
\newblock ``Fourier meets {M}{\"o}bius: Fast subset convolution''.
\newblock In Proceedings of the Thirty-Ninth Annual {ACM} Symposium on Theory of Computing.
\newblock \href{https://dx.doi.org/10.1145/1250790.1250801}{Pages 67--74}.
\newblock {ACM}~(2007).

\bibitem{CourcelleKante09}
Bruno Courcelle and Mamadou~Moustapha Kant{\'e}.
\newblock ``Graph operations characterizing rank-width''.
\newblock \href{https://dx.doi.org/10.1016/j.dam.2008.08.026}{Discrete Applied Mathematics {\bf 157}, 627--640}~(2009).

\bibitem{Jelinek10Grid}
V{\'{\i}}t Jel{\'{\i}}nek.
\newblock ``The rank-width of the square grid''.
\newblock \href{https://dx.doi.org/10.1016/j.dam.2009.02.007}{Discrete Applied Mathematics {\bf 158}, 841--850}~(2010).

\bibitem{ODonnell14}
Ryan O'Donnell.
\newblock ``Analysis of {B}oolean functions''.
\newblock \href{https://dx.doi.org/10.1017/CBO9781139814782}{Cambridge University Press}. ~(2014).

\end{thebibliography}

\clearpage
\appendix

\section{The base of the branching join}\label{app:join-barrier}

A linear layout has no branching join to improve (\autoref{thm:linear-layout-fourier}).
The branching join is another matter.
The result below concerns the abstract algebraic join obtained at a branching vertex when the restriction maps and crossing form are treated as arbitrary black-box data.
Exploiting the structure of a particular graph or rank-decomposition escapes it, and \autoref{sec:stabjoin} shows what that buys.
In this abstract setting, let $U,V,W$ be $\F_2$-spaces of dimension at most $k$, let $P:U\to W$ and $Q:V\to W$ be the restriction maps, and let $B:U\times V\to\F_2$ be the crossing bilinear form, with input tables $F:U\to\mathbb C$ and $G:V\to\mathbb C$.
The join is
\[
H(w)=\sum_{Pu+Qv=w}F(u)G(v)(-1)^{B(u,v)}.
\]
To prove the following barrier theorem, we use the Walsh--Hadamard transform: for $g:\F_2^{m}\to\C$ it is $\hat g(\tau)=\sum_{\sigma}g(\sigma)(-1)^{\langle\tau,\sigma\rangle}$, computable in $O(m2^{m})$ operations, and it turns convolution over $\F_2^{m}$ into pointwise multiplication.

\begin{theorem}[Abstract generic join barrier]\label{thm:join-mm-barrier}
Suppose every such join, over arbitrary $P$, $Q$, $B$ and arbitrary input tables $F,G$, can be computed in $O(2^k\poly(k))$ arithmetic operations in characteristic zero.
Then square matrix multiplication has exponent $\mmomega=2$.
More generally, a bound $O(2^{\alpha k}\poly(k))$ implies $\mmomega\le\max\{2,2\alpha\}$.
\end{theorem}

\begin{proof}
Let $m\ge1$, let $N=2^m$, and let $A,B_0,C,D$ be $m$-dimensional vector spaces over $\F_2$.
Put
\[
U=A\oplus B_0,\qquad V=C\oplus D,\qquad W=A\oplus C,
\]
so the join width is $k=2m$.
Define
\[
P(a,b)=(a,0),\qquad Q(c,d)=(0,c), \qquad B((a,b),(c,d))=b\cdot d.
\]
Let $M$ and $N'$ be arbitrary $N\times N$ matrices, indexed as $M_{a,b}$ and $N'_{b,c}$.
Set
\[
F(a,b)=M_{a,b}, \qquad G(c,d)=2^{-m}\sum_{b\in B_0}N'_{b,c}(-1)^{b\cdot d}.
\]
Then, for $(a,c)\in A\oplus C$,
\[
\begin{aligned}
H(a,c) &= \sum_{b,d} M_{a,b}G(c,d)(-1)^{b\cdot d} \\ &= \sum_{b,b'}M_{a,b}N'_{b',c} \left(2^{-m}\sum_d(-1)^{(b+b')\cdot d}\right) \\ &= \sum_b M_{a,b}N'_{b,c}.
\end{aligned}
\]
Thus the join output is the matrix product $MN'$.
Forming $G$ is just $N$ Walsh--Hadamard transforms of length $N$, costing $O(N^2\log N)$.
A join algorithm running in $O(2^{\alpha k}\poly(k))$ would therefore multiply $N\times N$ matrices in $O(N^2\log N+N^{2\alpha}\polylog (N))$ arithmetic operations, so $\mmomega\le\max\{2,2\alpha\}$.
Taking $\alpha=1$ gives $\mmomega=2$.
\end{proof}

This theorem is therefore a barrier for a black-box base-$2$ improvement of arbitrary branching joins, and the child tables SOP instances produce may exhaust only part of the space of pairs $F,G$.
It leaves open a better base for our problem, matrix-multiplication speedups for the structured joins of specific graph decompositions, and \autoref{thm:linear-layout-fourier}, whose dynamic program has no join at all.

\begin{remark}[The hard maps are graph-realizable]
The maps and bilinear form used in the proof can be realized by an actual cut of a graph.
Take left vertices $a_i,b_i$, right vertices $c_i,d_i$, and outside vertices $y_i^A,y_i^C$ for $i=1,\ldots,m$.
Add edges $a_i y_i^A$, $c_i y_i^C$, and $b_i d_i$.
At the join with $X_L=\{a_i,b_i:i\in[m]\}$, $X_R=\{c_i,d_i:i\in[m]\}$, and $Y=\{y_i^A,y_i^C:i\in[m]\}$, the restrictions to $Y$ give $P(a,b)=(a,0)$ and $Q(c,d)=(0,c)$, while the crossing term is $b\cdot d$.
This reproduces the join geometry of the matrix-multiplication reduction.
The theorem remains a black-box statement because it allows arbitrary child tables $F$ and $G$.
\end{remark}

This is the same fine-grained setting as fast dynamic programming on branch decompositions, where generalized fast subset convolution and fast matrix multiplication can accelerate several problem-specific joins~\cite{BodlaenderLRV10,vanRooijBodlaenderLeeuwenRossmanithVatshelle18}.
Fast subset convolution itself, on an $n$-element ground set, improves the naive $3^n$ subset convolution to $2^n\poly(n)$ arithmetic operations over the sum-product ring~\cite{BjorklundHusfeldtKaskiKoivisto07}.
In the same spirit, Hiraishi, Imai, Iwata, and Lin obtain an $O^*(2^{\mmomega\operatorname{bw}(G)/2})$ branch-decomposition algorithm for the Ising partition function by casting its join as matrix multiplication~\cite{HiraishiImaiIwataLin18}, exactly the acceleration our barrier shows is unavoidable in the black-box regime.
Whether the present rank-signature join admits a general $2^{(\mmomega/2)k}\poly(k)$ implementation on the child tables that actually arise remains an interesting open implementation and theory question.

\section{Details of the linear-layout dynamic program}\label{app:linear-layout-proof}

\thmlinearlayoutfourier*

\begin{proof}
For $i=0,\ldots,n$, put $X_i=\{v_1,\ldots,v_i\}$ and $Y_i=V\setminus X_i$.
For a partial assignment $\bx\in\{0,1\}^{X_i}$, define its signature toward the unprocessed side by $\sigma_i(\bx)=\bx^{\mathsf T}\bA[X_i,Y_i]\in\F_2^{Y_i}$.
Let $\phi_i(\bx)$ contain all unary terms in $X_i$ and all quadratic sign terms with both endpoints in $X_i$.
For a fixed mode $a$, store the sparse table
\[
A_i^{(a)}[\sigma] = \sum_{\substack{\bx\in\{0,1\}^{X_i}\\ \sigma_i(\bx)=\sigma}} \omega_r^{a\phi_i(\bx)}.
\]
Since $\rank_{\F_2}\bA[X_i,Y_i]\le k$, at most $2^k$ signatures occur at every step.

Initialize $A_0^{(a)}[0]=1$.
Suppose the next vertex is $v=v_{i+1}$ and let $\delta\in\{0,1\}$ be its assigned value.
An old signature $\sigma\in\F_2^{Y_i}$ has a coordinate $\sigma(v)$, which is the parity of selected neighbors of $v$ in $X_i$.
Adding $v$ contributes the unary phase $b_v\delta$ and the internal edge phase $\eta\delta\sigma(v)$.
The new signature on $Y_{i+1}$ is $\sigma'=\sigma|_{Y_{i+1}}+\delta\bA[v,Y_{i+1}]$, so the update is
\[
A_{i+1}^{(a)}[\sigma'] \;{+}{=}\; A_i^{(a)}[\sigma]\,\omega_r^{a(b_v\delta+\eta\delta\sigma(v))}.
\]
Each occurring old signature has two extensions, so each step performs $O(2^k)$ table updates, up to polynomial-time manipulation of signature coordinates.

Initializing a leaf costs two table updates, one per assignment of its variable, and extending a prefix scans its table against the two values of the next variable, which gives the operation count.
At $i=n$, the only signature is empty and
\[
\sum_{\bx\in\{0,1\}^V}\omega_r^{a f(\bx)} = \omega_r^{ac}A_n^{(a)}[\varnothing].
\]
Running only $a=1$ gives one amplitude.
Running all $r$ modes and applying the inverse transform in \autoref{eq:inverse-dft} gives all residue counts.
If $\lrw(G)\le k$, apply the construction to the supplied witnessing layout.
\end{proof}

\section{Graph transformations behind the parameter comparison}\label{app:parameter-proofs}

\lemsopminorline*

\begin{proof}
We must show that $G_C$ arises from $L(N_C)$ by vertex deletions, edge deletions, and edge contractions.
The announced bound $\tw(G_C)\le\tw(L(N_C))=\cc(N_C)$ then follows by minor-monotonicity of treewidth.

Recall the differences between the two graphs.
The vertices of $L(N_C)$ are the bond indices of the tensor network of $C$, and two are adjacent when they occur on a common gate tensor.
The vertices of $G_C$ are the free segment variables $x_{i,j}$ of \autoref{sec:sop}.
We go from the first graph to the second in three steps.

\paragraph{Step 1 (collapse each wire segment).}
Group the bonds of $N_C$ by the wire segment that contains them.
Two bonds lie in the same segment when no Hadamard separates them, with the wire's input and output legs as its outer boundaries.
Inside a segment every gate is diagonal, namely a $\T$ gate or one leg of a $\CZ$, and so leaves the wire value unchanged.
The segment's bonds therefore carry a common value and form a connected path in $L(N_C)$, with consecutive bonds adjacent through the diagonal gate between them.
Contracting this path to a single vertex yields one vertex per segment that contains a bond.
A segment lying strictly between two Hadamards always contains a bond, and its vertex is the unpinned segment variable $x_{i,j}$.

\paragraph{Step 2 (delete the pinned segments).}
The two outer segments of each wire, before its first Hadamard and after its last, carry the pinned values $y_a$ and $z_a$.
Their open input and output indices are not vertices of $L(N_C)$ because open indices are not edges of $N_C$.
The only vertex an outer segment can contribute is formed in Step~1 from bonds outside the wire's outermost Hadamards.
We delete it and fold its contribution into the coefficients $b_v$ and the constant $c$, as in \autoref{sec:sop}.
When a wire begins and ends with a Hadamard, as every wire of \autoref{ex:circuit} does, its outer segments contain no bond and nothing is deleted.
The remaining vertices are exactly the vertex set $V$ of $G_C$.

\paragraph{Step 3 (delete canceled edges).}
Let $J$ be the simple graph on the vertices left after Steps~1 and 2.
Its edges come from the sign terms.
A $\CZ$ on wires $a,b$ and an interior Hadamard each place their two indices on a common tensor, which Step~1 contracts to the corresponding unpinned segment vertices.
No $\T$ gate joins distinct unpinned variables.
If $m_{uv}$ sign terms connect $u$ and $v$, they contribute $4m_{uv}x_ux_v$ to the pinned exponent modulo $8$.
This is an edge of $G_C$ when $m_{uv}$ is odd and cancels when $m_{uv}$ is even.
Hence $J$ is a supergraph of $G_C$, and $G_C$ is obtained from it by deleting the canceled edges.
\end{proof}

\lemrealizeanygraph*

\begin{proof}
Use one qubit per vertex $v\in W$.
Apply one $\H$ gate to every qubit, then one $\CZ$ gate between the wires of $u$ and $v$ for every edge $\{u,v\}\in F$, and finally one more $\H$ gate to every qubit.

On the wire of $v$, the two $\H$ gates create the three segment variables $x_{v,0},x_{v,1},x_{v,2}$.
The input and output basis states fix $x_{v,0}$ and $x_{v,2}$.
The two $\H$ interactions contribute the factors $(-1)^{x_{v,0}x_{v,1}}$ and $(-1)^{x_{v,1}x_{v,2}}$.
After substituting the pinned values, these are unary phases in $x_{v,1}$ and create no quadratic edges among the free variables.

Every $\CZ$ gate for an edge $\{u,v\}\in F$ lies between the two $\H$ layers.
The current variables on its wires are $x_{u,1}$ and $x_{v,1}$, so the gate contributes the quadratic factor $(-1)^{x_{u,1}x_{v,1}}$.
Thus the remaining variables are precisely $\{x_{v,1}:v\in W\}$, and two interact if and only if the corresponding vertices of $\Gamma$ are adjacent.
The SOP variable graph is therefore $\Gamma$.
\end{proof}

\lemblowupwidths*

\begin{proof}
Trees have rank-width at most $1$.
By Theorem~2 of Adler and Kant\'e, the linear rank-width of every forest equals its pathwidth.
Since complete binary trees have pathwidth $\Omega(h)$, $\lrw(B_h)=\Omega(h)$~\cite{AdlerK15}.
The passage from $B_h$ to $\Gamma_{h,t}$ is a true-twin substitution, which replaces each vertex by pairwise-adjacent vertices that share its open neighborhood.
Such substitutions are standard in the rank-width graph-operation view~\cite{CourcelleKante09}.
Here, take a width-$1$ rank-decomposition of $B_h$ and replace each leaf $v$ by a small subcubic tree whose leaves are the vertices of $K_t^v$.
Across cuts inherited from $B_h$, the cut matrix is obtained from that of $B_h$ by duplicating rows and columns, so the rank does not increase.
Across cuts that split a single twin clique, all rows on one side are equal, so the cut-rank is at most $1$.
These are the only cuts in the constructed decomposition because each substituted clique is attached as its own subtree at the corresponding leaf of the original decomposition.
Hence $\rw(\Gamma_{h,t})\le1$.

For the linear-rank-width lower bound, choose one representative from every clique $K_t^v$.
The induced subgraph on these representatives is $B_h$.
Linear rank-width is monotone under vertex deletion, by restricting a linear layout and deleting the matching rows and columns of each cut matrix.
Therefore
\[
\lrw(\Gamma_{h,t})\ge \lrw(B_h)=\Omega(h).
\]
Finally, for $h\ge1$, take any tree edge $\{u,v\}$ of $B_h$.
In $\Gamma_{h,t}$ the blown-up cliques $K_t^u$ and $K_t^v$ are complete and joined to one another by a complete bipartite graph, so $K_t^u\cup K_t^v$ induces $K_{2t}$.
Hence $\tw(\Gamma_{h,t})\ge\tw(K_{2t})=2t-1$.
For $h=0$, the single clique $K_t^v$ gives $\tw(\Gamma_{0,t})\ge t-1$.
\end{proof}

\section{Detailed comparison with related rank-width methods}\label{app:kk}

This appendix expands the comparisons sketched in the introduction: with the rank-width contraction algorithm of Kuyanov and Kissinger~\cite{KuyanovKissinger26}, the most closely related independent work; with Feynman-path decision diagrams~\cite{DBLP:journals/corr/abs-2510-06775}; and with the focused rank-width of Codsi and Laakkonen~\cite{CodsiLaakkonen26}.
Throughout, $R$ denotes the width of a supplied rank-decomposition, which plays the role of the parameter $k$ of the main text.

\paragraph{Boundary states and cut-signature tables.}

The table $A^{(a)}_u$ of \autoref{sec:fourier} is a coefficient representation of the function the subtree below $u$ induces on the rest of the graph.
Fix a vertex $u$ with outside $Y=\overline X_u$, and for $\by\in\{0,1\}^{Y}$ let
\[
R^{(a)}_u(\by)=\sum_{\bz\in\{0,1\}^{X_u}}\omega_r^{a\phi_u(\bz)}\,(-1)^{a\,\bz^{\top}\bA[X_u,Y]\by}
\]
be the residual boundary function.
The identity below relates the two.

\begin{lemma}[Boundary-state duality]\label{lem:boundary-duality}
Writing $\langle\cdot,\cdot\rangle$ for the $\F_2$ inner product on $\F_2^{Y}$,
\[
R^{(a)}_u(\by)=\sum_{\bsigma}A^{(a)}_u[\bsigma]\,(-1)^{a\langle\bsigma,\by\rangle}.
\]
\end{lemma}

\begin{proof}
The crossing sign depends on $\bz$ only through $\sigma_u(\bz)=\bz^{\top}\bA[X_u,Y]$, so grouping the sum by signature gives the display.
For odd $a$ the kernel $(-1)^{a\langle\cdot,\cdot\rangle}$ is the Walsh kernel.
\end{proof}

\noindent
For odd $a$ the two sides are therefore a Walsh--Hadamard transform pair on the row space of $\bA[X_u,Y]$ and its annihilator quotient, so they carry the same information in Fourier-dual bases.

\paragraph{The Kuyanov--Kissinger bounds.}
Kuyanov and Kissinger give a contraction algorithm for graph-like ZX diagrams parameterized by rank-width.
In their notation, their Proposition~4.1 gives a running time $\widetilde O(\operatorname{flops}(T))$ in terms of the decomposition tree $T$, their Corollary~4.1 relaxes this to the coarser bound $\widetilde O(4^R)$, and their Corollary~4.2 improves the bound to $\widetilde O(2^R)$ for linear decompositions.
Their statements use $\widetilde O(\cdot)$ without an explicit definition.
The proof of their Corollary~4.1 bounds $\operatorname{flops}(T)\le|V_G|\,2^{2R}$, where $V_G$ is the diagram's spider set, and absorbs that factor, so the tilde hides at least a factor linear in the diagram size.
We therefore write their bounds as $4^R\poly(|G|)$ and $2^R\poly(|G|)$, with $|G|$ the diagram size, which is the object their algorithm runs on as \autoref{tab:route-comparison} charges each route.
This is conservative for them.
Under the direct translation of \autoref{prop:four-guises} each SOP variable is a spider, and on the graph-like class of their Section~4 the reading goes back the other way, so $|G|$ plays the role of our $n$.

\paragraph{The correspondence.}
There is a close correspondence between graph-like ZX diagrams and quadratic SOPs in the sign-edge form of \autoref{eq:sop}.
In the graph-like setting treated in their Section~4, a closed diagram is, up to local Clifford data, a sum of powers over Boolean spider variables with quadratic phase interactions.
In the graph-like case each representation can be read off from the other.
The base of their branching bound and of their linear bound therefore mirrors the distinction we draw between branching rank-decompositions and linear layouts (\autoref{tab:route-comparison}).

\paragraph{Decomposition-sensitive costs.}
The rooted SOP DP on a decomposition tree $T$ costs $C_{\mathrm{SOP}}(T)\poly(n)$ per Fourier mode, where
\[
C_{\mathrm{SOP}}(T)=\sum_{v\ \mathrm{internal}}2^{\rho_{L(v)}+\rho_{R(v)}} .
\]
Their Proposition~4.1 charges, at each cubic vertex of the \emph{unrooted} tree with incident cut-ranks $\rho_1,\rho_2,\rho_3$, the term $2^{\min(\rho_1+\rho_2,\,\rho_1+\rho_3,\,\rho_2+\rho_3)}$.
Rooting a cubic tree selects, at every cubic vertex, one of its three incident pairs, and their expression takes the cheapest pair at every vertex simultaneously, which no single rooting need achieve.
Their decomposition-sensitive local exponent is therefore never larger than the pair a fixed rooting selects, and their fixed-decomposition bound can consequently be finer than $C_{\mathrm{SOP}}(T)$.

\paragraph{What differs.}
The algorithmic statements are nevertheless not the same result in different language.
Their routine contracts a ZX diagram by maintaining parity maps together with complex state vectors.
Ours evaluates an exact finite-modulus SOP by dynamic programming on the cut signatures of the SOP variable graph $G_C$, returns the full residue histogram when iterated over modes, and comes with the tensor-line-graph relation of \autoref{lem:sop-minor-line}.
By \autoref{lem:boundary-duality} the per-cut objects are directly related: the parity-map/state-vector pair stores a normalization of the residual boundary state $R^{(a)}_u$, while the SOP table stores its Walsh coefficients $A^{(a)}_u$.
The two objects therefore carry the same information in Fourier-dual bases, manipulated by different operations.
The two algorithms thus have the same coarse branching and linear rank-width dependence, but they differ in coefficient model, output, graph preprocessing, intermediate representation, and decomposition-sensitive cost.
Neither dominates the other asymptotically, and which carries the smaller polynomial factor is an empirical question.
\autoref{tab:kk-comparison} places the two algorithms side by side.

\begin{table}[htbp]
\centering \footnotesize
\renewcommand{\arraystretch}{1.12}
\caption{The present rank-width dynamic program and the Kuyanov--Kissinger contraction algorithm side by side.
Both are parameterized by the rank-width $R$ of a supplied decomposition.
Their bounds are written here as $\poly(|G|)$ rather than with the $\widetilde O(\cdot)$ of their statements, which absorbs a factor equal to the diagram's spider count.
Here $|G|$ is the diagram size, our $n$ under the correspondence described below.}
\label{tab:kk-comparison}
\begin{tabular}{@{}L{0.30\linewidth}L{0.32\linewidth}L{0.30\linewidth}@{}}
\toprule
\bf Method & \bf Object computed & \bf Leading dependence \\
\midrule
This paper, Fourier DP & one SOP Fourier mode, that is, one amplitude & $4^R\poly(n)$ branching, $2^R\poly(n)$ linear \\
This paper, residue histogram & all residues $N_j$ & $r\,4^R\poly(n)$ over Fourier modes \\
Kuyanov--Kissinger~\cite{KuyanovKissinger26} & graph-like ZX scalar contraction & $4^R\poly(|G|)$ branching, $2^R\poly(|G|)$ linear \\
\bottomrule
\end{tabular}
\end{table}

\paragraph{Boundary states and Feynman-path decision diagrams.}

Fix a variable order $x_1,\ldots,x_n$ and a prefix length $i$, and let $X$ be the first $i$ variables.
Two prefix assignments $\bz,\bz'\in\{0,1\}^{X}$ with the same cut signature $\bsigma_X(\bz)=\bsigma_X(\bz')$ induce, by \autoref{lem:boundary-duality}, the same boundary character $(-1)^{a\langle\bsigma,\cdot\rangle}$ on the remaining variables: their residual contributions differ only by the accumulated phase $\omega_r^{a\phi_X(\bz)}$.
At cut $i$ there are therefore at most $2^{\rho_i}$ residual characters, where $\rho_i$ is the cut-rank of the prefix.
The linear SOP DP of \autoref{thm:linear-layout-fourier} stores the \emph{aggregate coefficients} of these characters, namely the table $A^{(a)}_X$.
A Feynman-path decision diagram instead stores the corresponding \emph{residual functions} as diagram nodes, subject to its reduction and scalar-normalization rules.
This correspondence explains why both methods exhibit the same $2^{k}$ phenomenon along a variable order.
The correspondence is not a node-by-node identity with a reduced diagram: reduction rules may merge residual functions beyond signature equality, and normalization conventions differ.
The branching SOP DP has no single global variable order, which is where the two representations genuinely part ways (\autoref{cor:separating-family}).

\paragraph{Focused rank-width.}

Codsi and Laakkonen parameterize ZX contraction by decompositions that treat non-Clifford vertices specially~\cite{CodsiLaakkonen26}.
\autoref{cor:stabjoin} is the analogous guarantee for the sign-edge SOP subclass: when disjoint blocks of the decomposition each hold at most one $a$-magic vertex and together cover all of them, the stabilizer-rank optimization is fixed-parameter tractable in the cut-ranks of every active vertex.
Block interiors never enter the bound, since switching at their roots prunes them, so they may be refined by arbitrary leaf trees of unbounded cut-rank.
The two settings are not directly comparable: Codsi--Laakkonen operate on graph-like ZX diagrams and may simplify the diagram first, so neither the graph objects nor the practical pipelines need coincide.

\section{Closure of quadratic phase functions}\label{app:qpf-closure}

\lemqpfclosure*

\begin{proof}
For \ref{qpf-prod}, the data add: $\ell+\ell'$ coefficientwise modulo $4$, $q\oplus q'$, and $K\cap K'$.
For \ref{qpf-affine}, substitute and expand with \autoref{eq:xor-lift} and \autoref{eq:parity-lift}.
Linear parts of $q$ fold into $\ell$ (taking $\mu=2$), constants into $\theta$.
For \ref{qpf-sum1}, write $x=(x',x_w)$ and split off the last variable, $\ell(x)=\ell'(x')+c_wx_w$ and $q(x)=q'(x')+x_w\lambda(x')$, where $\lambda(x')=\bigoplus_{j<w}q_{jw}x'_j$ is the parity paired with $x_w$ by the quadratic form.
Row-reduce $\bm Mx=\bm d$ on the column of $x_w$, so that $x_w$ is either determined by the remaining variables or unconstrained by them, and the residual rows constrain $x'$ alone.
In the first case $x_w=\alpha(x')$ affinely and \ref{qpf-affine} eliminates it.
Otherwise $x_w$ is free and
\[
\sum_{x_w\in\{0,1\}}i^{c_wx_w}(-1)^{x_w\lambda(x')} =1+i^{c_w}(-1)^{\lambda(x')}
=\begin{cases}
2\,[\lambda(x')=0] & c_w=0,\\
(1+i)\,i^{3\lambda(x')} & c_w=1,\\
2\,[\lambda(x')=1] & c_w=2,\\
(1-i)\,i^{\lambda(x')} & c_w=3,
\end{cases}
\]
the middle cases using $\tfrac{1-i}{1+i}=i^{3}$ and $\tfrac{1+i}{1-i}=i$.
Each right-hand side is one further affine constraint or one factor of the form \autoref{eq:parity-lift}, so the partial sum is again a single QPF.
For \ref{qpf-push}, choose coordinates in which $\ker\Pi$ is spanned by the last $\dim\ker\Pi$ variables (an affine substitution) and apply \ref{qpf-sum1} repeatedly.
That leaves a QPF on a space of dimension $\rank\Pi$, which $\Pi$ identifies with $\operatorname{im}\Pi$.
When $\Pi$ is not surjective, $\operatorname{im}\Pi$ is itself cut out by linear equations on $\F_2^p$, and adding them to the support makes the function vanish off the image, as \ref{qpf-push} requires.
The total sum is the case $p=0$.
Every step is $\F_2$ linear algebra on $O(w^2)$ data, and \ref{qpf-push} then writes its result on $\F_2^p$, which is $O(p^2)$ further data.
\end{proof}

\section{More on the stabilizer-rank optimization}\label{app:stabjoin-proofs}

Fix a Fourier mode $a$ and a rooted rank-decomposition $\mathcal T$.
Fix an antichain $B$ as in \autoref{sec:stabjoin}, so that a vertex is active when it does not lie strictly below a member of $B$.
At a member of $B$ the computation starts from the magic price and immediately enters the ordinary table representation.
Every other active vertex uses the per-mode DP of \autoref{thm:fourier-speedup}.

\begin{algorithm}[tb!]
\caption{Stabilizer-rank optimization for one Fourier mode.}\label{alg:stabjoin}
\renewcommand{\algorithmicrequire}{\textbf{Input:}}
\renewcommand{\algorithmicensure}{\textbf{Output:}}
\begin{algorithmic}[1]
\Require Quadratic SOP as in \autoref{eq:sop}, a rooted rank-decomposition $\mathcal T$, a mode $a\in\Z_r$, and an antichain $B$.
\Ensure The mode sum $\widehat N(a)=\sum_{\bx}\omega_r^{af(\bx)}$.
\For{each active vertex $u$ of $\mathcal T$ in postorder}
  \If{$u\in B$}
    \State construct the QPF sum for the table at $u$ by \autoref{lem:magic-decomp}
    \State evaluate the sum on every cut signature to obtain $A_u^{(a)}$
  \ElsIf{$u$ is a leaf}
    \State initialize $A_u^{(a)}$ as in \autoref{thm:fourier-speedup}
  \Else
    \State compute $A_u^{(a)}$ from $A_{L_u}^{(a)}$ and $A_{R_u}^{(a)}$ by the per-mode join
  \EndIf
\EndFor
\State \Return $\omega_r^{ac}A_{r_{\mathcal T}}^{(a)}[\varnothing]$
\end{algorithmic}
\end{algorithm}

\begin{theorem}[Magic-bounded rank-width DP]\label{thm:hybrid-dp}
For every mode $a$ and antichain $B$, \autoref{alg:stabjoin} computes $\widehat N(a)$ exactly, in the total cost $T(B)$ of \autoref{eq:total-cost}.
\end{theorem}

\begin{proof}
Proceed over the active vertices in postorder.
At a switch vertex $b\in B$, \autoref{lem:magic-decomp} constructs a sum of at most $\mathrm{sr}(\tau_b)$ QPFs for the exact mode-$a$ table, and evaluating that sum on every cut signature produces $A_b^{(a)}$.
An active leaf outside $B$ receives the ordinary leaf table of \autoref{thm:fourier-speedup}.
At every other active internal vertex, the induction hypothesis supplies the two exact child tables, so the per-mode join produces $A_t^{(a)}$ by \autoref{lem:mode-from-counts}.
The value returned at the root is therefore $\widehat N(a)$.

Constructing the QPF sum at $b$ costs $\mathrm{sr}(\tau_b)\poly(n)$, and evaluating it on all $2^{\rho_b}$ signatures costs $\mathrm{sr}(\tau_b)2^{\rho_b}\poly(n)$ (\autoref{lem:qpf-closure}).
Every ordinary join scans at most $2^{\rho_{L_t}}2^{\rho_{R_t}}$ signature pairs, as in \autoref{thm:fourier-speedup}.
Leaf initialization and the final lookup are polynomial, nothing below $B$ is computed, and summing these charges gives the bound of \autoref{eq:total-cost}.
\end{proof}

Choosing $B=\varnothing$ gives the per-mode rank-width DP unchanged, while choosing only the root gives the global stabilizer-decomposition bound because $\rho_{r_{\mathcal T}}=0$.

\paragraph{The optimizer.}
The best total cost $\min_BT(B)$ is attained by the recurrence of \autoref{eq:switch-recurrence}, in which $C(u)$ is the minimum for the subtree rooted at $u$ once the common polynomial factor is removed.
The first branch switches at $u$ and prunes both child subtrees.
The second combines the optimal antichains of the two child subtrees and performs the ordinary join at $u$.
These are exactly the two possibilities for an antichain restricted to the subtree of $u$, so induction proves that the recurrence returns the exact minimum.
One postorder sweep stores $C(u)$ and one choice bit at each vertex, and backtracking returns a witnessing antichain.\lean{Formal/Stab/Paper.lean}{switch_optimum_spec}
The sweep takes $O(n)$ arithmetic operations and applies to any supplied constructive stabilizer-rank price, including $O(2^{0.3963\tau})$ for Clifford$+\T$.

\paragraph{Magic-focused decompositions.}
\autoref{cor:gottesman-knill} needs no antichain at all, since it reads the whole weight off the instance.
The remaining consequence switches at several incomparable vertices.

\begin{corollary}[Magic-focused width]\label{cor:stabjoin}
Fix a mode $a$.
Suppose disjoint subtrees (\emph{blocks}) each contain at most one $a$-magic vertex and together cover all of them, and take $B$ to be their roots.
If $k$ bounds the cut-rank at every active vertex, then \autoref{alg:stabjoin} runs in $O(4^{k}\poly(n))$ time.
\end{corollary}

\begin{proof}
Each block root $b$ has $\tau_b\le1$, so its switch costs at most $\mathrm{sr}(1)(1+2^{\rho_b})\poly(n)\le2(1+2^k)\poly(n)$.
Every other active internal vertex charges at most $2^{\rho_{L_t}+\rho_{R_t}}\poly(n)\le4^k\poly(n)$.
Nothing is charged inside a block.
There are linearly many active vertices, which gives the bound.
\end{proof}

\noindent
Block interiors are pruned, so their cut-ranks are unconstrained.

\paragraph{Mixed instances.}
\propmixedseparation*

\begin{proof}
Let $\Gamma_s$ consist of a clique $K_s$ and an $s\times s$ grid, vertex-disjoint, together with one edge $uv$ joining a clique vertex $u$ to a grid vertex $v$.
Let $C_s$ be the circuit over $\{\H,\CZ\}$ that \autoref{lem:realize-any-graph} produces for $\Gamma_s$, and let $D_s$ be $C_s$ with one $\T$ gate inserted on each of the $s$ wires carrying a clique vertex, as in \autoref{cor:separating-family}.
A $\T$ gate is diagonal, so it shifts a unary coefficient and adds no edge: the SOP variable graph is $G_{D_s}=\Gamma_s$ on $n_s=s+s^2$ variables, the $\T$-count is $\tau_s=s$, and a vertex is $a$-magic only if it carries one of those phases, so the clique vertices are the only ones.

Cut-ranks in $\Gamma_s$ are small on the clique side.
For $\varnothing\ne X\subsetneq V(K_s)$ every row of $\bA[X,\overline X]$ is the all-ones vector on $V(K_s)\setminus X$, extended by the single entry at $v$ for the row of $u$ if $u\in X$.
There are at most two distinct rows, so the cut-rank is at most $2$.
The cut that separates the grid from the clique is crossed by $uv$ alone and has rank $1$.

Take any rank-decomposition of the clique side, any rank-decomposition of the grid side, and join the two at the root.
Choose $B$ to contain only the root of the grid side.
Every ordinary join in the clique subtree costs at most $2^2\cdot2^2=16$ up to the polynomial factor.
The grid side is Clifford, so $\mathrm{sr}(0)=1$, and switching across its rank-$1$ boundary costs $1+2$ up to the same factor.
Nothing inside the grid is charged.
The root join costs $2^{1+1}=4$, every charge is constant, and only $O(s)$ vertices are active.
Thus \autoref{thm:hybrid-dp} gives a $\poly(n_s)$ bound, and the recurrence above returns an antichain no more expensive.

For the pure routes, rank-width does not increase on induced subgraphs~\cite{Oum05RankWidthVertexMinors,Oum17Survey}, and the $s\times s$ grid has rank-width $s-1$~\cite{Jelinek10Grid}, so $\rw(G_{D_s})\ge s-1$.
The stabilizer-decomposition route is charged on $\tau_s=s$.
\end{proof}

\section{Gadget constructions and WMC encodings}\label{app:gadget-wmc}

\paragraph{Gadget universality.}
The first two lemmas make explicit the finite Fourier steps used in the proof sketch of \autoref{thm:gadget-universality}.
We then give its full proof.

\begin{lemma}[Common root-list form]\label{lem:common-root-list}
Let $S\subseteq\Dcycl_r$ be finite and let $M=\operatorname{lcm}(r,8)$.
There are integers $h\ge0$ and $m\ge0$ such that every $\alpha\in S$ can be written as
\[
\alpha=2^{-h/2}\sum_{\bs\in\{0,1\}^m}\omega_M^{e_{\alpha,\bs}}
\]
for suitable exponents $e_{\alpha,\bs}\in\Z_M$.
\end{lemma}

For a nonempty set $T\subseteq[n]$ and $\bu\in\{0,1\}^n$, write
\[
p_T(\bu)=\bigoplus_{i\in T}u_i\in\{0,1\},
\]
where $\oplus$ denotes XOR of Boolean variables.
We absorb $T=\emptyset$ into the constant term.

\begin{proof}[Proof of \autoref{lem:common-root-list}]
If $S=\varnothing$, take $h=m=0$, so assume henceforth that $S$ is nonempty.
Each $\alpha\in\Dcycl_r$ is a finite integer combination of terms $\omega_r^b2^{-t/2}$.
Since $r\mid M$, we have $\omega_r^b=\omega_M^{bM/r}$.
Choose one integer $h$ at least as large as every exponent $t$ occurring in the finitely many representations of the elements of $S$, so that $h\ge t$ throughout.
Because $M$ is divisible by $8$, we have $\sqrt2=\omega_M^{M/8}+\omega_M^{-M/8}$, and writing $2^{-t/2}=2^{-h/2}2^{(h-t)/2}$ gives
\[
\omega_r^b2^{-t/2} =\omega_M^{bM/r}\,2^{-h/2}\,(\sqrt2)^{h-t}.
\]
The exponent $h-t$ is a nonnegative integer, which is exactly where $h\ge t$ is used, so $(\sqrt2)^{h-t}=(\omega_M^{M/8}+\omega_M^{-M/8})^{h-t}$ expands into a finite sum of $M$th roots with nonnegative integer coefficients.

We now clear all integer coefficients, so that $\alpha$ becomes a plain list of $M$th roots with no coefficients.
A positive coefficient repeats its root that many times.
A minus sign is absorbed by $-\omega_M^q=\omega_M^{q+M/2}$, valid because $M$ is even.
Thus every $\alpha\in S$ takes the form $2^{-h/2}\sum_{s=1}^{\ell_\alpha}\omega_M^{e_s}$ with the common $h$ but an $\alpha$-dependent list length $\ell_\alpha$.

It remains to bring all lists to one common length $2^m$.
Padding a list with the canceling pair $\omega_M^0+\omega_M^{M/2}=0$ preserves its value and raises its length by $2$, so it reaches any larger length of the same parity.
The $\ell_\alpha$ need not be even, so first increase $h$ by $2$, rewriting $2^{-h/2}\sum_s\omega_M^{e_s}=2^{-(h+2)/2}\sum_s2\omega_M^{e_s}$, which duplicates every summand and makes every length even.
Rename this larger value $h$ again.
Now choose $m\ge1$ with $2^m\ge\max_{\alpha\in S}\ell_\alpha$ and pad each list to length exactly $2^m$, indexing its entries by $\bs\in\{0,1\}^m$.
This gives $\alpha=2^{-h/2}\sum_{\bs\in\{0,1\}^m}\omega_M^{e_{\alpha,\bs}}$, as claimed.
\end{proof}

\begin{lemma}[Parity-linear phase tables]\label{lem:parity-linear-table}
Let $f:\{0,1\}^n\to\Z_M$, and let $R=M2^N$ with $N\ge n$.
Put $L=R/M$.
Then there are coefficients $c\in\Z_R$ and $c_T\in\Z_R$, for every nonempty $T\subseteq[n]$, such that
\[
L f(\bu)\equiv c+\sum_{\emptyset\ne T\subseteq[n]}c_Tp_T(\bu)\pmod R
\]
for all $\bu\in\{0,1\}^n$.
\end{lemma}

\begin{proof}[Proof of \autoref{lem:parity-linear-table}]
Fix for each $\bu$ an integer representative $\tilde f(\bu)\in\Z$ of $f(\bu)\in\Z_M$, and set $g=L\tilde f$.
Since $LM=R$, the class of $g(\bu)$ modulo $R$ does not depend on the choice of representatives.

For $T\subseteq[n]$, let
\[
\chi_T(\bu)=(-1)^{\sum_{i\in T}u_i}.
\]
The functions $\chi_T$, $T\subseteq[n]$, are orthonormal with respect to the inner product $\langle\varphi,\psi\rangle=2^{-n}\sum_{\bu\in\{0,1\}^n}\varphi(\bu)\psi(\bu)$.
Since there are $2^n$ of them, they form a basis for the space of real-valued functions on $\{0,1\}^n$~\cite{ODonnell14}.
Hence
\[
g(\bu)=\sum_{T\subseteq[n]}\widehat g(T)\chi_T(\bu), \qquad \widehat g(T)=2^{-n}\sum_{\bu} g(\bu)\chi_T(\bu).
\]
Because $g=2^N\tilde f$ with $N\ge n$, each Fourier coefficient $\widehat g(T)=2^{N-n}\sum_{\bu}\tilde f(\bu)\chi_T(\bu)$ is an integer.
For $T\ne\emptyset$ we have $\chi_T=1-2p_T$, so
\[
g(\bu) = \Bigl(\sum_T\widehat g(T)\Bigr) +\sum_{T\ne\emptyset}(-2\widehat g(T))\,p_T(\bu).
\]
Set $c\equiv\sum_T\widehat g(T)$ and $c_T\equiv-2\widehat g(T)$ modulo $R$.
The required identity follows because $g=L\tilde f\equiv Lf\pmod R$.
\end{proof}

\thmgadgetuniv*

\begin{proof}[Proof of \autoref{thm:gadget-universality}]
First suppose all entries lie in $\Z[\omega_r,1/\sqrt2]$.
Apply \autoref{lem:common-root-list} to the finite set of all entries of all gates, and call the resulting $m$ Boolean indices $\bs\in\{0,1\}^m$ of the common root list the \emph{selector bits}.
Let $a$ be the maximum arity of a gate in $\Gate$, put $M=\operatorname{lcm}(r,8)$, choose $N_{\Gate}=2a+m$, and set $R=M2^{N_{\Gate}}$ and $L=R/M$.

For a $k$-qubit gate $U$, applying \autoref{lem:common-root-list} entry by entry gives a phase table
\[
F_U:\{0,1\}^{2k+m}\to\Z_M, \qquad \langle \bz|U|\by\rangle=2^{-h/2}\sum_{\bs\in\{0,1\}^m}\omega_M^{F_U(\by,\bz,\bs)},
\]
with zero entries represented by canceling root lists.

Let $\bu=(\by,\bz,\bs)$ and $n=2k+m\le N_{\Gate}$, so that $u_i$ is the $i$th bit of $\bu$.
By \autoref{lem:parity-linear-table},
\[
L F_U(\bu)\equiv c+\sum_{\emptyset\ne T\subseteq[n]}c_T\,p_T(\bu)\pmod R.
\]
A singleton $T=\{i\}$ has $p_T(\bu)=u_i$, so its term $c_{\{i\}}u_i$ is already unary.
For every $T$ with $|T|\ge2$ introduce a fresh Boolean \emph{parity variable} $q_T$ and a Boolean \emph{averaging variable} $\lambda_T$, and set
\[
f_U(\bu,\bm{q},\blambda)= c+\sum_{i\in[n]}c_{\{i\}}\,u_i+\sum_{|T|\ge2}c_T\,q_T +\frac R2\sum_{|T|\ge2}\lambda_T\Bigl(q_T+\sum_{i\in T}u_i\Bigr) \pmod R,
\]
a sign-edge phase polynomial in $(\bu,\bm{q},\blambda)$.
Summing each $\lambda_T$ over $\{0,1\}$ applies \autoref{lem:parity-constraint-gadget} once per non-singleton $T$, which forces $q_T=p_T(\bu)$ and contributes a factor $2$.
With $b_U=\#\{T:|T|\ge2\}$,
\[
2^{-b_U}\sum_{\bm{q},\blambda}\omega_R^{f_U(\bu,\bm{q},\blambda)}=\omega_R^{L F_U(\bu)}=\omega_M^{F_U(\bu)}.
\]
Therefore
\[
\langle \bz|U|\by\rangle = 2^{-(h+2b_U)/2}\sum_{\bs,\bm{q},\blambda}\omega_R^{f_U(\by,\bz,\bs,\bm{q},\blambda)},
\]
so $U$ has a constant-size sign-edge SOP gadget, and its constants depend only on $\Gate$.

Composing the circuit identifies the output wire bits of each gate with the input wire bits of the next, and sums the product of the gate amplitudes over these shared internal wire bits.
Multiplying the amplitudes adds the gadget exponents over the common modulus $R$, and distinct gates use disjoint auxiliary variables.
The result is therefore a single quadratic SOP whose size is linear in the number of elementary gates, since each gate contributes a constant number of variables and terms.
Pinning the input and output bits substitutes constants for the boundary variables and turns each incident quadratic term into a unary or constant term.
Parallel sign edges on the same pair add modulo $R$, and either survive as one sign edge or cancel.

Conversely, suppose a gate entry is represented by a constant-size sign-edge SOP gadget over an even modulus $R$.
The gadget is a finite sum of powers of $\omega_R$, multiplied by a power of $2^{-1/2}$.
Therefore the entry lies in $\Z[\omega_R,1/\sqrt2]$.
Since the gate set is finite and the gadget modulus is common, all entries lie in one common dyadic-cyclotomic ring.
\end{proof}

\paragraph{Size of the fixed gadgets.}
All quantities introduced in the proof are constants depending only on $\Gate$.
Write $a$ for the maximum arity, $m$ for the number of selector bits, $M=\operatorname{lcm}(r,8)$, and $N_\Gate=2a+m$.
The construction then uses the fixed modulus $R=M2^{N_\Gate}$, and the number of auxiliary parity and averaging variables per gate is bounded by $2(2^{N_\Gate}-N_\Gate-1)$, plus the $m$ selector bits.
These are constants for a fixed gate set, but not necessarily small constants.
Thus the hidden constant in the linear-size translation, and the linear dependence on the modulus in the full-histogram algorithm (\autoref{sec:fourier}), depend only on the finite gate-set description.

\paragraph{Gadget-expanded treewidth.}
\corgadgettw*

\begin{proof}
Fix a tree decomposition of the line graph $L(N_C)$ of width $\cc(N_C)$.
For a gate tensor $g$ of the fixed set $\Gate$, let $I_g$ be its incident bond indices.
Here $|I_g|$ is bounded by the arity bound of $\Gate$.
Any two members of $I_g$ co-occur in the tensor $g$, so $I_g$ is a clique of $L(N_C)$.
Every clique of a graph is contained in some bag of every tree decomposition: the subtrees of bags containing the individual vertices pairwise intersect, so by the Helly property of subtrees of a tree they have a common bag.
Pick such a bag $B_g$ for every gate.

Build the \emph{intermediate gadget graph} $H$ as follows.
Start from $L(N_C)$ and, for every gate $g$ replaced by a gadget, add that gadget's $O(1)$ auxiliary variables, namely the selector, parity, and averaging variables of \autoref{thm:gadget-universality}.
Add as well all edges the gadget's phase polynomial induces among the auxiliaries and between auxiliaries and the boundary indices $I_g$.
Extend the tree decomposition by hanging, below $B_g$, one fresh bag $B_g'=I_g\cup\mathrm{aux}(g)$.
All new edges have both endpoints in $B_g'$, each new vertex occurs in exactly one bag, and $|B_g'|$ is a constant $c_{\Gate}$ depending only on the gate set.
Hence $\tw(H)\le\max\{\cc(N_C),\,c_{\Gate}-1\}$.

It remains to observe that the pinned gadget-expanded SOP graph $G_C^{\mathrm{gad}}$ is a minor of $H$.
The construction of the pinned SOP graph from the circuit representation performs only three kinds of operation, each of which preserves minors.
It identifies a Hadamard-free wire's two endpoint indices (an edge contraction in $H$, where the two indices are adjacent), it deletes pinned boundary variables and variables eliminated by simplification (vertex deletions), and it deletes canceled sign edges (edge deletions).
Treewidth is minor-monotone, so $\tw(G_C^{\mathrm{gad}})\le\tw(H)\le\max\{\cc(N_C),\,c_{\Gate}\}$.
\end{proof}

Note what is \emph{not} claimed.
The exact inequality $\tw(G_C)\le\cc(N_C)$ of \autoref{lem:sop-minor-line} concerns the native graph.
After gadgetization the guarantee degrades only to the gate-set constant, which leaves the $2^{O(\cc(N_C))}\poly(n)$ consequence intact.

\paragraph{The algebraic-WMC identities.}
We first verify that the algebraic weighted count of $F_C^{\mathrm{strict}}$ equals $S(f)$.
The clauses of $s_e\leftrightarrow(x_u\wedge x_v)$ make the assignment of every $s_e$ a function of $\bx$, so the satisfying assignments of $F_C^{\mathrm{strict}}$ are in bijection with $\{0,1\}^{V}$.
On the assignment extending $\bx$, the weight product is
\[
\prod_{v}w(x_v)\;\prod_{e}w(s_e) =\prod_{v:\,x_v=1}\omega_r^{b_v}\;\prod_{\{u,v\}\in E:\,x_ux_v=1}(-1) =\omega_r^{\sum_v b_vx_v}\,\omega_r^{\eta\sum_{\{u,v\}\in E}x_ux_v},
\]
using $\eta=r/2$ and $\omega_r^{\eta}=-1$.
Multiplying by the external scalar $\omega_r^{c}$ gives $\omega_r^{f(\bx)}$, and summing over the satisfying assignments gives $S(f)$.

For $F_C^{\mathrm{soft}}$, the two implications force $s_e=0$ unless $x_u=x_v=1$.
In the remaining case $s_e$ is free, and summing over it gives
\[
1+(\omega_r^\eta-1)=\omega_r^\eta.
\]
Thus each edge again contributes $\omega_r^{\eta x_ux_v}$, and the same unary weights and external scalar show that the weighted count is $S(f)$.
\qed

\lemvartoinc*

\begin{proof}
For the soft encoding, each edge $e=\{x,y\}$ gives the path
\[
x-C_{e,1}-s_e-C_{e,2}-y
\]
in the incidence graph, with no other vertices or edges in its gadget.
Hence $I(F_C^{\mathrm{soft}})$ is obtained from $G_C$ by subdividing every edge three times.
When $E\ne\varnothing$, subdivision preserves treewidth and $\tw(G_C)\ge1$, which proves the stated equality.
When $E=\varnothing$, both graphs are edgeless and have treewidth $0$.

For the strict encoding, take a tree decomposition of $G_C$.
For each edge $e=\{x,y\}$ choose a bag containing $x$ and $y$ and attach to it the fresh bag $\{x,y,s_e\}$, with three further child bags $\{x,s_e,C_{e,1}\}$, $\{y,s_e,C_{e,2}\}$, and $\{x,y,s_e,C_{e,3}\}$ for the clause vertices of the equivalence.
Every new incidence edge is covered inside this local subtree, every variable and clause vertex occupies a connected set of bags, and the largest new bag has size four.
\end{proof}

\paragraph{Rank-width of the incidence graphs.}
We first prove the claim of \autoref{sec:feynmandd} that the strict incidence graph of the $K_n$ instance has rank-width $\Omega(n)$, although $\rw(K_n)=1$.
First, the once-subdivided clique $S(K_n)$ occurs \emph{induced} in the strict incidence graph.
Keep the variable vertices $x_v$ and, for each edge $e=\{u,v\}$, the single clause vertex of $s_e\vee\neg x_u\vee\neg x_v$.
Each kept clause vertex then has degree two among the kept vertices, subdividing the edge $\{x_u,x_v\}$.
Rank-width is monotone under taking induced subgraphs, so it suffices to bound $\rw(S(K_n))$.

Pick the balanced edge of any rank-decomposition of $S(K_n)$ with respect to \emph{weights}, unit on the $n$ original vertices and zero on the subdivision vertices, so that both sides carry $\Omega(n)$ original vertices.
Balancing all vertices alike would not guarantee this.
Let $P,Q$ be the original vertices on the two sides, $|P|,|Q|\ge n/3$, and let $s_{pq}$ subdivide edge $pq$.
If at least $|P|/2$ vertices $p\in P$ have some $s_{pq}$ on the $Q$ side, choosing one for each gives distinct unit columns in the cut matrix.
Otherwise at least $|P|/2$ vertices $p$ have every $s_{pq}$ on the $P$ side, and the columns indexed by $Q$, restricted to these subdivision rows, are nonzero with pairwise disjoint supports.
Either way the cut-rank is $\Omega(n)$.

The soft incidence graph is the threefold subdivision of $K_n$.
Balance the original vertices in any of its rank-decompositions and let $P,Q$ be the original vertices on opposite sides, with $|P|,|Q|=\Omega(n)$.
Each of the $|P||Q|$ corresponding paths crosses the cut in at least one of its four positions, so one position is a crossing on $\Omega(n^2)$ paths.
At either internal position, the crossing edges contain a diagonal submatrix of that order because the internal vertices are private to their paths.
At an endpoint position, their columns or rows restrict to unit vectors on $\Omega(n)$ distinct original vertices.
In either case the cut-rank is $\Omega(n)$, so the soft incidence graph has rank-width $\Omega(n)$ as well.
\qed

The dense-block variant used on three experimental instances introduces parity chains over each side of a complete bipartite block instead of subdividing the edges of $G_C$.
Its width analysis is therefore separate from \autoref{lem:variable-to-incidence}.

\end{document}